\documentclass[reqno]{amsart}
\usepackage{amsmath,amsthm,amsfonts,amssymb,srcltx}
\usepackage{graphicx} 
\usepackage{epstopdf}
\usepackage{enumerate}
\usepackage{url}
\usepackage{bm}
\usepackage{epsfig}
\usepackage{color}
\usepackage{subcaption}
\usepackage{hyperref}
\usepackage{pdflscape}
\usepackage[all]{xy}
\usepackage[T1]{fontenc}
\usepackage{relsize}

\allowdisplaybreaks

\theoremstyle{plain}
 \newtheorem{thm}{Theorem}[section]
  \newtheorem{prop}[thm]{Proposition}

    \newtheorem{coro}[thm]{Corollary}
 \theoremstyle{definition}    
  \newtheorem{dfn}[thm]{Definition}
  \newtheorem{ass}[thm]{Assumption}
  \newtheorem{prob}[thm]{Problem}
  \newtheorem{exa}[thm]{Example}

\theoremstyle{remark}
  \newtheorem{rem}[thm]{Remark}

\numberwithin{equation}{section}\numberwithin{figure}{section}

\def\Hom{{\text{\rm{Hom}}}}

\def\rchi{{\hbox{\raise1.5pt\hbox{$\chi$}}}}

\def\isom{\cong}

\newcommand{\bea}{\begin{eqnarray}}
\newcommand{\eea}{\end{eqnarray}}
\newcommand{\be}{\begin{equation}}
\newcommand{\ee}{\end{equation}}

\newcommand{\Res}{\mathop{\rm Res}}

\title[Topological recursion and uncoupled BPS structures II]{
Topological recursion and uncoupled BPS structures II: Voros symbols and the $\tau$-function}

\author{Kohei Iwaki}%
\address{Graduate School of Mathematical Sciences, 
The University of Tokyo, 
3-8-1 Komaba, Meguro-ku, Tokyo, 153-8914, Japan}
\email{iwaki@ms.u-tokyo.ac.jp}

\author{Omar Kidwai}%
\address{Graduate School of Mathematical Sciences, 
The University of Tokyo, 
3-8-1 Komaba, Meguro-ku, Tokyo, 153-8914, Japan}
\email{kidwai@ms.u-tokyo.ac.jp}

\date{}

\begin{document}

\large
\setcounter{section}{0}
\maketitle
\begin{abstract}
   We continue our study of the correspondence between BPS structures and topological recursion in the uncoupled case, this time from the viewpoint of quantum curves. For spectral curves of hypergeometric type, we show the Borel-resummed Voros symbols of the corresponding quantum curves solve Bridgeland's ``BPS Riemann-Hilbert problem''. In particular, they satisfy the required jump property in agreement with the generalized definition of BPS indices $\Omega$ in our previous work. Furthermore, we observe the Voros coefficients define a closed one-form on the parameter space, and show that (log of) Bridgeland's $\tau$-function encoding the solution is none other than the corresponding potential, up to a constant. When the quantization parameter is set to a special value, this agrees with the Borel sum of the topological recursion partition function $Z_{\rm TR}$, up to a simple factor.
   
\end{abstract}

\tableofcontents

\section{Introduction}

{
This paper continues our study \cite{IK20} of the relationship between the theory of BPS structures and the formalism of topological recursion, in the special case where the former is \emph{uncoupled}. In the present work, we approach this based on the relationship that each of these two theories has with \emph{exact WKB analysis} --- in particular, the theory of \emph{quantum curves}. As before, our testing ground will be concrete examples related to the Gauss hypergeometric equation and its confluent degenerations, the \emph{spectral curves of hypergeometric type}, summarized in Table \ref{table:classical} below.

In our previous work, we gave a formula for the topological recursion free energies $F_g$ as a sum over BPS cycles, weighted by BPS indices $\Omega(\gamma)$ whose definition we introduced generalizing \cite{GMN09,BS13}. In the present paper we build on this result (which is formal in $\hbar$), following the topological recursion construction of the quantum curve in \cite{IKoT-I,IKoT-II} to solve Bridgeland's \emph{BPS Riemann-Hilbert problem} \cite{Bri19} (which is analytic in nature). Our solution is given by the Borel-resummed Voros symbols of the quantum curve, and we furthermore obtain the corresponding ``$\tau$-function'', together with its interpretation in terms of a potential for the Voros coefficients.

We will state our main result in detail below after giving a brief background of these two theories, together with the results of our previous paper \cite{IK20}.


\begin{table}[t]
\begin{center}
\begin{tabular}{ccc}\hline
$\underset{\rm (label)}{\rm Equation}$ & $Q_{\bullet}(x)$  & Assumption
\\\hline\hline
\parbox[c][4.0em][c]{0em}{}
$\underset{\rm (HG)}{\rm Gauss}$
&
\begin{minipage}{.35\textwidth}
\begin{center}
$\dfrac{{m_\infty}^2 x^2 
- ({m_\infty}^2 + {m_0}^2 - {m_1}^2)x 
+ {m_0}^2}{x^2 (x-1)^2}$
\end{center}
\end{minipage}
& ~~~\quad
\begin{minipage}{.35\textwidth}
\begin{center}
$m_0, m_1, m_\infty \neq 0$,\\
$m_0 \pm m_1 \pm m_\infty \ne 0$.
\end{center}
\end{minipage}
\\\hline
\parbox[c][3.0em][c]{0em}{}
$\underset{\rm (dHG)}{\rm Degenerate~Gauss}$
&
\begin{minipage}{.35\textwidth}
\begin{center}
$\dfrac{{m_\infty}^2 x + {m_1}^2 - {m_\infty}^2}{x(x-1)^2}$
\end{center}
\end{minipage}
&
\begin{minipage}{.35\textwidth}
\begin{center}
$m_1, m_\infty \neq 0$,\\
$m_1 \pm m_\infty \neq 0$.
\end{center}
\end{minipage}
\\\hline
\parbox[c][3.0em][c]{0em}{}
$\underset{\rm (Kum)}{\rm Kummer}$ 
&
\begin{minipage}{.3\textwidth}
\begin{center}
$\dfrac{x^2 + 4 m_\infty x + 4 {m_0}^2}{4x^2}$
\end{center}
\end{minipage}
&
\begin{minipage}{.3\textwidth}
\begin{center}
$m_0\neq 0$,\\
$m_0 \pm m_\infty \neq 0$.
\end{center}
\end{minipage}
\\\hline
\parbox[c][3.0em][c]{0em}{}
$\underset{\rm (Leg)}{\rm Legendre}$ 
& $\dfrac{m_\infty^2}{x^2-1}$
&
$m_\infty \neq 0$.
\\\hline
\parbox[c][3.0em][c]{0em}{}
$\underset{\rm (Bes)}{\rm Bessel}$ 
& $\dfrac{x + 4m_0^2}{4x^2}$
&
$m_0 \neq 0$.
\\\hline
\parbox[c][3.0em][c]{0em}{}
$\underset{\rm (Whi)}{\rm Whittaker}$ 
& $\dfrac{x - 4m_\infty}{4 x}$
& $m_\infty \neq 0$.
\\\hline
\parbox[c][3.0em][c]{0em}{}
$\underset{\rm (Web)}{\rm Weber}$ 
& $\dfrac{1}{4} x^2 - m_\infty$
& $m_\infty \neq 0$.
\\\hline
\parbox[c][3em][c]{0em}{}
$\underset{\rm (dBes)}{\rm Degenerate~Bessel}$ 
& $\dfrac{1}{x}$
& --
\\\hline
\parbox[c][2.5em][c]{0em}{}
$\underset{\rm (Ai)}{\rm Airy}$ 
& $x$
& --
\\\hline
\end{tabular}
\end{center}
\caption{Spectral curve $\Sigma_\bullet : y^2 - Q_{\bullet}(x) = 0$ 
of hypergeometric type, where 
$\bullet \in \{ {\rm HG}, {\rm dHG}, {\rm Kum}, 
{\rm Leg}, {\rm Bes}, {\rm Whi}, {\rm Web}, {\rm dBes},  {\rm Ai} \}$ 
labels the equation.} 
\label{table:classical}
\vspace{-7mm}
\end{table}
\subsection{BPS structures}

BPS structures, introduced by Bridgeland \cite{Bri19}, are an axiomatization of the structure of Donaldson-Thomas (DT) invariants of a Calabi-Yau 3 triangulated category with a stability condition. It consists of a triple $(\Gamma, Z, \Omega)$ of a ``charge lattice'' $\Gamma$ with a skew-symmetric pairing, a ``central charge'' homomorphism $Z : \Gamma \to {\mathbb C}$, and a collection of numbers (``BPS indices'') $\Omega : \Gamma \to {\mathbb Q}$, associated to each $\gamma\in\Gamma$, satisfying some conditions. This is a special case of the stability structure in \cite{KS08}. 

Given a BPS structure, Bridgeland also associated a Riemann-Hilbert problem on ${\mathbb C}^\ast = {\mathbb C}^\ast_{\hbar}$ with the coordinate $\hbar$, which we call the \emph{BPS Riemann-Hilbert problem}, seeking a collection of 
piecewise holomorphic (or meromorphic; see  \S \ref{subsection:Voros-vs-minimal-intro} below and \cite{Bar}) twisted-torus-valued functions $X_{\ell}(\hbar)$, labeled by a ray $\ell \subset {\mathbb C}^\ast$, with a collection of {\em BPS rays} on which the functions have discontinuity. More precisely, $X_\ell$ is defined on a half-plane ${\mathbb H}_\ell \subset {\mathbb C}^\ast$ whose bisecting ray is given by $\ell$, and when we vary $\ell$ so that it crosses over a BPS ray $\ell_{\rm BPS}$, then $X_{\ell}$ has a discontinuity described by a cluster-like birational transformation, called the \emph{BPS automorphism}, encoded by the BPS invariants $\Omega(\gamma)$ of $\gamma \in \Gamma$ satisfying $Z(\gamma) \in \ell_{\rm BPS}$. 
The BPS Riemann-Hilbert problem is closely related to the one studied by Gaiotto-Moore-Neitzke in \cite{GMN08, GMN09}, in which the BPS spectra for 
``class ${\mathcal S}$ theories''
were investigated. These Riemann-Hilbert problems allows us to understand the {\em Kontsevich-Soibelman wall-crossing formula}, a non trivial relation among the BPS invariants under variations of stability conditions, as a kind of iso-Stokes property around the ``irregular singular point'' $\hbar = 0$ (see also \cite{BTL, FFS}). This wall-crossing structure is axiomatized in the notion of {\em variation of BPS structures}.

If a variation of BPS structures over a complex manifold $M$ is nicely parametrized (``miniversal''), the solution to the associated BPS Riemann-Hilbert problem is naturally regarded as a collection of piecewise-holomorphic functions on $M \times {\mathbb C}^\ast$, and can be used to define a certain $\hbar$-dependent vector field on $M$. In this setting, Bridgeland also defined a function on $M \times {\mathbb C}^\ast$, which we call the \emph{BPS $\tau$-function} and denote by $\tau_{\rm BPS}$, by the condition that $\log\tau_{\rm BPS}$ be a Hamiltonian for this vector field (see \S\ref{sec:tau-def}).


If the BPS structure is uncoupled, it was shown in \cite{Bri19} that the solutions to the BPS Riemann-Hilbert problem can be solved explicitly in terms of gamma functions, and the associated BPS $\tau$-function is expressed in terms of a product of Barnes $G$-functions. Although non-trivial wall-crossing does not occur in the uncoupled situation, it was observed in \cite{Bri19, Bri20} that the BPS $\tau$-function has an alternative interesting interpretation. Namely,  $\tau_{\rm BPS}$ for the BPS Riemann-Hilbert problem associated with the natural variation of (uncoupled) BPS structures arising from resolved conifold through Donaldson-Thomas theory provides a ``non-perturbative partition function'' for the resolved conifold; that is, the asymptotic expansion of $\tau_{\rm BPS}$ when $\hbar \to 0$ coincides with the generating series of all-genus Gromov-Witten invariants. This surprising discovery was also generalized in \cite{Stoppa}.

The study of BPS Riemann-Hilbert problems associated to spectral curves of hypergeometric type 
through our previous work \cite{IK20} --- in particular, the construction of their solutions and $\tau$-functions via topological recursion --- will be the main theme of our present work. 
We will recall how the BPS structures arise from such spectral curves in \S \ref{subsection:review-part-1}. 

\subsection{Topological recursion and spectral curves of hypergeometric type}

Topological recursion (TR) is an algorithm generalizing the construction of correlation functions and partition functions in matrix models, whose input data is an algebraic curve enhanced with some additional data, called a \emph{spectral curve} (\cite{EO, CEO}). Mathematically, to a given spectral curve, TR constructs a doubly-indexed sequence $\{ W_{g,n} \}_{g \ge 0, \, n \ge 1}$ of meromorphic multi-differentials on the spectral curve, and a sequence $\{ F_g \}_{g \ge 0}$ of complex numbers.

These outputs of TR are expected to encode information of enumerative geometry such as Gromov-Witten invariants, Hurwitz numbers, etc. (essentially due to various dualities in mathematical physics), and there is a large class of spectral curves for which this expectation is known to hold true (see \cite{EO2, Eynard-11, DMNPS13, DN16} etc.). As represented by the celebrated Kontsevich-Witten theorem \cite{Kon, Wit}, generating series of such enumerative invariants are expected to be related to $\tau$-functions of some corresponding integrable systems. 
In fact, at the level of formal power series of $\hbar$, KdV and Painlev\'e $\tau$-functions have already been constructed via (the discrete Fourier transform of) the TR partition function $Z_{\rm TR} = \exp(\sum_{g \ge 0} \hbar^{2g-2} F_g )$ for certain spectral curves; see \cite{EO, IS, IM, IMS, MO18, I19, EG19, MO19} for example. The actual $\tau$-function should be given as the Borel sum of those formal series, and thus, in this paper, we denote by $\tau_{\rm TR}$ the Borel sum of the TR partition function. Although the Borel summability of TR partition functions is not proved in full generality (see \cite{Eynard-19} for the growth estimate of the coefficients), those which appear in this paper are Borel summable along all but finitely many rays. 

In \cite{IKoT-I, IKoT-II}, the first named author together with collaborators applied the TR formalism to the spectral curves of hypergeometric type, and obtained a result on {\em quantum curves} which states the following: after the principal specialization (i.e., setting $z_1 = \cdots = z_n = z$), a certain generating function of a primitive $F_{g,n}$ of $W_{g,n}$ (i.e., a function satisfying\footnote{We need a slight modification when $(g,n) = (0,2)$.} $d_{z_1} \cdots d_{z_n} F_{g,n} = W_{g,n}$) gives a WKB solution to a Schr\"odinger-type ODE, which we denote by ${\bm E}$, whose classical limit recovers the spectral curve. 
Indeed, we may add a tuple of parameters ${\bm \nu}$ which parametrizes the integration contour to define the primitive $F_{g,n}$, and hence, we obtain a family ${\bm E} = {\bm E}({\bm \nu})$ of quantum curves. We note that they also depend on the tuple ${\bm m} = (m_s)$ of parameters contained in the definition of $Q(x)$, which is one of $Q_\bullet (x)$ in Table \ref{table:classical}. We also observe that the resulting quantum curve is equivalent to an ODE satisfied by classical special functions; for example, the quantum curve of $\Sigma_{\rm HG}$ is equivalent to the Gauss hypergeometric differential equation. This allows us to study the TR invariants in terms of various properties of special functions. As one of the main results in \cite{IKoT-I, IKoT-II}, an explicit formula in terms of Bernoulli numbers was given for the $F_g$, as well as the {\em Voros coefficients} of the quantum curve ${\bm E}({\bm \nu})$, defined as a term-wise integral of the formal series given as the logarithmic derivative of the WKB solution along cycles on the spectral curve. It turns out that the Borel sum of generating series of free energy and Voros coefficients can be interpreted through the language of BPS structures.


\subsection{Brief review of results of Part I} 
\label{subsection:review-part-1}

To formulate the results of \cite{IK20}, let us first recall how BPS structures are constructed from spectral curves of hypergeometric type. There is a canonical way to construct (a variation of) BPS structures from such (a family of) spectral curves (see \cite[\S 7]{Bri19} for the general case, and see also \cite{GMN08} as the origin of this construction). For a given spectral curve $\Sigma = \Sigma_{\bullet}$ in Table \ref{table:classical}, we take the charge lattice and pairing to be $\Gamma = \{\gamma \in H_1(\widetilde{\Sigma}, {\mathbb Z}) ~|~ \iota_\ast \gamma = - \gamma \}$ and the intersection pairing on it, respectively. Here,  $\widetilde{\Sigma}$ is a partial compactification of ${\Sigma}$ (see \S \ref{section:hypergeometric-curvs-structure}), and $\iota$ is the covering involution $\iota : \widetilde{\Sigma} \to \widetilde{\Sigma}$. The central charge is given by the period map $Z(\gamma) = \oint_{\gamma} y dx = \oint_{\gamma} \sqrt{Q(x)} \, dx$, and the BPS indices $\Omega(\gamma)$ were defined by a certain weighted counting of degenerations of spectral networks (Stokes graphs) for the meromorphic quadratic differential $Q(x) \,dx^2$.  We recall these details more precisely in \S \ref{subsection:saddle-and-BPS-indices}.

The list of BPS structures which we obtained in \cite{IK20} following this construction is summarized concisely in Table \ref{table:BPS-str-HG}.
In all our examples, since the curves $\Sigma$ in Table \ref{table:classical} are of genus $0$, $\Gamma$ has the trivial intersection pairing so that the resulting BPS structure is \emph{a fortiori} uncoupled. 
The rank of $\Gamma$ is equal to the number of mass parameters ${\bm m} = (m_{s})$ attached to even order poles $s$ of $Q(x)$. If we vary ${\bm m}$, we have a variation of BPS structure parametrized by ${\bm m} \in M$, where $M$ is the space of mass parameters satisfying the assumptions in Table \ref{table:classical}. We can describe the central charge $Z(\gamma)$ explicitly as a function of ${\bm m}$, and it is easy to check that the variation is miniversal.  

We recall that our definition of $\Omega(\gamma)$ generalizes the one given in \cite{GMN08, BS13, Bri19} in order to include degenerations involving a saddle connection with endpoint a simple pole of $Q(x) \,dx^2$, and degenerate ring domains (see \cite[\S 3.6]{IK20}). With this generalization, we have shown in \cite{IK20} that the TR free energy $F_g$ for the spectral curve $\Sigma$ is described as a sum over all BPS cycles whose central charge lie in a generically chosen half-plane 
(\cite[Theorem 5.3]{IK20}). 
The following formula for TR free energy with $g \ge 2$ is one of the main results in \cite{IK20}:  
\begin{equation} \label{eq:main-formula-from-part-I}
F_g = \frac{B_{2g}}{2g(2g-2)} \sum_{\substack{ \gamma \in \Gamma \\ Z(\gamma) \in {\mathbb H}}} 
\Omega(\gamma) \, \left( \frac{2\pi i}{Z(\gamma)} \right)^{2g-2}.  
\end{equation}
Here, ${\mathbb H}$ is any half-plane whose boundary rays are not BPS. This universal formula is valid for all our examples, and shows that each term in the free energy is a contribution of a BPS cycle associated with a degenerate spectral network; furthermore the coefficient of each term knows the BPS index of the BPS cycle. We conjectured that the formula \eqref{eq:main-formula-from-part-I} is still valid for higher degree spectral curves if the associated BPS structure is uncoupled, and checked this expectation numerically for a few examples of degree 3 spectral curves (\cite[\S 5.3]{IK20}). 

However, in our previous work, we did not give any explanation for the ``naturality'' of our generalized definition of $\Omega(\gamma)$; namely, we defined them so that the universal formula \eqref{eq:main-formula-from-part-I} is valid. Therefore, in the present paper, we will try to fill this gap\footnote{In the recent paper \cite{Haiden}, Haiden discussed the Donaldson-Thomas theory of certain CY3 categories associated with a quadratic differential with simple poles. We see that the BPS indices defined in \cite[\S 3.6]{IK20} indeed agree with the ones obtained in \cite{Haiden}.} and show the validity of our generalization of $\Omega(\gamma)$ by showing they control the analytic behaviour of (Borel-resummed) Voros coefficients. We will make this more precise below.


\subsection{Results of Part II}

The BPS indices $\Omega(\gamma)$ are defined by weighted counting of degenerate spectral networks, but they have another interpretation in Gaiotto-Moore-Neitzke \cite{GMN08}. That is, the BPS indices should appear in the exponents in the formula which describe the mutation of the Fock-Goncharov cluster coordinates (on the moduli space of flat ${\rm SL}_2({\mathbb C})$-connections) caused by the degeneration of the spectral network. 
On the other hand, as is shown in \cite{I16, All18, Kuw20}, the Fock-Goncharov cluster coordinates can be identified (on the oper locus, where this makes sense) with the Borel-resummed Voros symbols (exponential of Voros coefficients) in the exact WKB analysis of a certain Schr\"odinger-type equation with the spectral curve as its classical limit (see also \cite{IN14, IN15}), where the mutation formula (BPS automorphism) is a consequence of the Stokes phenomenon for the Voros symbols. 

Keeping these previous works in mind, to investigate the relation between BPS structures and TR, we study the Stokes structures of the Voros symbols for the quantum curves constructed through TR from the spectral curves of hypergeometric type obtained in \cite{IKoT-I, IKoT-II}, and compare those to the BPS Riemann-Hilbert problems associated with the BPS structure studied in our previous work \cite{IK20}. 
We note here that such analysis of the Stokes phenomenon for the Voros symbols of the (confluent) hypergeometric differential equations are not new; some special cases (e.g., for specific ${\bm \nu}$) have already been studied in the exact WKB literature  (see \cite{Aoki-Tanda, ATT, AIT, Takei08, KoT11, KKKT11, KKT14} for example). We also note that Stokes structures of the Voros symbols for Schr\"odinger-type equations were also studied in \cite{DDP93, IN14} in a general setting.  

Let us summarize our results more precisely below.


\subsubsection{Voros solution to the (almost-doubled) BPS Riemann-Hilbert problem}
\label{subsub:Voros-solution-intro}

A natural BPS Riemann-Hilbert problem, which is compared to the Stokes structure on the TR side, is closely related to the one  associated with the {\em doubled} BPS structure which has a charge lattice $\Gamma_{\rm D} = \Gamma \oplus \Gamma^{\vee}$. In what follows, we identify a relative homology class in $H_1(\overline{\Sigma}, D_\infty, {\mathbb Z})$, where $D_\infty$ is the divisor supported on the poles of $Q(x) dx^2$ of order $\ge 2$, with an element in $\Gamma^\vee$ through the Poincar\'e-Lefschetz duality.

For the quantum curves ${\bm E} = {\bm E}({{\bm \nu}})$ constructed in \cite{IKoT-I, IKoT-II}, first we show that the Voros coefficients for relative homology classes (``path Voros coefficients'') have an explicit expression as a sum over all BPS cycles (which is analogous to the formula \eqref{eq:main-formula-from-part-I} for the TR free energy): 
for a given $\beta \in \Gamma^\vee$ satisfying $\iota_\ast \beta = - \beta$, we have
\[
V_{\beta} =  \sum_{k \ge 1} \hbar^{k} V_{\beta, k},
\]
with the coefficients being given in terms of the Bernoulli polynomials: 
\begin{align}
\label{eq:introvoros}
V_{\beta, k}
& = 
\sum_{\substack{\gamma \in \Gamma \\ Z(\gamma) \in {\mathbb H} \\ \Omega(\gamma) \ne -1}} 
\Omega(\gamma) \, \langle \beta, \gamma \rangle \, 
\frac{B_{k+1}\bigl( \frac{1+\nu(\gamma)}{2} \bigr)}{k(k+1)} \, 
\left( \frac{2 \pi i}{Z(\gamma)} \right)^k  \notag \\[-.5em]
& \quad + \sum_{\substack{\gamma \in \Gamma \\ Z(\gamma) \in {\mathbb H} \\ \Omega(\gamma) = -1}} 
\Omega(\gamma) \, \frac{\langle \beta, \gamma \rangle}{2}  \, 
\frac{B_{k+1}\bigl(\frac{\nu(\gamma)}{2} \bigr)+ B_{k+1}\bigl( 1 + \frac{\nu(\gamma)}{2} \bigr)}{k(k+1)} \, 
\left( \frac{2 \pi i}{Z(\gamma)} \right)^k
\end{align}
(see \S\ref{subsubsection:Voros-coefficient} for the precise notation). We also have an explicit expression of Voros coefficients $V_{\gamma}$ for homology classes $\gamma \in \Gamma$ (``cycle Voros coefficients'') in terms of the parameters ${\bm m}$ and ${\bm \nu}$.
Together, we find that 
\begin{itemize}
\item 
the Voros symbols $e^{V_{\gamma}}$, $e^{V_{\beta}}$ are Borel summable as a formal series of $\hbar$ along any ray $\ell$ which is not the BPS rays for the corresponding BPS structure given in \cite{IK20}, and
\item 
the Borel-resummed Voros symbols jumps when $\ell$ crosses a BPS ray due to the Stokes phenomenon, and the action of {\em Stokes automorphism} on the Voros symbols are given in the exactly same manner as the BPS automorphism (see Proposition \ref{prop:jump-of-Voros-symbols-in-BPS-form}) with the BPS indices $\Omega(\gamma)$ defined in \cite{IK20}. This justifies our definition of $\Omega(\gamma)$. 

\end{itemize}

The second observation tells us that the Borel-resummed Voros symbols of the quantum curve gives a certain solution to the BPS Riemann-Hilbert problem. To describe the underlying BPS structure, let us introduce the {\em almost-doubled} BPS structure $(\Gamma_{\text{\rm \DJ}},Z_{\rm \text{\rm \DJ}},\Omega_{\text{\rm \DJ}})$, which has the charge lattice $\Gamma_{\text{\DJ}} = \Gamma \oplus \Gamma^{\ast}$ where 
\begin{equation}
\Gamma^{\ast} = \{ \beta \in \Gamma^{\vee} ~ | ~ \iota_\ast \beta = - \beta \}, 
\end{equation} 
with trivially extended central charge $Z_{\text{\DJ}}$ and BPS indices $\Omega_{\text{\DJ}}$ (we note that $(\Gamma_{\text{\rm \DJ}},Z_{\rm \text{\rm \DJ}},\Omega_{\text{\rm \DJ}})$ is also uncoupled).  
The {\em constant term} $\xi$ in the small $\hbar$ asymptotics of the solution is specified by the aforementioned parameter ${\bm \nu}$ together with a specific {\em quadratic refinement} (which introduces an appropriate sign to give a twisted-torus-valued function).


  
   \begin{thm}[{Theorem \ref{thm:voros-soln}}]
 \label{thm:voros-soln-intro} 
 Let $(\Gamma, Z, \Omega)$ denote a BPS structure obtained from any spectral curve of hypergeometric type. The Borel sum of cycle and path Voros symbols of the corresponding quantum curve ${\bm E}({\bm \nu})$ provide a meromorphic solution to the BPS Riemann-Hilbert problem associated to the almost-doubled BPS structure $(\Gamma_{\text{\rm \DJ}},Z_{\rm \text{\rm \DJ}},\Omega_{\text{\rm \DJ}})$, and with constant term $\xi_{{\text{\rm \DJ}},\bm \nu}$ given explicitly in \eqref{eq:nuxi}. 
  \end{thm}

The reason for considering the sublattice $\Gamma^{*}$ is essentially to avoid a square root type singularity in the Borel-resummed Voros symbols (i.e., so we can obtain a {\em meromorphic} solution to the BPS Riemann-Hilbert problem), although the reader can take the full lattice if they are happy to accept such a singularity. 

We call the above solution the {\em Voros solution} to the almost-doubled BPS Riemann-Hilbert problem, and denote it by $X^{\rm Vor}_{\ell}$.  We give its explicit expression in \eqref{eq:vorosexplicit}.


\subsubsection{Relation between the Voros solution and other solutions}
\label{subsection:Voros-vs-minimal-intro}

Similar problems have been considered in \cite{All19,Bri19,Bar}\footnote{
The spectral curves of hypergeometric type, the main examples considered in this paper, are excluded as {\em non-amenable} cases 
in \cite{All19}.
}. 
For the special choice of $\xi$ where $\xi(\gamma)=1$ if $\Omega(\gamma)\neq 0$, Bridgeland gave the explicit general piecewise-\emph{holomorphic} solution, which we will denote by $X^{\rm hol}_\ell$ 
for all uncoupled BPS structures in \cite{Bri19}. This was generalized by \cite{Bar} who noted that, even in the uncoupled case, holomorphic solutions do not usually exist for other values of $\xi$, but solved the problem \emph{meromorphically} in this case for arbitrary $\xi$. 
\cite{All19} provides a (meromorphic) solution for the general (coupled) set up of a quadratic differential with essentially arbitrary pole configuration, but at a specific value of $\xi$, and in terms of Fock-Goncharov coordinates. Finally, \cite{Bri19} solves the problem for arbitrary $\xi$ in the simplest nontrivial coupled case of the cubic oscillator, also using Fock-Goncharov coordinates.

Among these, our work is closely related to the problem considered by Barbieri \cite{Bar}. Like her work, we provide a solution for arbitrary $\xi$, at the cost of considering only uncoupled examples. We note that our solution is not the so-called {\em minimal solution} $X^{\rm min}_{\ell}$ in the sense of \cite{Bar}, which is reduced to the holomorphic solution in the original case $\xi(\gamma) = 1$ considered by Bridgeland \cite{Bri19}, but in \S\ref{sec:relationtominimal} we show that the minimal solution and Voros solution differ only by a simple factor and provide an explicit relation between the two. We will make this more precise in \S \ref{sec:relationtominimal}.

\subsubsection{Voros potential, BPS $\tau$-function and TR partition function function}
\label{subsub:Voros-potential-intro}

{

We observe that Voros coefficients of the quantum curve ${\bm E}( {\bm \nu})$ naturally define a formal power series valued one-form on the space $M$ of parameters ${\bm m}$ which is \emph{closed} (in fact exact), and arises from a potential function $\phi$; 
in particular, we have a formal power series valued function on $M$ satisfying  \begin{equation}
d_M \phi = \sum_{s} V_{\beta_s} dm_s, 
\end{equation}
where $\{ \beta_s \}$  is a natural basis of $\Gamma^\ast$. We call $\phi$ the \emph{Voros potential}. The following theorem tell us that the Voros potential (more precisely, exponential of its $\hbar$-derivative) is a topological recursion / quantum curve analogue of the BPS $\tau$-function.   

\begin{thm}[Theorem \ref{thm:taupotential}]
Let $(\Gamma_{\text{\rm \DJ}},Z_{\text{\rm \DJ}},\Omega_{\text{\rm \DJ}})$ be the almost-doubled variation of BPS structures corresponding to one of the families of spectral curves of hypergeometric type, and $\ell$ be any non-BPS ray. 
The BPS $\tau$-function $\tau^{\rm Vor}_{{\rm BPS}, \ell}$ associated to the Voros solution $X^{\rm Vor}_{\ell}$ is given by the Borel sum of the $\hbar$-derivative of the Voros potential as: 
\begin{equation}
\tau_{{\rm BPS}, \ell}^{\rm Vor} = c_{\ell} \, \mathcal{S}_{\ell}e^{- \partial_\hbar \phi},
\end{equation}
where $c_\ell$ is an explicit constant factor (see Theorem \ref{thm:taupotential}).
\end{thm}

We can also express $\tau_{\rm BPS, \ell}^{\rm Vor}$ as a product of modified Barnes $G$-functions, which is closely related to the BPS $\tau$-function obtained by Alexandrov-Pioline \cite{AP}. 

In all examples, we furthermore find there exists a choice (not unique) of the parameter ${\bm \nu}={\bm \nu}_*$ so that we have $\xi_{{\text{\rm \DJ}},{\bm \nu}_*}(\gamma_{\rm BPS}) = 1$ for all BPS cycles $\gamma_{\rm BPS}$. If we specialize the parameter to this value, we have 


\begin{coro}
\label{coro:maintheorem-1}
Let $(\Gamma_{\text{\rm \DJ}},Z_{\text{\rm \DJ}},\Omega_{\text{\rm \DJ}})$ be the almost-doubled variation of BPS structures corresponding to one of the families of spectral curves of hypergeometric type, and $\ell$ be any non-BPS ray. If we specialize the parameter value ${\bm \nu}={\bm \nu}_*$, the BPS $\tau$-function $\tau_{\rm BPS, \ell}^{\rm Vor}$ is reduced to the BPS $\tau$-function $\tau_{\rm BPS, \ell}^{\rm hol}$ associated with the holomorphic solution $X_{\ell}^{\rm hol}$. 
More precisely, we have
\begin{equation}
\tau_{{\rm BPS}, \ell}^{{\rm Vor}} \,\Big|_{{\bm \nu} = {\bm \nu}_*} 
=\varkappa \cdot \tau_{\rm BPS, \ell}^{\rm hol}, 
\end{equation}
where $\varkappa$ is either 1 or a simple explicit function independent of $\ell$. 
\end{coro}

This BPS $\tau$-function $\tau_{\rm BPS, \ell}^{\rm hol}$ was originally obtained by Bridgeland in \cite{Bri19}, where its explicit expression in terms of the Barnes $G$-function was given. 
Through a well-known formula for the asymptotic expansion of the Barnes $G$-function and the formula \eqref{eq:main-formula-from-part-I} of our previous work \cite{IK20}, we can relate $\tau_{\rm BPS}^{\rm hol}$ to the Borel-resummed topological recursion partiton function $\tau_{\rm TR}$. More precisely, we have:

\begin{thm}[Theorem \ref{thm:tauhol}]
Let $(\Gamma_{\text{\rm \DJ}},Z_{\text{\rm \DJ}},\Omega_{\text{\rm \DJ}})$ be the almost-doubled variation of BPS structures corresponding to one of the families of spectral curves of hypergeometric type. 
Then Bridgeland's BPS $\tau$-function $\tau_{\rm BPS}^{\rm hol}$ and the Borel-resummed topological recursion partition function $\tau_{\rm TR}$ are related as:
\begin{equation}
\tau_{{\rm BPS}, \ell}^{{\rm hol}} = {\tau}^{\mathsmaller{>0}}_{{\rm TR}, \ell}.
\end{equation}
Here ${\tau}^{\mathsmaller{>0}}_{{\rm TR}, \ell}$ 
denotes the Borel-resummed TR partition function without the genus $0$ part. 
\end{thm}

}

These results establish a precise link between objects arising from the topological recursion (quantum curves and their Voros coefficients/potential), and objects arising from the BPS structures (solution to the BPS RHP, $\tau$-function, etc.) coming from the same initial data; i.e., the spectral curves of hypergeometric type. We summarize our results as a dictionary between uncoupled BPS structures and topological recursion, depicted in Table \ref{table:bpstr} below, which holds precisely in all our examples (we believe these examples to be exhaustive of finite uncoupled cases when the spectral curve is of degree 2). 

Although the TR is applicable to higher degree spectral curves (\cite{BHLMR-12, BE-12, BE-16}), the existence of the associated BPS structure remains conjectural: we expect that Gaiotto-Moore-Neitzke's approach in \cite{GMN12} allows us to define a natural BPS structure for higher degree spectral curves. We will briefly verify our picture continues to hold in this case based on numerical experiments in our previous work, following computations by Y.\,M.\,Takei \cite{YM-Takei20}.

\begin{table}[h]
\begin{center}  \begin{tabular}{|c|c|} \hline
   \parbox[c][1.75em][c]{0em}{}  
   \textbf{BPS structure}
    & \textbf{TR / quantum curve} 
        \\ \hline 
        
    \parbox[c][1.75em][c]{0em}{}  
   BPS index $\Omega(\gamma)$ &
    coefficients in $F_g$ (alien derivative)
    \\ \hline
     \parbox[c][1.75em][c]{0em}{}     constant term $\xi$ & quantization parameter ${\bm \nu}$
    \\ \hline
    \parbox[c][1.75em][c]{0em}{} 
    solution to the BPS RHP $X_\ell$
    & Borel-resummed Voros symbol $\mathcal{S}_\ell e^{V}$ of quantum curve 
    \\ \hline 
    \parbox[c][1.75em][c]{0em}{} 
    BPS $\tau$-function $\tau^{\rm Vor}_{\rm BPS,{\ell}}$
    & Borel-resummed Voros potential $\mathcal{S}_{\ell}\exp{\left(-{\partial_\hbar\phi}\right)}$
     \\ \hline 
    \parbox[c][1.75em][c]{0em}{} 
    $\tau_{{\rm BPS}, \ell}^{{\rm hol}}$  ($\xi(\gamma_{\rm BPS})\equiv1)$ 
    & Borel-resummed TR partition function  
    $\tau_{\rm TR,{\ell}}= {\mathcal S}_{\ell}[{Z}_{\rm TR}]$ 
     \\ \hline
  \end{tabular} 
   \vspace{+1.em}
     \caption{A dictionary between uncoupled BPS structures and topological recursion.}
     \label{table:bpstr}
     \vspace{-7mm}
\end{center}     
\end{table}

\subsection{Remarks}
\begin{enumerate}
\item Based on our result, we would like to emphasize the following point. For the special parameter value ${\bm \nu}= {\bm \nu}_\ast$ we have shown that the Borel sum of the ($g>0$) TR partition function is equal to $\tau^{\rm hol}_{\rm BPS}$, so that in particular $\tau_{\rm TR }$ encodes the solution to the BPS Riemann-Hilbert problem. That is, $\tau_{\rm TR}$ is characterized by its discontinuity structure and asymptotics, so that one may approach the partition function of TR from the perspective of Riemann-Hilbert-type methods. While the correspondence is literal for all examples we have considered, we hope that this perspective may be useful to deepen our understanding of the topological recursion in more general cases.

\item We also note a similar structure to our result appears in the solution to a different (related) problem, but in a certain limit of quadratic differentials with four second-order poles (``nodal limit of the four-punctured sphere''). In particular, we note the resemblance to the expression for the $1$-loop part of the twisted effective superpotential\footnote{$\widetilde{W}^{\rm eff}$ is itself identified with $\lim\limits_{\epsilon_2\to 0} \epsilon_2 \log{Z_{\rm Nek}}$, the \emph{Nekrasov-Shatashvili free energy}.} $\widetilde{W}^{\rm eff}$ in the conjecture of Nekrasov-Rosly-Shatashvili \cite{NRS11,TV15,HK,JN18}, identified with the generating function of ($\hbar$-)opers. There, a hypergeometric spectral curve also appears, essentially as the four-punctured sphere breaks up into two three-punctured spheres.  However, that calculation corresponds to the $\epsilon_1 \rightarrow 0$ limit, whereas the approach of this paper corresponds to the $\epsilon_2=-\epsilon_1=\hbar$ limit. It is somewhat puzzling that these two regimes produce a similar structure.

\item While $\tau_{\rm BPS}$ differs slightly from our $\tau_{\rm TR}$, the expression for $\tau_{\rm TR}$ in the Weber case (a single BPS state) coincides exactly with the special function $\gamma_\hbar$ appearing in the work of Nekrasov-Okounkov \cite{NO} as the perturbative part of the Nekrasov partition function. While the difference may be attributed to possible variants of the definition of $\tau_{\rm BPS}$ and choices of renormalization schemes defining the partition function, it would be desirable to understand this relationship better.

\item Finally, we mention that our results can also be interpreted in the context of so-called \emph{Joyce structures} \cite{Bri19-2}. This is a recently proposed geometric structure arising naturally on the space of stability conditions from the theory of wall-crossing and Donaldson-Thomas invariants (see also \cite{BSt20, AP, AP2, AST21} for the associated hyperk\"ahler structure). Our solution to the BPS Riemann-Hilbert problem provides a canonical example of one such geometry, though we leave its detailed computation and further investigation to future work.
\end{enumerate}

\subsection{Organization} 
The paper is structured as follows.
In \S \ref{sec:BPS} we recall the notion of BPS structures and 
the define the associated BPS Riemann-Hilbert problem and $\tau$-function. In \S \ref{sec:BPS-from-HG-curves}, we focus on the BPS structures arising from spectral curves of hypergeometric type, recall our previous results, and prepare for the proof of the main results. In \S \ref{sec:quantumcurves} we recall the basics of topological recursion, and the relevant results on quantum curves and 
associated the Voros coefficients, describe them in terms of BPS structures and express their Stokes phenomena.
Finally, in \S \ref{sec:final}, we assemble the pieces and solve the BPS Riemann-Hilbert problem, relate $\tau_{\rm BPS}$ to the TR side, and remark on evidence in higher rank examples. 

\subsection{Acknowledgements} 
We thank 
Dylan Allegretti, Anna Barbieri, Tom Bridgeland, Lotte Hollands, Kento Osuga, and Yumiko Takei
for helpful conversations. 
O.K.'s work was supported by 
a JSPS Postdoctoral Fellowship for Foreign Researchers in Japan. 
This work was also supported by JSPS KAKENHI Grant Numbers 
16K17613, 16H06337, 17H06127, 19F19738 and 20K14323.

}

\section{Review of BPS structures and the BPS Riemann-Hilbert problem}
\label{sec:BPS}

Here we recall the basic facts on BPS structures, including the results of our previous work \cite{IK20}, define the BPS Riemann-Hilbert problem \cite{Bri19}, and give some preparatory lemmas about the variations of BPS structures we will work with. 

\subsection{Definition of BPS structure}

\begin{dfn}[{\cite[Definition 2.1]{Bri19}}]
 A \emph{BPS structure} is a tuple $(\Gamma,Z,\Omega)$ of the following data:
\begin{itemize}
\item a free abelian group of finite rank $\Gamma$ equipped with an antisymmetric pairing 
\begin{equation}
\langle \cdot , \cdot \rangle : \Gamma \times \Gamma \rightarrow \mathbb{Z},
\end{equation} 
\item a homomorphism of abelian groups $Z \colon \Gamma \rightarrow \mathbb{C}$, and
\item a map $\Omega : \Gamma \to {\mathbb Q}$,
\end{itemize}
satisfying the conditions
\begin{itemize}
\item \emph{Symmetry}: $\Omega(\gamma)=\Omega(-\gamma)$ for all $\gamma \in \Gamma$.
\item \emph{Support property}: for some (equivalently, any) choice of norm $\|\cdot\|$ on 
$\Gamma \otimes \mathbb{R}$, there is some $C>0$ such that
\begin{equation} \label{eq:support-property}
\Omega(\gamma)\neq  0 \implies |Z(\gamma)|>C\cdot \|\gamma\|.
\end{equation}
\end{itemize}
We call $\Gamma$ the \emph{charge lattice}, and the homomorphism $Z$ is called the \emph{central charge}. 
The rational numbers $\Omega(\gamma)$ are called the \emph{BPS indices} or \emph{BPS invariants}.
\end{dfn}

Let us recall some useful terminology \cite{Bri19}: 
\begin{itemize} 
\item
An \emph{active class} or \emph{BPS state} is an element $\gamma \in \Gamma$ which has 
nonzero BPS index, $\Omega(\gamma)\neq 0$.
\item 
A \emph{BPS ray} is a subset of ${\mathbb C}^\ast$  
which can be written as $\ell_\gamma = Z(\gamma) \, {\mathbb R}_{> 0}$ for some active class $\gamma$.


\end{itemize}
\noindent We note that the central charge $Z(\gamma)$ for an active class $\gamma$ is nonzero 
due to the support property. 

We can also consider certain classes of BPS structures with especially nice properties. A BPS structure $(\Gamma, Z, \Omega)$ is said to be
\begin{itemize} 
\item 
\emph{finite} if there are only finitely many active classes,
\item
\emph{uncoupled} if $\langle \gamma_{1},\gamma_{2}\rangle=0$ holds whenever 
$\gamma_{1},\gamma_{2}$ are both active classes
\item
\emph{integral} if the BPS invariant $\Omega$ takes values in ${\mathbb Z}$.
\end{itemize}

Though we will stick to such examples in this paper, in order to formulate the BPS Riemann-Hilbert problem in the next section, we may weaken these conditions. We call a BPS structure \emph{ray-finite} if there are finitely many active classes $\gamma$ with $Z(\gamma)\in \ell$ for a given BPS ray $\ell$, and \emph{generic} if $\langle \gamma_1,\gamma_2\rangle=0$ whenever $\gamma_1, \gamma_2$ are active and $Z(\gamma_1)$ and $Z(\gamma_2)$ lie on the same BPS ray. 

If the BPS structure is not uncoupled, we call it \emph{coupled}. This is in fact the more typical situation, but the analysis of the BPS structure and the corresponding Riemann-Hilbert problem given in the next section is much more difficult in 
this case (see \cite{Bri19,All19}). All our considerations in this paper will be for 
finite, uncoupled and integral BPS structures.

\subsection{Twisted torus and BPS automorphism}

For a lattice $\Gamma$ with an antisymmetric pairing $\langle \cdot, \cdot \rangle$, 
we define the corresponding \emph{twisted torus} $\mathbb{T}_-$ as the set of twisted homomorphisms into $\mathbb{C}^\ast$:
\begin{align}
\mathbb{T}_- := \left\{\xi : \Gamma \rightarrow \mathbb{C}^\ast  \; | \; 
\xi{(\gamma_1+\gamma_2)}= (-1)^{\langle\gamma_1,\gamma_2\rangle}\xi(\gamma_1)\xi(\gamma_2)\right\}
\end{align}

We also have the usual torus $\mathbb{T}_+:= \Hom(\Gamma,\mathbb{C}^*)$. We denote by $x_{\gamma} : {\mathbb T}_{-} \to {\mathbb C}^\ast$
the \emph{twisted character} naturally associated to $\gamma \in \Gamma$,
defined by ${x}_\gamma(\xi) = {\xi}(\gamma)$, which generate the coordinate ring of ${\mathbb T}_{-}$. Similarly, the ordinary \emph{characters} of the usual torus are denoted $y_\gamma$. The twisted and usual characters satisfy
\begin{equation}
x_\gamma x_{\gamma'} = (-1)^{\langle \gamma,\gamma' \rangle}x_{\gamma+\gamma'}, \qquad y_\gamma y_{\gamma'} = y_{\gamma+\gamma'}.
\end{equation}

Since $\mathbb{T}_+$ acts on $\mathbb{T}_-$ in the obvious way, a point $\xi\in \mathbb{T}_-$ can be understood as a point in $\mathbb{T}_+$ together with a ``base point'' $\sigma \in \mathbb{T}_-$. There is a particularly natural set of possible base points, intertwining usual inversion with the map $x_\gamma \mapsto x_{-\gamma}$, those taking values in $\mathbb{Z}_2$. These are the so-called \emph{quadratic refinements} $\sigma:\Gamma\rightarrow\{\pm1\}$. 

{Since $\mathbb{T}_+$ carries a Poisson structure given by 
\begin{equation} \label{eq:ordinary-poiss}
     \{y_\gamma,y_{\gamma'}\}=\langle \gamma,\gamma'\rangle \cdot y_{\gamma} y_{\gamma'}.
\end{equation}
We may transport this to $\mathbb{T}_-$ using any such base point, giving it a Poisson structure too (independent of the choice).
}

The twisted torus ${\mathbb T}_-$ will {essentially} be the space in which the solution to 
an associated Riemann-Hilbert problem takes values. The Riemann-Hilbert jump contour will 
be a finite union of rays in ${\mathbb C}^\ast$ whose coordinate is denoted by $\hbar$.
Here, a ray is a subset of the form 
$\ell = e^{i \vartheta}{\mathbb R}_{>0}\subset {\mathbb C}^\ast$.
Given a ray $\ell$, we associate a \emph{BPS automorphism} ${\mathbb S}(\ell)$ 
as the jump factor in the Riemann-Hilbert problem. 
In \cite{Bri19}, the BPS automorphisms were defined as a time 1 Hamiltonian flow of 
a certain function which contains the information of BPS/DT invariants.  
Note that these operators were also considered in \cite{KS08, GMN08}.

In general, defining the BPS automorphism needs a careful analysis of convergence problems, 
and one may need to work with an appropriate completion of the twisted torus.
However, thanks to \cite[Proposition 4.2]{Bri19}, 
the BPS automorphisms for finite, uncoupled, integral BPS structures  
are given by the explicit birational automorphisms
${\mathbb S}(\ell) : {\mathbb T}_- \dashrightarrow {\mathbb T}_-$:
\begin{equation} \label{eq:BPS-auto}
{\mathbb S}(\ell)^\ast ({x}_\beta) := {x}_\beta \, 
\prod_{\substack{\gamma \in \Gamma \\ Z(\gamma) \in \ell}} 
(1 - {x}_\gamma)^{\Omega(\gamma) \langle \gamma, \beta \rangle} 
\quad (\beta \in \Gamma).
\end{equation}
Here we define ${\mathbb S}(\ell)$ via its pullback ${\mathbb S}(\ell)^\ast$ 
acting on the coordinate ring of ${\mathbb T}_{-}$, also called the {\em Kontsevich-Soibelman transform} in \cite{GMN09}.
Note that the operator acts trivially on twisted characters 
$x_\gamma$ corresponding to ``null'' elements $\gamma \in \Gamma$ with zero intersection $\langle\gamma,\gamma' \rangle$ for all $\gamma' \in \Gamma$.
We will take \eqref{eq:BPS-auto} as the definition of the BPS automorphism, 
and consider the associated Riemann-Hilbert problem below.

\subsection{BPS Riemann-Hilbert problem}
\label{section:BPS-RHP}

Here we formulate the Riemann-Hilbert problem associated with a ray-finite, generic, integral BPS structure, 
following \cite{Bri19}. We will call it the \emph{BPS Riemann-Hilbert problem} for short. 
Roughly speaking, in the BPS Riemann-Hilbert problem, 
we seek a twisted-torus-valued, piecewise-meromorphic map on ${\mathbb C}^\ast$
which jumps when crossing a BPS ray, with the jump given precisely by the BPS automorphism.  
Note that the BPS Riemann-Hilbert problem was defined for more general BPS structures 
(i.e., for ``convergent'' BPS structure) in \cite{Bri19}; see \cite[Appendix A, B]{Bri19} 
for technical details in defining the BPS automorphisms in general setting.

Given a ray $\ell =  e^{i \vartheta}\,{\mathbb R}_{>0} \subset {\mathbb C}^\ast$ for some 
$\vartheta \in {\mathbb R}$, we denote by
\begin{equation}
{\mathbb H}_\ell := \{ \hbar \in {\mathbb C}^\ast ~|~ |\arg \hbar - \vartheta| < \frac{\pi}{2} \}
\end{equation}
the half plane centred on $\ell$. 
Then, the BPS Riemann-Hilbert problem is formulated as follows {(without loss of generality we may consider $\mathbb{H}_\ell$ as the domain)}:

\begin{prob}[{\cite{Bri19,Bar}}]
\label{prob:holRHP}
Let $(\Gamma, Z, \Omega)$ be a ray-finite, generic, integral BPS structure.  
Fix an element $\xi \in \mathbb{T}_-$ (``the constant term''). 
For every non-BPS ray $\ell \subset \mathbb{C}^\ast$, find meromorphic functions
\begin{equation}
 X_{\ell,\gamma} \colon \mathbb{H}_\ell \rightarrow \mathbb{C}
\end{equation}
for every $\gamma\in\Gamma$, with the following properties:
\begin{itemize}
    \item[(RH1)] \emph{Jumping.} 
    Suppose that two non-BPS rays $\ell_1, \ell_2 \subset {\mathbb C}^\ast$ form the boundary rays 
    of an acute sector $\Delta \subset {\mathbb C}^\ast$ taken in the clockwise order. 
    Then, 
    \begin{equation} \label{eq:RHP-jump}
        X_{\ell_2}(\hbar)={\mathbb S}(\Delta) \, X_{\ell_1}(\hbar)
    \end{equation}
    whenever $\hbar \in \mathbb{H}_{\ell_1}\cap \mathbb{H}_{\ell_2}$ and $0< |\hbar|\ll 1$. 
    Here ${\mathbb S}(\Delta)$ is the birational automorphism of ${\mathbb T}_-$ given by 
    the product of all BPS automorphisms ${\mathbb S}_\ell$ with $\ell \subset \Delta$. More explicitly, 
    \begin{equation} \label{eq:RHP-jump-explicit}
    X_{\ell_2, \gamma}(\hbar) = 
    X_{\ell_1, \gamma}(\hbar) \,
    \prod_{\substack{\gamma' \in \Gamma \\ Z(\gamma') \in \Delta}} 
    (1 - X_{\ell_1, \gamma'}(\hbar))^{\Omega(\gamma') \langle \gamma', \gamma \rangle},
    \qquad (\gamma \in \Gamma).
    \end{equation}

    \item[(RH2)] \emph{Asymptotics at 0.} 
    For any $\gamma \in \Gamma$, as $\hbar\rightarrow 0$ in $\mathbb{H}_\ell$, we have 
    \begin{equation}
 \, X_{\ell,\gamma}(\hbar) \sim  e^{-Z(\gamma)/\hbar} \xi(\gamma)
    \end{equation}
    whenever $\ell$ is a non-BPS ray.
    
    \item[(RH3)] \emph{Polynomial growth at $\infty$.} 
   For every $\gamma \in \Gamma$ and non-BPS ray $\ell$, there exists some $k$
    \begin{equation}
        |\hbar|^{-k}<|X_{\ell,\gamma}(\hbar)|<|\hbar|^k
    \end{equation}
    whenever $|\hbar|\gg 0$ in $\mathbb{H}_\ell$.
\end{itemize}
{We denote the collection of all $X_{\ell,\gamma}$ by $X_\ell$, which should then be understood as meromorphic maps 
\begin{equation}
    X_{\ell}:\mathbb{H}_{\ell}\setminus{\Xi_\ell} \rightarrow \mathbb{T}_{-},
\end{equation}
where $\Xi_\ell$ denotes the set of zeroes of all $X_{\ell,\gamma}$, for $\gamma\in\Gamma$.}
\end{prob}


\begin{rem}
Note that the BPS Riemann-Hilbert problem is closely related
to the infinite-dimensional Riemann-Hilbert problem encountered by Gaiotto-Moore-Neitzke \cite{GMN09} in constructing 
special Darboux coordinate functions on the moduli space of Higgs bundles (see also \cite{BTL, FFS}).
The solution is connected to exact WKB analysis in \cite{Bri19, All19}, 
and is expected to furthermore satisfy a kind of 
thermodynamic Bethe ansatz equation 
\cite{Gaiotto14, GMN09, Neitzke17}. 
Many interesting properties remain to be understood, even conjecturally. 
Solutions to the BPS Riemann-Hilbert problem should furthermore provide 
a map between cluster varieties and spaces of stability structures, 
proved in certain cases by \cite{All19}. Apart from the uncoupled cases and the works \cite{All19,Bri19}, very little is currently known about the general solution.
\end{rem}

For any finite, uncoupled, and integral BPS structure, Bridgeland \cite{Bri19} gave the explicit and unique piecewise-\emph{holomorphic} solution to this problem for a specific choice of constant term ($\xi(\gamma) = 1$ for every active $\gamma\in\Gamma$) in terms of the gamma function. Barbieri observed in \cite{Bar} that for other values of $\xi$, there will usually be no holomorphic solution, but generalized the solution of \cite{Bri19} as follows. To write it down, we use a variant of the gamma function,
\begin{equation} \label{eq:Lambda-function-original}
\Lambda(w,\eta) = \frac{e^{w} \, \Gamma(w+\eta)}{\sqrt{2\pi} \, w^{w + \eta - \frac{1}{2}}},
\end{equation}
and will later sometimes write $\Lambda(w)=\Lambda(w,0) = \Lambda(w,1)$. 
It is a multivalued function due to the factor $w^{w+\eta+1/2} = \exp((w+\eta+1/2) \log w)$, and we will restrict the domain of the function to ${\rm Re} \, w > 0$ and take the principal branch of $\log w$. The asymptotic expansion of $\Lambda(w, \eta)$ is given explicitly (see e.g. \cite{Bar}) as 
\begin{equation}
    \log \Lambda(w,\eta) \sim  \sum_{k \ge 1} \frac{B_{k+1}(1 - \eta)}{k(k+1)} w^{-k},
\end{equation}
valid when $w \to \infty$ in the complement of any closed sector containing ${\mathbb R}_{< 0}$.

\begin{thm}[{\cite{Bar,Bri19}}] \label{thm:Anna-RHP-sol}
Let $(\Gamma, Z, \Omega)$ be a finite, uncoupled, integral BPS structure, and fix $\xi \in \mathbb{T}_-$.
 Then the associated BPS Riemann-Hilbert problem has a solution, which associates to a non-BPS ray $\ell$ the map $X^{\rm min}_{\ell}$ 
whose value on $\gamma \in \Gamma$ is given by\footnote{We have corrected a sign error appearing in \cite{Bar}.} 
\begin{equation} \label{eq:Anna-solution}
X^{\rm min}_{\ell, \gamma}(\hbar) = e^{ -{Z(\gamma)}/{\hbar}} \cdot \xi(\gamma) \cdot 
\prod_{\substack{\gamma' \in \Gamma \\ Z(\gamma') \in i {\mathbb H}_{\ell} }} 
{\displaystyle\Lambda}\left(\frac{Z(\gamma')}{2 \pi i \hbar},\frac{\log\xi(-\gamma')}{2\pi i} \right)^{\Omega(\gamma') \langle \gamma, \gamma' \rangle}.
\end{equation}
Here, we fix the branch of the logarithm appearing in the formula so that ${\rm Im}\log{\xi}(\gamma) \in [0,2\pi)$ for $\gamma \in \Gamma$. 
When $\xi$ is chosen so that $\xi(\gamma)=1$ for all active $\gamma\in \Gamma$, this becomes the unique solution which is holomorphic on the half planes ${\mathbb H}_\ell$, denoted 
\begin{equation} \label{eq:X-hol}
X^{\rm hol}_{\ell, \gamma}(\hbar) = 
e^{ -{Z(\gamma)}/{\hbar}} \cdot \xi(\gamma) \cdot 
\prod_{\substack{\gamma' \in \Gamma \\ Z(\gamma') \in i {\mathbb H}_{\ell} }} 
{\displaystyle\Lambda}\left(\frac{Z(\gamma')}{2 \pi i \hbar} \right)^{\Omega(\gamma') \langle \gamma, \gamma' \rangle}.
\end{equation}
\end{thm}

\noindent Note that the solution has the structure of a product over BPS states, since the exponent will vanish whenever $\Omega$ does. The main observations of \cite{Bar} are
\begin{itemize}
    \item There is an explicit solution $X^{\min}$ given by \eqref{eq:Anna-solution}.
    \item The solution is not unique, but any two solutions are related by multiplication by a meromorphic function $f(\hbar)$ of $\hbar$ satisfying $f(0)=1$.
    \item The solution $X^{\rm min}$ is distinguished by the property that all poles and zeroes must be simple, may lie only at $\frac{Z(\gamma)}{\log{\xi}+2\pi i k}$, with $\gamma \in \Gamma$ and $k \in \mathbb{Z}$, and for each $\gamma \in \Gamma$ the points  $\frac{Z(\gamma)}{\log{\xi}+2\pi i k}$, $k\in\mathbb{Z}$ lie on the same side of $\ell_\gamma$.
\end{itemize}

In our case, we will solve (almost) the same Riemann-Hilbert problem using Voros symbols, but our solution will differ from $X^{\rm min}$. We will use the expression \eqref{eq:Anna-solution} to relate the two precisely in \S\ref{sec:relationtominimal}.

\subsection{Variation of BPS structures} 

BPS structures become even more interesting when considered in families. In this case, we require additional axioms ensuring the behaviour is the same as the Donaldson-Thomas invariants, in particular the wall-crossing phenomenon.
  
\begin{dfn}[{\cite[Definition 3.3]{Bri19}}] 
A \emph{variation of BPS structure} over a complex manifold $M$ 
consists of a collection of BPS structures 
$(\Gamma_{\bm m}, Z_{\bm m}, \Omega_{\bm m})$ attached to each ${\bm m} \in M$ 
satisfying the following axioms:

\begin{enumerate}[(VBPS1)]
\item The family $\Gamma_{\bm m}$ of charge lattices forms a local system over $M$, and the pairing 
$\langle \cdot ,\cdot \rangle$ is covariantly constant.
 
\item Given a covariantly constant family $\gamma_{\bm m} \in \Gamma_{\bm m}$, 
the central charges $Z_{\bm m}(\gamma_{\bm m})$ are holomorphic functions on $M$.
\item The constant $C$ appearing in the support property \eqref{eq:support-property} 
may be chosen uniformly on compact subsets of $M$
\item The \emph{wall-crossing formula} holds: for each acute sector 
$\Delta\subset \mathbb{C}^{*}$, the counterclockwise product
\begin{equation}
{\mathbb S}_{\bm m}(\Delta) := \prod_{\ell \subset \Delta} 
{{\mathbb S}_{\bm m}(\ell)} 
\end{equation} 
is covariantly constant as ${\bm m}$ varies, whenever the boundary rays of $\Delta$ are non-BPS throughout.
\end{enumerate}
\end{dfn}

The precise meaning of condition (VBPS4) for general variation of BPS structure 
is given in \cite[Appendix A]{Bri19}. 
For us, since we will only consider finite, uncoupled, integral BPS structures, 
the operator ${\mathbb S}_{\bm m}(\Delta)$ can be realized as a birational automorphism
of ${\mathbb T}_{-,{\bm m}}$, so we do not have to consider any completion. 

In all of our examples, the wall-crossing formula will be automatically satisfied as a result of the computation of \cite{IK20} -- in particular, the BPS indices do not vary with ${\bm m}$ at all, so the identity is trivially satisfied.

\subsection{BPS $\tau$-function}
\label{sec:tau-def}
Given a sufficiently nice variation of BPS structures, and a family of solutions to the corresponding BPS Riemann-Hilbert problems, Bridgeland \cite{Bri19} defined a notion of {\em $\tau$-function} encoding the solution in a single piecewise meromorphic function. To recall this, let us first introduce some terminology.

A variation of BPS structures $(\Gamma_{\bm m}, Z_{\bm m}, \Omega_{\bm m})$ 
over a complex manifold $M$ is said to be {\em framed} 
if the local system $(\Gamma_{\bm m})_{\bm m \in M}$ is trivial. For a framed variation of BPS structure, we can identify all the lattices $\Gamma_{\bm m}$ 
with a fixed lattice $\Gamma$ through the parallel transport. 
Using the identification, we will denote a framed variation of BPS structures by $(\Gamma, Z_{\bm m}, \Omega_{\bm m})$. 

\begin{dfn}
A framed variation of BPS structures $(\Gamma, Z_{\bm m}, \Omega_{\bm m})$ 
over $M$ is said to be {\em miniversal} if the ``period map''
\begin{equation}
\label{eq:periodmap}
\begin{array}{ccc}
M & \stackrel{}{\longrightarrow}&  \Hom(\Gamma, {\mathbb C}) \\
{\bm m} & \longmapsto& Z_{\bm m}
\end{array}
\end{equation} 
is a local isomorphism. 

\end{dfn}

Now, suppose we are given a framed, miniversal variation of BPS structure 
$(\Gamma, Z_{\bm m}, \Omega_{\bm m})$ over $M$. The miniversality property ensures we have the isomorphism 
\begin{equation}\label{eq:miniversal}
    D\pi: T_{\bm m}M \rightarrow {\rm Hom}(\Gamma,\mathbb{C}) = \Gamma^\vee \otimes \mathbb{C}
\end{equation}
and {dually we may identify $T_{\bm m}^*M$ with $\Gamma  \otimes\mathbb{C}$.
Then, the pairing $\langle \cdot, \cdot \rangle$ on $\Gamma$ induces 
a Poisson structure on $M$. If we take a ($\mathbb{C}$-)basis $\gamma_1, \dots, \gamma_r$ of $\Gamma$ and use $(z_1,\dots, z_r) = (Z_{\bm m}(\gamma_1), \dots, Z_{\bm m}(\gamma_r))$ as a coordinate on $M$, then the Poisson bracket is explicitly given as 
\begin{equation} \label{eq:possion-M}
    \{ z_i, z_j \} = \langle \gamma_i, \gamma_j \rangle 
    \qquad (i,j = 1,\dotsm r).
\end{equation}
}


For a family of solutions $\{ X_{\ell}({\bm m}, \hbar) \}_{{\bm m} \in M}$ to the BPS Riemann-Hilbert problem, we denote by $Y_\ell({\bm m}, \hbar) \in {\mathbb T}_{+}$ the ``non-twisted part''
\begin{equation}
Y_\ell({\bm m}, \hbar) := e^{Z_{\bm m}/\hbar} \, \xi^{-1} \, X_\ell({\bm m}, \hbar).
\end{equation}
{This is a collection of piecewise meromorphic functions (one for each $\gamma \in \Gamma)$ 
\begin{equation}
Y_{\ell,\gamma} : M \times {\mathbb H}_\ell \to {\mathbb{C}}
\end{equation} }
with discontinuities at ${\bm m} \in M$ such that the ray $\ell$ is a BPS ray for the corresponding BPS structure $(\Gamma_{\bm m}, Z_{\bm m}, \Omega_{\bm m})$. 
Thanks to the isomorphism \eqref{eq:miniversal}, for any non-BPS ray $\ell$, 
\begin{equation}
{\partial_\hbar} \log Y_{\ell}: M \times \mathbb{H}_{\ell} \rightarrow {\rm Hom}(\Gamma,\mathbb{C})
\end{equation}
defines a vector field on $M$ with discontinuities at this locus. Then, we can associate the $\tau$-function corresponding to the family as follows.
\begin{dfn}[{\cite[\S 4.6]{Bri19}}] 
Suppose $X({\bm m},\hbar)$ is a family of solutions to the BPS Riemann-Hilbert problem associated to a framed miniversal variation of BPS structures over $M$.
A \emph{BPS $\tau$-function} 
for the family of solutions $X({\bm m},\hbar)$, associated to a {non-BPS} ray $\ell$, 
{\begin{equation}
    \tau_{\ell}:M\times \mathbb{H}_\ell \rightarrow \mathbb{C}
\end{equation}}
\noindent is a {possibly-multivalued} meromorphic function defined so that the Hamiltonian vector field of the function
$2\pi i \, \log \tau_{{\rm BPS}, \ell}$ coincides with the vector field ${\partial_\hbar} \log Y_{\ell}$, 
where the Poisson structure on $M$ is given by \eqref{eq:possion-M}.

\end{dfn}

In terms of the coordinates above (for a chosen basis $\gamma_1, \dots, \gamma_r \in \Gamma$), we may write the condition for being a $\tau$-function as
\begin{equation} \label{eq:defining-tau}
    \frac{1}{2 \pi i} \, \dfrac{\partial \log Y_{\ell, \gamma_j}}{\partial \hbar} 
    =   \sum_{i=1}^{{r}} \langle \gamma_i, \gamma_j \rangle 
    \dfrac{\partial \log{\tau_{{\rm BPS},\ell}}}{\partial z_i}.
\end{equation}

In general, it is unknown if there exists a $\tau$-function, and uniqueness is not guaranteed. However, if we additionally impose that
$\tau_{{\rm BPS}, \ell}$ is invariant under simultaneous rescaling of the coordinates $z_i$ and $\hbar$, and
if the antisymmetric form $\langle \cdot, \cdot \rangle$ is nondegenerate, 
then the $\tau$-function is uniquely specified up to a multiplication 
by an element in ${\mathbb C}^\ast$.


Recall that Barbieri \cite{Bar} provided a canonical piecewise-meromorphic solution to the BPS Riemann-Hilbert problem for arbitrary uncoupled BPS structures. 
The $\tau$-function was not computed in that work, but considered in a recent paper \cite{AP} by Alexandrov-Pioline.  The formula is given as follows (see also \cite{BBS}).
\begin{prop}[{cf. \cite[\S 4]{AP}}] \label{prop:annatau}
Let $(\Gamma,Z,\Omega)$ be any miniversal variation of finite, uncoupled, integral BPS structures.
A BPS $\tau$-function for the family of minimal solutions to the BPS Riemann-Hilbert problem is given by
\begin{equation}
\label{eq:mintau}
    \tau^{\rm min}_{\rm BPS,{\ell}}({\bm m},\xi, \hbar)= \prod_{\substack{\gamma \in \Gamma \\ Z(\gamma) \in i {\mathbb H}_{\ell}}} 
 \Upsilon \left(\frac{Z(\gamma)}{2\pi i \hbar},
 \frac{\log \xi(-\gamma)}{2\pi i}\right)^{\Omega(\gamma)}
\end{equation}
where $\Upsilon$ is given explicitly in terms of the Barnes $G$-function {\rm (\cite{Barnes})}:
\begin{equation}
\label{eq:upsilon}
    \Upsilon(w,\eta)=\dfrac{e^{-\zeta'(-1)}e^{\frac{3}{4}w^2}G(w+\eta+1)}{(2\pi)^{w/2}w^{w^2/2}\Gamma(w+\eta)^\eta}.
\end{equation}
\end{prop}

\noindent The formula \eqref{eq:mintau} is a consequence of the relation 
$\frac{d}{dw}\log\Upsilon(w,\eta)=w\frac{d}{dw}\log\Lambda(w,\eta)$. We will also use the following formula which describes the asymptotic expansion of \eqref{eq:upsilon}:
\begin{equation} \label{eq:asymptotic-upsilon}
    \log \Upsilon(w,\eta) \sim - \frac{B_2(\eta)}{2} \log w + \sum_{k \ge 2} \frac{B_{k+1}(1 - \eta)}{(k-1)(k+1)} w^{1-k}
\end{equation}
for $w \to \infty$, valid when $w \to \infty$ in the complement of any open sector containing ${\mathbb R}_{< 0}$. 
See \cite[Appendix B.4]{AP} for these properties of $\Upsilon$-function.  
We also note that, when $\xi(\gamma) = 1$ holds for any active class $\gamma$, \eqref{eq:mintau} is reduced to the BPS $\tau$-function associated with the holomorphic solution \eqref{eq:X-hol} which was originally studied by Bridgeland in \cite{Bri19}: 
\begin{equation} \label{eq:tau-BPS-hol}
\tau_{\rm BPS}^{\rm hol}({\bm m}, \hbar) = 
\prod_{\substack{\gamma \in \Gamma \\ Z(\gamma) \in i {\mathbb H}_{\ell}}} 
 \Upsilon \left(\frac{Z(\gamma)}{2\pi i \hbar}\right)^{\Omega(\gamma)}.
\end{equation}
Here we write $\Upsilon(w) = \Upsilon(w,0)$.

\begin{rem}
In \cite{Bri20}, Bridgeland proved that the $\tau$-function for a 
BPS Riemann-Hilbert problem arising from the Donaldson-Thomas theory 
for the resolved conifold gives a ``non-perturbative partition function''
for the resolved conifold. This means that the asymptotic expansion
of the BPS $\tau$-function coincides with the generating series of 
all genus Gromov-Witten invariants. We expect there should be an analogue of our story in that setting as well.
\end{rem}

\section{BPS structures from spectral curves of hypergeometric type}
\label{sec:BPS-from-HG-curves}

In this section, we will recall the definition and computation of BPS structures associated to spectral curves of hypergeometric type, as described in detail in our previous work \cite{IK20}. 

\subsection{Spectral curves of hypergeometric type}
\label{section:hypergeometric-curvs-structure}

Let us summarize the geometric structure (i.e., properties of cycles and paths on) the {\em spectral curves of hypergeometric type} which are collected in Table \ref{table:classical}. We sometimes borrow terminology from {exact WKB analysis}; see \cite{KT98} for background.


In the rest of this section, we denote by $Q(x)$ one of the rational functions $Q_{\bullet}(x)$ for $\bullet \in \{{\rm HG}, {\rm dHG}, {\rm Kum}, {\rm Leg}, {\rm Bes}, {\rm Whi}, {\rm Web}, {\rm dBes}, {\rm Ai} \}$, and by $\varphi = Q(x) dx^2$ the associated meromorphic quadratic differential on $X={\mathbb P}^1$. We also introduce the following.
\begin{itemize}
\item 
$P \subset X$ is the set of poles of $\varphi$. It has a decomposition ${P} = {P}_{\rm ev} \cup {P}_{\rm od}$ into the subsets of even order poles and odd order poles of $\varphi$. Note that all our examples in Table \ref{table:classical} satisfy $\infty \in P$.
\item 
$T \subset X$ is the set of zeros and simple poles of $\varphi$. Borrowing a terminology in WKB literature, we call elements in $T$ {\em turning points} since they will be turning points of the quantum curves ${\bm E} = {\bm E}_{\bullet}$ constructed in \S \ref{sec:quantumcurves} below\footnote{
In the WKB analysis, $Q(x)$ plays the role of the potential function of a Schr\"odinger-type ODE. Traditionally, zeros of $\varphi$ are called turning points in WKB literature. It was pointed by \cite{Ko2} that simple poles of $\varphi$ also play similar roles to the usual turning points (called there {\em turning points of simple pole type}). Taking this into account, in this paper, by a turning point we mean a zero or a simple pole of $\varphi$. 
}. 

\item
The set ${\rm Crit} = P \cup T$ is called the set of {\em critical points} of $\varphi$. Its subset ${\rm Crit}_{\infty} = P \setminus T$ consists of {\em infinite critical points} of $\varphi$; that is, the poles of $\varphi$ of order $\ge 2$.  The points in $T = {\rm Crit} \setminus {\rm Crit}_\infty$ are also called {\em finite critical points} in this context.

\end{itemize}

The quadratic differential $\varphi$ is parametrized by a tuple of parameters ${\bm m} = ( m_s )_{s \in {P_{\rm ev}}}$ which we call {\em mass parameters} ({\em temperatures} in TR language). We will write $\varphi = \varphi({\bm m})$ when we emphasize the dependence on the mass parameters ${\bm m}$. In what follows, we assume the following condition on the mass parameters.

\begin{ass} \label{ass:genericity}
The mass parameters lie in the set $M = M_\bullet$: 
\begin{align}
M_{\rm HG} & := \{(m_0, m_1, m_\infty) \in {\mathbb C}^3 ~|~ 
m_0 m_1 m_\infty \Delta_{\rm HG}(\bm m) \ne 0 \}. 
\\
M_{\rm dHG} & := \{(m_1, m_\infty) \in {\mathbb C}^2 ~|~ 
m_1 m_\infty (m_1 + m_\infty)(m_1 - m_\infty) \ne 0 \}. 
\\
M_{\rm Kum} & := \{(m_0, m_\infty) \in {\mathbb C}^2 ~|~ 
m_0 (m_0 + m_\infty)(m_0 - m_\infty) \ne 0 \}.
\\
M_{\rm Leg} & := \{m_\infty \in {\mathbb C} ~|~ m_\infty \ne 0 \}.
\\
M_{\rm Bes} & := \{m_0 \in {\mathbb C} ~|~ m_0 \ne 0 \}. 
\\
M_{\rm Whi} & := \{m_\infty \in {\mathbb C} ~|~ m_\infty \ne 0 \}.
\\
M_{\rm Web} & := \{m_\infty \in {\mathbb C} ~|~ m_\infty \ne 0 \}.
\end{align}
where we set 
\begin{equation} \label{eq:HG-Delta}
\Delta_{\rm HG}(\bm m) = (m_0 + m_1 + m_{\infty})(m_0 + m_1 -m_{\infty})
(m_0 - m_1 + m_{\infty})(m_0 - m_1 - m_{\infty}).
\end{equation}
For degenerate Bessel curve (${\rm dBes}$) and Airy curve (${\rm Ai}$), 
we set $M_{\rm dBes} = M_{\rm Ai} = \emptyset$ since they have no mass parameters.
\end{ass}
It is not hard to check that $\varphi$ has only simple zeros, and no collision of zeros and poles occurs under the Assumption \ref{ass:genericity}.

Let us consider the Riemann surface
\begin{equation} \label{eq:spectral-curve-x-y}
\Sigma := \{ (x, y) \in {\mathbb C}^2 ~|~ y^2 - Q(x) = 0 \}, 
\end{equation}
associated with the quadratic differential $\varphi$.
This is a branched double cover of $X \setminus P$ with the projection map $\pi : (x, y) \mapsto x$. 
We denote by $\iota : (x,y) \mapsto (x,- y)$ the covering involution of $\Sigma$.  
We will mainly use $x$ as a local coordinate\footnote{
Technically speaking, we must fix branch cuts and identify the coordinate $x$ of 
the base ${\mathbb P}^1$ with a coordinate of $\Sigma \setminus \pi^{-1}(T)$ on the first sheet 
(then, $\iota(x)$ gives a coordinate in the second sheet).} on $\Sigma \setminus \pi^{-1}(T)$.
We note that, the Riemann surface $\Sigma$ for all the examples in Table \ref{table:classical} 
are of {\em genus $0$}, and hence topologically are punctured spheres, with punctures corresponding to poles of $\varphi$. 

We denote by $\overline{\Sigma}$ the compactification of $\Sigma$, obtained by adding preimages of poles of $\varphi$, on which the square root $\sqrt{\varphi} = \sqrt{Q(x)} \, dx$ gives a meromorphic differential. 
More precisely, to obtain $\overline{\Sigma}$, we add to $\Sigma$ a single point for each $s \in P_{\rm od}$, and add a pair $s_{\pm}$ of two points for each $s \in P_{\rm ev}$:  
\begin{equation}
\overline{\Sigma} = \Sigma \cup P_{\rm od} \cup \{s_+, s_- ~|~ s \in P_{\rm ev} \}.
\end{equation}
Here we identify $s \in P_{\rm od}$ with its unique lift on $\overline{\Sigma}$. 
We keep using the same notations $\pi : \overline{\Sigma} \to X$ and $\iota : \overline{\Sigma} \to \overline{\Sigma}$ for the projection map and the covering involution. 
The label of two points $s_{\pm}$, both of which are mapped to $s \in P_{\rm ev}$ by $\pi$, is chosen so that the following residue formula holds:
\begin{equation} \label{eq:sign-convention-preimages}
\Res_{x = s_{\pm}} \sqrt{Q(x)} \, dx = \pm m_s. 
\end{equation}

For later purposes, we also introduce a partial compactification $\widetilde{\Sigma}$ of $\Sigma$ obtained by filling punctures corresponding to simple poles of $\varphi$:  
\begin{equation}
\widetilde{\Sigma} := \Sigma \cup (P_{\rm od} \cap T) 
= \overline{\Sigma} \setminus D_{\infty}.
\end{equation}
Here we set 
\begin{equation}
D_\infty := \pi^{-1} ({\rm Crit}_{\infty}) ~\subset \overline{\Sigma}.
\end{equation}
If $\varphi$ has no simple poles, then $\widetilde{\Sigma} = \Sigma$. 

\begin{exa}[Gauss hypergeometric curve {\cite[\S 2.3.1]{IKoT-II}}] \label{exa:Gauss}
The Gauss hypergeometric curve $\Sigma_{\rm HG}$ is defined from the meromorphic quadratic differential 
$\varphi_{\rm HG} = Q_{\rm HG}(x) dx^2$ where the rational function 
\begin{equation} \label{eq:Sigma-HG}
Q_{\rm HG}(x) := \frac{m_{\infty}^2 x^2 - (m_\infty^2 - m_1^2 + m_0^2) x + m_0^2}{x^2(x-1)^2}.
\end{equation}
Under the assumption ${\bm m} = (m_0, m_1, m_\infty) \in M_{\rm HG}$, the poles are $P = P_{{ \rm ev}} = \{0, 1, \infty\}$ and the turning points are $T_{} = \{b_1, b_2 \}$, where $b_1, b_2$ are the two simple zeros of $Q_{\rm HG}(x)$ (since there are no simple poles, we have $P \cap T_{} = \emptyset$). 
Topologically, $\Sigma_{\rm HG}$ is a sphere with six punctures, and the  
compactification is given by 
\begin{equation}
\overline{\Sigma}_{\rm HG} := \Sigma_{\rm HG} \cup 
\{0_+, 0_-, 1_+, 1_-, \infty_+, \infty_- \}.  
\end{equation}
{  We note that $\widetilde{\Sigma}_{\rm HG} = \Sigma_{\rm HG}$ holds, 
and we have $D_{\infty} = \{0_+, 0_-, 1_+, 1_-, \infty_+, \infty_- \}$.}
\end{exa}

\subsection{Cycles on the spectral curve}
\label{section:cycle-path}

Let us fix notation for (relative) homology classes on $\widetilde{\Sigma}$ which will be used to describe the BPS structure in the rest of this section, and to define the Voros coefficients in \S\ref{sec:quantumcurves}. 
For spectral curves of hypergeometric type, since $\widetilde{\Sigma}$ is a punctured sphere, the homology group $H_1(\widetilde{\Sigma}, {\mathbb Z})$ is generated by the {\em residue cycles} $\gamma_a$ (i.e., the class represented by a positively oriented small circle) around the puncture $a \in D_\infty$ satisfying the relation
\begin{equation} \label{eq:relation-among-cycles}
\sum_{s \in P_{\rm od} \cap D_\infty} \gamma_s + \sum_{s \in P_{\rm ev}} (\gamma_{s_+} + \gamma_{s_-}) = 0.
\end{equation}
The covering involution of $\widetilde{\Sigma}$ also acts on the homology group, and in particular, we have $\iota_\ast \gamma_{s} = \gamma_{s}$ for $s \in P_{\rm od}$ and $\iota_\ast \gamma_{s_{\pm}} = \gamma_{s_{\mp}}$ for $s \in P_{\rm ev}$. 

On the other hand, for each $s \in P_{\rm ev}$, we may associate a relative homology class $\beta_s \in H_1(\overline{\Sigma}, D_\infty, {\mathbb Z})$ which is represented by a path from $s_-$ to $s_+$ on $\overline{\Sigma}$. The action of covering involution on this class is given by $\iota_\ast \beta_s = - \beta_s$. Among various relative homology classes, we are particularly interested in these classes.

Note that the intersection pairing $\langle \cdot, \cdot \rangle$ on $H_1(\widetilde{\Sigma},{\mathbb Z})$ is trivial. 
There is another intersection pairing 
defined through Poincar\'e-Lefschetz duality:
\begin{equation} \label{eq:non-degenerate-pairing}
( \cdot,\cdot ) : 
H_1(\widetilde{\Sigma}, {\mathbb Z}) \times 
H_1(\overline{\Sigma}, D_\infty, {\mathbb Z}) \to {\mathbb Z}
\end{equation}
which is non-degenerate. This allows us to identify $\beta \in H_1(\overline{\Sigma}, D_\infty, {\mathbb Z})$ with the element $(\cdot, \beta) \in H_1(\widetilde{\Sigma}, {\mathbb Z})^\vee = 
{\rm Hom}(H_1(\widetilde{\Sigma}, {\mathbb Z}), {\mathbb Z})$.
The intersection pairing between the above (relative) homology classes is given by\footnote{
We use the convention for the intersection pairing that $( \text{$x$-axis}, \text{$y$-axis} ) = 
- ( \text{$y$-axis}, \text{$x$-axis} ) = 1$.} 
\begin{equation} \label{eq:intersetion-gamma-beta}
( \gamma_{s_\pm},  \beta_{s'} ) = \pm \delta_{s,s'} 
\quad (s, s' \in P_{\rm ev}),
\end{equation}
where $\delta$ in the right hand side is the Kronecker symbol.

\subsection{BPS structures from quadratic differentials}
In the previous work \cite{IK20}, we constructed BPS structures from the quadratic differentials in Table \ref{table:classical}. Here we briefly recall the construction based on the spectral network associated with these quadratic differentials. We note that the construction of BPS structure given here works more generally (see \cite[\S 7]{Bri19}).

\subsubsection{Lattice and central charge} 
Recall that a BPS structure requires a charge lattice $\Gamma$ and a central charge $Z$. Following \cite{GMN09, BS13}, we define:
\begin{dfn} \label{def:central-charge}
The charge lattice $\Gamma$ is the sublattice 
of $H_1(\widetilde{\Sigma}, {\mathbb Z})$ defined by
\begin{equation}
\Gamma := \{ \gamma \in H_1(\widetilde{\Sigma}, {\mathbb Z}) 
~|~ \iota_\ast \gamma = - \gamma \}
\end{equation}
equipped with the intersection pairing $\langle \cdot , \cdot \rangle : \Gamma \times \Gamma \to {\mathbb Z}$. The central charge $Z : \Gamma \to {\mathbb C}$ is the period integral of $\sqrt{\varphi}$: 
\begin{equation}
Z(\gamma) := \oint_{\gamma} \sqrt{Q(x)} \, dx.
\end{equation}
\end{dfn}

The lattice $\Gamma$ is called the hat-homology group in \cite{BS13}. It will sometimes be convenient to extend the domain of definition of $Z$ from $\Gamma$ to the whole lattice $H_1(\widetilde{\Sigma}, {\mathbb Z})$ in a natural manner. Since any element $\gamma \in \Gamma$ can be written as a sum of $\gamma_{s_\pm}$ in $H_1(\tilde{\Sigma}, {\mathbb Z})$, we may use the formula 
\begin{equation}
Z(\gamma_{s_\pm}) =\pm 2 \pi i m_{s} \quad (s \in P_{\rm ev}), 
\end{equation}
which follows from \eqref{eq:sign-convention-preimages}, to express the central charge $Z(\gamma)$ for a general element. We will write $\Gamma = \Gamma_{\bm m}$ and $Z = Z_{\bm m}$ when we discuss the dependence on the mass parameters. 

As we have seen in previous subsection, the intersection pairing $\langle \cdot , \cdot \rangle$ is trivial in our example. Thus the BPS structures we will obtain are automatically uncoupled.

\subsubsection{Spectral networks}
To define the BPS indices, we use the notion of a \emph{WKB spectral network} \cite{GMN12} associated to the quadratic differential $\varphi$ considered in the previous subsection\footnote{A more general construction can be made for arbitrary tuples of meromorphic $k$-differentials, but it is more complicated. Except briefly in \S\ref{sec:conjectures}, we only consider BPS structures arising from quadratic differentials in this paper, so we omit the details.}. 
This can also be understood as the \emph{Stokes graph} (c.f., \cite{KT98}) of the quantum curve ${\bm E}({\bm \nu})$ which will be constructed in \S \ref{sec:quantumcurves} through the topological recursion.

Let us define a \emph{trajectory} of the quadratic differential to be any maximal curve on $X$ along which
\begin{equation}
\mathrm{Im}\, e^{-i \vartheta}\int^{x}{ \sqrt{Q(x)} dx} = {\rm constant}.
\end{equation}
\noindent is satisfied. There is a well-constrained  possible local behaviour of these trajectories, but a rich global behaviour. The general structure of trajectories of quadratic differentials is described in detail in the classic book by Strebel \cite{St84} (see also \cite{BS13}). 

We are primarily interested in trajectories with at least one endpoint on a turning point --- we call these \emph{critical trajectories}. The set of (oriented) critical trajectories is called the \emph{spectral network $\mathcal{W}_\vartheta({\varphi})$ of $\varphi$ at phase $\vartheta$}. In the WKB literature, the critical trajectories are called {\em Stokes curves}, and form the locus where the Borel-resummed WKB solutions have a discontinuity.

In this paper, we will only use the end result of the previous work \cite{IK20}, so we will not give details and figures of spectral networks, but we refer the reader to that paper, and proceed to write down the BPS structure that was obtained there.

\subsubsection{Degenerate spectral networks and saddle trajectories}
\label{subsection:saddle-and-BPS-indices}

We say the spectral network $\mathcal{W}_\vartheta(\varphi)$ is \emph{degenerate} if it contains a {\em saddle trajectory} (i.e., a critical trajectory both of whose endpoints are turning points) of phase $\vartheta$. 
%
%
These degenerate spectral networks are of paramount importance in this paper and many applications. In the physics of 4d $\mathcal{N}=2$ QFTs, they correspond to BPS states in the spectrum of the theory \cite{GMN08, GMN12}. From a mathematical point of view, they correspond to stable objects in a 3-Calabi-Yau category associated with a quiver with potential determined by $\varphi$ \cite{BS13}. 

Let us summarize the possible types of saddle trajectories which can appear in degenerate spectral networks arising from our examples\footnote{Although there is a fifth type of saddle trajectory (which forms the boundary of a ``non-degenerate ring domain'' and the associated BPS cycles gives $\Omega(\gamma_{\rm BPS}) = -2$ as its BPS invariant) in general, we will not treat it here since this type never appears in our examples.}: 
\begin{itemize}
\item 
A {\em type I}\, saddle connects two distinct simple zero of $\varphi$. 
\item 
A {\em type II}\, saddle connects a simple zero and a simple pole of $\varphi$.
\item 
A {\em type III}\, saddle connects two distinct simple poles of $\varphi$.
\item 
A {\em type IV}\, saddle is a closed curve which forms the boundary of a degenerate ring domain (i.e., a maximal region foliated by closed trajectories around a second order pole of $\varphi$).
\end{itemize}

\subsubsection{Definition of BPS invariants}

Given a degenerate spectral network with a saddle trajectory of type I, II, or III, we can associate a homology class $\gamma_{\rm BPS} \in \Gamma$ represented by the closed cycle (up to its orientation) on $\widetilde{\Sigma}$ obtained as the pullback by $\pi : \widetilde{\Sigma} \to X$ of the saddle trajectory. On the other hand, for a degenerate spectral network with a type IV saddle, we define  $\gamma_{\rm BPS} = \gamma_{s_+} - \gamma_{s_-} \in \Gamma$ (up to its orientation) where $s$ is the second order pole of $\varphi$ in the degenerate ring domain surrounded by the type IV saddle. We refer to homology classes $\gamma_{\rm BPS} \in \Gamma$ defined in this way as \emph{BPS cycles}, and the collection of all BPS cycles (together with their type) as the \emph{BPS spectrum of $
\varphi$}. Degenerate spectral networks only appears when the phase $\vartheta$ coincides with the argument of the central charge $Z(\gamma_{\rm BPS})$ of a certain BPS cycle. 

To define the BPS invariants, we should avoid a locus in the parameter space $M$ of $\varphi = \varphi({\bm m})$ on which a multiple degeneration appears in the spectral network simultaneously. We say the parameter ${\bm m}$ is {\em generic} if it lies in the complement $M \setminus W$ of a set $W = W_\bullet$ defined by
\begin{equation}
W_\bullet = \Bigl\{ {\bm m} \in M_\bullet ~|~ 
\text{there exist distinct BPS cycles $\gamma, \gamma'$ satisfying 
$\frac{Z_{\bm m}(\gamma)}{Z_{\bm m}(\gamma')} \in {\mathbb R}_{> 0}$} \Bigr\}
\end{equation}

\noindent for $\bullet \ne {\rm Leg}$, and $W_{\rm Leg} = \emptyset$ for the Legendre case. 
Otherwise, we say ${\bm m}$ is \emph{non-generic}. We denote by $M' = M \setminus W$ the {\em locus of generic parameters}.

Now we recall the definition of the BPS indices introduced in \cite{GMN09, GMN12, BS13, IK20} (see also the recent work \cite{Haiden} which arrives at the same BPS indices for BPS cycles associated with Type II and III saddles via Donaldson-Thomas theory).

\begin{dfn}
For each ${\bm m} \in M'$, we define the {\em BPS indices} $\{ \Omega(\gamma) \}_{\gamma \in \Gamma}$ as a collection of integers defined as follows: 
\begin{equation} \label{eq:def-of-BPS-indices}
\Omega(\gamma)= \begin{cases} 
      +1 & \quad \text{if $\gamma$ is a BPS cycle associated with a type I saddle,} \\
      +2 & \quad \text{if $\gamma$ is a BPS cycle associated with a type II saddle,} \\
      +4 & \quad \text{if $\gamma$ is a BPS cycle associated with a type III saddle}, \\
      - 1 & \quad \text{if $\gamma$ is a BPS cycle associated with a degenerate ring domain}. 
   \end{cases}
\end{equation}
and we set $\Omega(\gamma) = 0$ for any non-BPS cycles $\gamma \in \Gamma$.
\end{dfn}

In our previous work, \cite{IK20}, given any of the examples of $\varphi$ from Table \ref{table:classical}, we determined the BPS structure following $\eqref{eq:def-of-BPS-indices}$ at any ${\bm m} \in M'$. The results are summarized in Table 3 below. We refer the reader to \cite{IK20} for additional details.

\begin{thm}[{\cite[\S 4]{IK20}}]
Let $\bullet$ be the label of any spectral curve of hypergeometric type in Table \ref{table:classical}. For each ${\bm m} \in M_\bullet'$, the BPS structure determined by $\varphi_\bullet({\bm m})=Q_{\bullet}(x)dx^{2}$ is the one given in Table \ref{table:BPS-str-HG}. In particular,  the BPS structure obtained is finite, integral and uncoupled.
\end{thm}

\begin{table}[t]
\begin{center}
\begin{tabular}{cccc}\hline
Spectral curve & $\pm\gamma_{\rm BPS}$ & $\pm Z(\gamma_{\rm BPS})/2\pi i$ & $\Omega(\gamma_{\rm BPS})$ 
\\\hline\hline
\parbox[c][2.5em][c]{0em}{}
Gauss 
& $ \pm (\gamma_{0_+} + \gamma_{1_\epsilon } + \gamma_{\infty_\epsilon' }) $ 
\hspace{+.5em} $(\epsilon, \epsilon' \in \{ \pm \})$ 
&$ \pm ( m_0+\epsilon\, m_1+ \epsilon' m_\infty)$
& $1$
\\[-.5em] 
& $\pm (\gamma_{s_+} - \gamma_{s_-})$
\quad $(s \in \{0,1,\infty \})$
&$\pm \, 2 m_s$
& $-1$
\\[+.3em]\hline
\parbox[c][2.5em][c]{0em}{}
\begin{minipage}{.15\textwidth}
\begin{center} Degenerate Gauss \end{center}
\end{minipage}
& $\pm (\gamma_{1_+} + \gamma_{\infty_\epsilon})$ 
\quad $(\epsilon \in \{ \pm \})$
&$\pm (m_1+\epsilon\, m_\infty)$
& $2$
\\[-.em]
& $\pm (\gamma_{s_+} - \gamma_{s_-})$
\quad $(s \in \{1,\infty \})$
&$\pm 2 m_s$
& $-1$
\\[+.3em]\hline
\parbox[c][2.5em][c]{0em}{}
Kummer 
& $\pm (\gamma_{0_+} +  \gamma_{\infty_\epsilon})$ 
~ $(\epsilon \in \{ \pm \})$
& $\pm (m_0 + \epsilon \, m_\infty)$
& $1$
\\[-.em]
& $\pm (\gamma_{0_+} - \gamma_{0_-})$
&$\pm 2 m_0$
& $-1$
\\[+.3em]\hline
\parbox[c][2.5em][c]{0em}{}
Legendre 
& $\pm \gamma_{\infty_+}$
&$\pm m_\infty$
& $4$ 
\\[-.em]
& $\pm (\gamma_{\infty_+} -  \gamma_{\infty_-})$
& $\pm 2 m_\infty$
& $-1$
\\[+.3em]\hline
\parbox[c][2.5em][c]{0em}{}
Bessel 
&
$ \pm (\gamma_{0_+} - \gamma_{0_-})$
& $\pm 2 m_0$
& 
$-1$
\\\hline
\parbox[c][2.5em][c]{0em}{}
Whittaker 
& $\pm \gamma_{\infty_+}$
& $\pm m_\infty$
& $2$
\\\hline
\parbox[c][2.5em][c]{0em}{}
Weber 
& $\pm \gamma_{\infty_+}$
&$\pm m_\infty$
& $1$
\\\hline
\end{tabular}
\end{center}
\caption{The BPS structures (BPS cycles, central charges, and BPS indices) associated 
to the spectral curves of hypergeometric type in Table \ref{table:classical}. 
As before, for $s \in P_{\rm ev}$, 
$\gamma_{s_{\pm}}$ is the positively oriented small circle around $s_{\pm}$.  
}
\label{table:BPS-str-HG}
\end{table}

The uncoupledness of the BPS structures in Table \ref{table:BPS-str-HG} is a consequence of the vanishing of the intersection pairing $\langle \cdot, \cdot \rangle$ on $\Gamma$. In particular, this vanishing implies that the associated BPS Riemann-Hilbert problem is trivial (i.e., all BPS automorphisms are identity maps). Therefore, for an arbitrarily chosen constant term $\xi \in {\mathbb T}_{-}$, the functions
\begin{equation}
X_{\ell, \gamma}(\hbar) =  e^{- Z(\gamma)/\hbar} \, \xi(\gamma) 
\qquad (\gamma \in \Gamma)
\end{equation}
attached to each ray $\ell \subset {\mathbb C}^\ast$ give a solution to the BPS Riemann-Hilbert problem, and it is the unique holomorphic solution which has no zeros or poles on ${\mathbb H}_{\ell}$. 
To obtain a non-trivial BPS Riemann-Hilbert problem, which will be compared with the Stokes structure arising from topological recursion and quantum curves, we will consider the almost-doubled BPS structure discussed next. 

\subsection{Doubled and almost-doubled BPS structure}
\label{sec:doubling}

Given a BPS structure, we may construct a new \emph{doubled BPS structure}  $(\Gamma_{\rm D}, Z_{\rm D}, \Omega_{\rm D})$ via the so-called doubling construction \cite[\S 2.8]{Bri19}. The doubled lattice is:

\begin{equation}
\Gamma_{\rm D}:=\Gamma \oplus \Gamma^{\vee},
\end{equation}
equipped with the new pairing 
\begin{equation}
\label{eq:doubledpairing}
\langle(\gamma_{1},\beta_{1}),(\gamma_{2},\beta_{2})\rangle 
:= \langle\gamma_{1},\gamma_{2}\rangle +\beta_{1}(\gamma_{2})-\beta_{2}(\gamma_{1})
\end{equation}
which is always nondegenerate. The central charge can be extended by an arbitrary homomorphism $Z^\vee: \Gamma^\vee \rightarrow \mathbb{C}^*$ as $Z_{\rm D}(\gamma,\beta):=Z(\gamma)+Z^\vee(\beta)$, and the BPS indices are defined to be 
\begin{equation}
\Omega_{\rm D}(\gamma,\beta) := \begin{cases} 
\Omega(\gamma) & \text{if $\beta = 0$}, \\
0 &  \text{otherwise}.
\end{cases} 
\end{equation}
Clearly the symmetry property continues to hold, and the support property is also equivalent to the support property for the original BPS structure. 

As we did in \S\ref{section:cycle-path}, we identify the dual lattice $\Gamma^\vee$ with a sublattice of $H_1(\overline{\Sigma}, D_\infty, {\mathbb Z})$ through the intersection paring $(\cdot,\cdot)$ given in \eqref{eq:non-degenerate-pairing}. 
Through the identification, the extended paring $\langle \cdot, \cdot \rangle$ on $\Gamma_{D}$ and the intersection pairing $(\cdot,\cdot)$ are related as follows:
\begin{equation} \label{eq:sign-of-pairings}
\langle \gamma, \beta \rangle = - \beta(\gamma) = - (\gamma, \beta)
\qquad (\gamma \in \Gamma, \, \beta \in \Gamma^\vee).
\end{equation}
Here and in what follows, we identify $\gamma \in \Gamma$ and $\beta \in \Gamma^\vee$ with $(\gamma,0) \in \Gamma_{\rm D}$ and $(0,\beta) \in \Gamma_{\rm D}$, respectively.

We will in fact use a slightly modified doubling procedure, which we call \emph{almost-doubling}, in which we take a full-rank sublattice 
\begin{equation}
\Gamma^{*} = \{ \beta \in H_1(\overline{\Sigma}, D_\infty; {\mathbb Z}) ~|~ \iota_\ast \beta = - \beta \}  \subset \Gamma^\vee
\end{equation}
in the second factor, and define 
\begin{equation} \label{eq:Gamma-almost-doubled-lattice}
\Gamma_{\text{\DJ}}:=\Gamma\oplus\Gamma^*    
\end{equation} 
equipped with the restricted pairing (denoted by the same symbol $\langle \cdot, \cdot \rangle$). 
 Together with the obvious restrictions $Z_{\text{\DJ}}$ and $\Omega_{\text{\DJ}}$ of $Z_{\rm D}$ and $\Omega_{\rm D}$ to the sublattice $\Gamma_{\text{\DJ}}$, we obtain a BPS structure $(\Gamma_{\text{\DJ}}, Z_{\text{\DJ}}, \Omega_{\text{\DJ}})$ which we call the  {\em almost-doubled BPS structure} associated to the original. The purpose of this restriction will be to ensure our solution to the BPS Riemann-Hilbert problem is single-valued; see Remark \ref{rem:use-of-almost-doubling} below. Note that the relative cycles $\beta_s$ defined for $s \in  P_{\rm ev}$ in \S \ref{section:cycle-path} are examples of paths satisfying $\iota_\ast \beta_s = - \beta_s$ since the covering involution exchanges $s_{+}$ and $s_{-}$. 
In fact, we can verify that $\{ \beta_s \}_{s \in P_{\rm ev}}$ generates the sublattice $\Gamma^{*}$.

The doubling and almost-doubling construction is useful since it produces nontrivial BPS Riemann-Hilbert problems from BPS structures which may have a trivial Riemann-Hilbert problem. In particular, even though the integers $\Omega(\gamma)\langle \gamma,\gamma' \rangle$ may be zero for the original BPS structure, it is easy to see that (almost-)doubling will allow the analogous expression to take on nonzero values. This is why our introduction of $\Omega(\gamma)=-1$ plays a significant role, but did not affect previous results on the subject -- any choice of $\Omega$ attributed to such BPS cycles would be killed by the intersection pairing in the original BPS structure.

\subsection{Miniversality of almost-doubled BPS structure}

In \cite{IK20} and above, the BPS structures corresponding to hypergeometric type spectral curves were considered as depending on a set of parameters ${\bm m}$, but we did not consider them as a family. However, they naturally fit together to form a variation of BPS structures:

\begin{prop}[{c.f., \cite[Claim 7.1]{Bri19}}] \label{prop:constructing-VBPS}
The family $(\Gamma_{\bm m}, Z_{\bm m}, \Omega_{\bm m})_{\,{\bm m} \in M'}$ gives a {(framed)} 
miniversal variation of finite, integral, uncoupled BPS structures over $M'$.  
Moreover, we can naturally extend the family to a miniversal variation of finite, integral, uncoupled BPS structures on the whole parameter space $M$. 
\end{prop}
\begin{proof}
Since no collision of zeros and poles of $\varphi$ occurs on the parameter space $M$, the collection $\{ \Gamma_{\bm m} \}_{{\bm m} \in M}$ of lattices forms a local system on the space. The rank of $\Gamma_{\bm m}$ are computed by using \cite[Lemma 2.2]{BS13}, and we may verify that the rank of each lattice $\Gamma_{\bm m}$ coincides with the dimension of $M$  (i.e., the number of even order poles of $\varphi$) by case-by-case checking.  {The triviality of this local system follows from the isomorphism described below in \eqref{eq:gammaiso}.} 
Moreover, the explicit expression of central charges in Table \ref{table:BPS-str-HG} shows that the period map 
is holomorphic and local isomorphism.

Although $\Omega_{{\bm m}}$ was not defined for ${\bm m} \in W$ since a multiple degeneration of spectral network can occur there, the definition of $\Omega_{{\bm m}}$ can be extended by continuity across $W$ due to our observation in \cite{IK20} which tells the set of BPS cycles and their BPS invariants are common for all ${\bm m} \in M'$ as is shown in Table \ref{table:BPS-str-HG} (that is, Table \ref{table:BPS-str-HG} can be taken as the definition of the BPS structure for all ${\bm m}\in M$). The family of BPS structure on $M$ thus obtained satisfies the support property and wall-crossing formula since our BPS structures are finite and uncoupled (all BPS automorphisms commute). Thus we have a miniversal variation of finite, integral, uncoupled BPS structures over $M$.
\end{proof}

Thanks to the miniversality of the original BPS structure, we have an identification between the tangent space $T_{\bm m}M$ and $\Gamma_{\bm m}^*\otimes\mathbb{C}$ ($\simeq\Gamma_{\bm m}^\vee\otimes\mathbb{C}$) via the derivative of the period map \eqref{eq:periodmap}. Thus the space $\Hom(\Gamma_{\bm m}^\ast, {\mathbb C})$, which parametrizes the arbitrary homomorphism $Z_{\bm m}^\vee$, is identified with the cotangent fiber $T_{\bm m}^\ast M$.  It is not hard to show then that:


\begin{coro}\label{cor:mini} For any family of BPS structures of hypergeometric type $(\Gamma_{\bm m},Z_{\bm m},\Omega_{\bm m})_{\,{\bm m} \in M}$, the family of almost-doubled BPS structures $(\Gamma_{\text{\rm \DJ}, {\bm p}}, Z_{\text{\rm \DJ}, {\bm p}}, \Omega_{\, {\text{\rm \DJ}}, \bm p})_{{ {\bm p}}\in T^*M}$ gives a miniversal variation of finite, integral, uncoupled BPS structures over $T^\ast M$. 
\end{coro}



Note that the same statement holds for the fully doubled family as well.


\subsection{Almost-doubled BPS Riemann-Hilbert problem and the BPS $\tau$-function}

We will be interested in the BPS Riemann-Hilbert problem (Problem \ref{prob:holRHP}) associated with the almost-doubled BPS structure constructed in the previous subsection. We will call it the {\em almost-doubled BPS Riemann-Hilbert problem} below.   

First of all, it is easy to observe that, the almost-doubled BPS structure $(\Gamma_{\text{\DJ}}, Z_{\text{\DJ}}, \Omega_{\text{\DJ}})$ constructed above is still {\em uncoupled}. This is because the original BPS structure $(\Gamma, Z, \Omega)$ is uncoupled, and the corresponding BPS indices $\Omega_{\text{\DJ}}$ are supported on the BPS cycles in the original BPS structure. Therefore, we may happily use the general formula in Theorem \ref{thm:Anna-RHP-sol} for the minimal solution even in the (almost)-doubled BPS Riemann-Hilbert problem, once we have fixed an extended constant term $\xi_{\text{\DJ}}$ which takes value in the almost-doubled twisted torus  
\begin{align}
{\mathbb{T}_{-,{\text{\DJ}}}} := \left\{\xi_{\text{\rm \DJ}} : \Gamma_{\text{\text{\DJ}}} \rightarrow \mathbb{C}^\ast ~ | ~ \xi_{\text{\rm \DJ}}{(\mu_1+\mu_2)}= (-1)^{\langle \mu_1,\mu_2\rangle}\xi_{\text{\rm \DJ}}(\mu_1)\xi_{\text{\rm \DJ}}(\mu_2)\right\}. 
\end{align}
whose elements may be written as $\xi_{\text{\rm \DJ}}=(\xi,\xi^\vee)$ through the splitting \eqref{eq:Gamma-almost-doubled-lattice}. Applying Theorem \ref{thm:Anna-RHP-sol}, we have the following expression of the minimal solution to the almost-doubled BPS Riemann-Hilbert problem associated with $(\Gamma_{\text{\DJ}}, Z_{\text{\DJ}}, \Omega_{\text{\DJ}})$ evaluated at $\mu \in \Gamma_{\text{\DJ}}$:  
\begin{equation} \label{eq:minimal-solution-almost-doubled}
X_{\ell, \mu}^{\rm min}(\hbar) = e^{- Z_{\text{\DJ}}(\mu)/\hbar} \, \xi_{\text{\DJ}}(\mu) \,  \prod_{\substack{\eta \in \Gamma_{\text{\DJ}} \\ Z_{\text{\DJ}}(\eta) \in i {\mathbb H}_{\ell} }} 
{\displaystyle\Lambda}\left(\frac{Z_{\text{\DJ}}(\eta)}{2 \pi i \hbar},\frac{\log{\xi}_{\text{\DJ}}(-\eta)}{2\pi i} \right)^{\Omega_{\text{\DJ}}(\eta) \langle \mu, \eta \rangle}
\end{equation}
or more explicitly in terms of the splitting $\Gamma_{\text{\rm \DJ}}=\Gamma\oplus\Gamma^*$: 
\begin{align}
X_{\ell, \gamma}^{\rm min}(\hbar) =&\,e^{- Z(\gamma)/\hbar} \, \xi(\gamma) \\
X_{\ell, \beta}^{\rm min}(\hbar) =& \,e^{- Z^\vee(\beta)/\hbar} \, \xi^\vee(\beta) \,  \cdot \prod_{\substack{\gamma' \in \Gamma\\ Z(\gamma') \in i {\mathbb H}_{\ell} }} 
 {\displaystyle\Lambda}\left(\frac{Z(\gamma')}{2 \pi i \hbar},\frac{\log{\xi}(-\gamma')}{2\pi i} \right)^{\Omega(\gamma') \langle \beta, \gamma' \rangle}
\end{align}
In \S \ref{sec:final}, we will compare this minimal solution to our Voros solution which arises from the topological recursion and quantum curves. Note from the expression the jumping structure is clear: $X^{\rm min}_{\ell,\gamma}$ do not jump at all and are very simple, whereas the $X^{\rm min }_{\ell,\beta}$ jump across BPS rays of the original $(\Gamma,Z,\Omega)$.

To describe the BPS $\tau$-function, let us use the mass parameters $(m_s)_{s \in P_{\rm ev}}$ which is a natural global coordinate on $M$ (as we have seen, $\dim M$ coincides with the number of even order poles of $\varphi$.)  To realize the coordinate as the central charges of a basis, it is convenient to use cycles with rational coefficients; that is, if we define a lattice 
\begin{equation}
    \Gamma_{\rm can} = \bigoplus_{s \in {P_{\rm ev}}} {\mathbb Z} \lambda_s, 
\end{equation}
together with a linear map $Z_{\rm can} : \Gamma_{\rm can} \to {\mathbb C}$ given by $Z_{\rm can}(\lambda_s) = 2 \pi i m_s$, then we have an isomorphism 
\begin{equation}\label{eq:gammaiso}
\Gamma_{\rm can} \otimes {\mathbb Q} \simeq \Gamma_{\bm m} \otimes {\mathbb Q} ~ : ~
\lambda_s \mapsto \frac{\gamma_{s_+} - \gamma_{s_-}}{2}
\end{equation} 
which identifies $Z_{\rm can}$ with $Z_{\bm m}$ (after $\otimes {\mathbb Q}$) for any ${\bm m} \in M$. We need such rational coefficients since $\{ \gamma_{s_+} - \gamma_{s_-} \}_{s \in P_{\rm ev}}$ is not a ${\mathbb Z}$-basis of $\Gamma_{\bm m}$ in general. 
In view of $\langle \beta_s , \gamma_{s'_+} - \gamma_{s'_-} \rangle = 2 \, \delta_{s, s'}$, we may also canonically identify the (almost) dual lattices $\Gamma_{\rm can}^{\vee} \simeq \Gamma^\ast_{\bm m}$ by identifying the dual element of $\lambda_s$ with $\beta_s$.

Thanks to the miniversality, we have the following identification: 
\begin{equation} \label{eq:derivative-of-period-map}
T_{\bm m} M \simeq \Gamma_{\bm m}^\vee \otimes {\mathbb C} \simeq \Gamma^\vee_{\rm can} \otimes {\mathbb C}. 
\end{equation}
Thus we may use the natural coordinate $(m_s)_{s \in P_{\rm ev}}$ as $z_i$'s in the definition of BPS $\tau$-function \eqref{eq:defining-tau}. 
More precisely, in terms of the coordinate 
$(m_s, m^\vee_s )_{s \in P_{\rm ev}}$ of $T^\ast M$ 
given by $Z^\vee_{\bm m}(\beta_s) = 2 \pi i m^\vee_s$, 
the BPS $\tau$-function for our almost-doubled BPS structure is defined by 
\begin{equation}
\label{eq:tau-concrete}
\frac{\partial}{\partial m^\vee_s} \log \tau_{\rm BPS, \ell}  = 0, 
\qquad
\frac{\partial}{\partial m_s} \log \tau_{\rm BPS, \ell}  = 
- \frac{\partial}{\partial \hbar} \log Y_{\ell, \beta_s}
\qquad (s \in P_{\rm ev}).
\end{equation}
Therefore, the BPS $\tau$-function does not depend on $m_s^\vee$. 
In particular, the BPS $\tau$-function associated with the minimal solution \eqref{eq:minimal-solution-almost-doubled} has the same expression given in Proposition \ref{prop:annatau}; that is, it is written as a product of $\Upsilon$-functions over the original lattice $\Gamma$. We will set $Z^\vee = 0$ (i.e., $m_s^\vee = 0$) in Section \ref{sec:final}, but this restriction evidently does not change the BPS $\tau$-function. 

\section{Topological recursion, quantum curves and Voros coefficients}
\label{sec:quantumcurves}

In the works \cite{IKoT-I, IKoT-II}, the first named author and collaborators studied the quantization of spectral curves of hypergeometric type via {\em topological recursion} (TR). See also \cite{BE-16} for generalities on the quantization of spectral curves, i.e. {\em quantum curves}. Moreover, an explicit expression for the Voros symbols of the associated quantum curves was obtained. 
In this section, we recall the topological recursion and the construction of the quantum curves, as well as the definition and form of their Voros symbols from those works, which we will use to solve to the Riemann-Hilbert problem in \S\ref{subsection:solving-BPS-RHP-by-Voros}.

\subsection{Definition of correlators and partition function}

Here we briefly recall several facts from Eynard-Orantin's theory of topological recursion (\cite{EO}) and define the objects playing a central role on the in the present paper.

\begin{dfn} \label{def:spectral-curve-TR}
A \emph{spectral curve} is a tuple $({\mathcal C}, x, y, B)$ of the following data:
\begin{itemize}
\item
a compact Riemann surface ${\mathcal C}$, 
\item 
a pair of non-constant meromorphic functions $x,y$ on ${\mathcal C}$ such that 
$dx$ and $dy$ never vanish simultaneously, and 
\item
a symmetric meromorphic bidifferential\footnote{
A bidifferential is a multidifferential with $n = 2$.} 
$B$ on ${\mathcal C}$ 
having a second order pole with biresidue $1$ along the diagonal, and holomorphic elsewhere.
\end{itemize}
\end{dfn}

We usually denote by $z$ a local coordinate of ${\mathcal C}$, 
and by $z_i$ a copy of the coordinate when we consider {\em multidifferentials}.  
Here, a meromorphic multidifferential is a meromorphic section of the line bundle 
$\pi_1^\ast(T^\ast {\mathcal C}) \otimes \cdots \otimes \pi_n^\ast(T^\ast {\mathcal C})$
on ${\mathcal C}^n$, where $\pi_j : {\mathcal C}^n \to {\mathcal C}$ is the $j$th projection map
\cite{DN16}. We often drop the symbol $\otimes$ when we express multidifferentials.

We denote by ${\mathcal R}$ the set of \emph{ramification points} of $x$, 
that is, ${\mathcal R}$ consists of zeros of $dx$;
here we consider $x$ as a branched covering map
${\mathcal C} \to \mathbb{P}^1$, and we call the images of ramification points
by $x$ the {\em branch points}.
We also assume that all ramification points of $x$ 
are simple so that the local conjugate map $z \mapsto \bar{z}$ 
near each ramification point is well-defined. 
Then, the \emph{topological recursion} (TR) is formulated as follows. 

\begin{dfn}[{\cite[Definition 4.2]{EO}}]
The \emph{Eynard-Orantin correlator} 
$W_{g, n}(z_1, \cdots, z_n)$ for $g \geq 0$ and $n \geq 1$
is defined as a meromorphic multidifferential on ${\mathcal C}$ 
by the recursive relation
\begin{align}
\label{eq:TR}
W_{g, n+1}(z_0, z_1, \cdots, z_n)
&:= \sum_{r \in {\mathcal R}}
\Res_{z = r} K_r(z_0, z)
\Bigg[
W_{g-1, n+2} (z, \overline{z}, z_1, \cdots, z_n)
\\
&\qquad\qquad
+
\sum_{\substack{g_1 + g_2 = g \\ I_1 \sqcup I_2 = \{1, 2, \cdots, n\}}}'
W_{g_1, |I_1| + 1} (z, z_{I_1})
W_{g_2, |I_2| + 1} (\overline{z}, z_{I_2})
\Bigg]
\notag
\end{align}
for $2g + n \geq 2$ with initial conditions given by
\begin{align}
W_{0, 1}(z_0) &:= y(z_0) dx(z_0),
\quad
W_{0, 2}(z_0, z_1) := B(z_0, z_1).
\end{align}
Here we set $W_{g,n} \equiv 0$ for a negative $g$, 
$K_r$ denotes the so-called ``recursion kernel'' given by
\begin{equation}
\label{eq:RecursionKernel}
K_r(z_0, z)
:= \frac{1}{2\big(y(z) - y(\overline{z})\big) dx(z)}
\int^{\zeta = z}_{\zeta = \overline{z}} B(z_0, \zeta)
\end{equation}
defined near a ramification point $r \in {\mathcal R}$,
$\sqcup$ denotes the disjoint union,
and the prime $'$ on the summation symbol in \eqref{eq:TR}
means that we exclude terms for
$(g_1, I_1) = (0, \emptyset)$
and
$(g_2, I_2) = (0, \emptyset)$
(so that $W_{0, 1}$ does not appear) in the sum.
We have also used the multi-index notation:
for $I = \{i_1, \cdots, i_m\} \subset \{1, 2, \cdots, n\}$
with $i_1 < i_2 < \cdots < i_m$, $z_I:= (z_{i_1}, \cdots, z_{i_m})$.
\end{dfn}

One of the central objects in Eynard-Orantin's theory is the \emph{genus $g$ free energy} $F_g$ ($g\geq 0$) associated to the spectral curve:
\begin{dfn}[{\cite[Definition 4.3]{EO}}]
For $g \geq 2$, the \emph{genus $g$ free energy} $F_g$ is defined by
\begin{equation}
\label{def:Fg2}
F_g := \frac{1}{2- 2g} \sum_{r \in {\mathcal R}} \Res_{z = r}
\big[\Phi(z) W_{g, 1}(z) \big],
\end{equation}
where $\Phi(z)$ is any primitive of $y(z) dx(z)$. The free energies $F_0$ and $F_1$ for $g=0$ and $1$ are also defined, but in a different manner (see \cite[\S 4.2.2 and \S 4.2.3]{EO} for the definition). 
\end{dfn}

We note that the right-hand side of \eqref{def:Fg2} indeed does not depend on the choice of $\Phi(z)$. 

\begin{dfn}
The generating series
\begin{equation} \label{eq:total-free-energy}
F_{\rm TR}(\hbar) := \sum_{g = 0}^{\infty} \hbar^{2g-2} F_g
\end{equation}
of $F_g$ is called the \emph{free energy} of the spectral curve. 
Its exponential 
\begin{equation}
Z_\mathrm{TR}(\hbar) := e^{F_{\rm TR}(\hbar)}
\end{equation} 
is called the \emph{(topological recursion) partition function} of the spectral curve. 
\end{dfn}

The free energy and the partition functions are regarded as a formal series of $\hbar$. It is known that usually these formal series are divergent (see \cite{Eynard-19} for the growth estimate for the coefficients). The objects arising in our examples are, nonetheless, {\em Borel summable} in all but finitely many ``singular'' directions.

\subsection{Applying TR to spectral curves of hypergeometric type} 
\label{section:applying-TR-to-HG}

In the rest of \S \ref{sec:quantumcurves}, we denote by $\Sigma$ one of the curves $\Sigma_{\bullet}$ for $\bullet \in \{{\rm HG}, {\rm dHG}, {\rm Kum}, {\rm Leg}, {\rm Bes}, {\rm Whi}, {\rm Web}, {\rm dBes}, {\rm Ai} \}$ in Table \ref{table:classical}. 
To apply TR to $\Sigma$, we should regard them as spectral curves in the sense of Definition \ref{def:spectral-curve-TR}. This was done in \cite[\S 2.3]{IKoT-II}, where an explicit rational parametrization of $\Sigma$ was given. 
That is, for each $\bullet$ in Table \ref{table:classical}, 
there exists a pair of rational functions 
$(x(z),y(z)) = (x_{\bullet}(z), y_{\bullet}(z))$ such that 
we have an isomorphism 
\begin{equation} \label{eq:meromorphic-parametrization}
\begin{array}{ccc}
{\mathcal C} \setminus {\mathcal P}
& \stackrel{\sim}{\longrightarrow} & \Sigma \\
\rotatebox{90}{$\in$} & & \rotatebox{90}{$\in$} \\
z & \longmapsto & (x(z), y(z))
\end{array}
\end{equation} 
of punctured Riemann surfaces, with the choice ${\mathcal C} = {\mathbb P}^1$ and the set ${\mathcal P}$ of poles of $x(z)$ and $y(z)$. 
Under the notations fixed in \S \ref{section:cycle-path}, the set ${\mathcal P}$ is of the form
\begin{equation}
{\mathcal P} = \{ p_s ~|~ s \in P_{\rm od} \} \cup \{ p_{s_{+}}, p_{s_{-}} ~|~ s \in P_{\rm ev} \} \subset {\mathcal C},
\end{equation}
where, for each $s \in P$, either $x(p_s) = s$ or $x(p_{s_\pm}) = s$ holds depending on the parity of pole order of $\varphi$ at $s$.  
The set ${\mathcal P}$ is in bijection with the set $\overline{\Sigma} \setminus \Sigma$, and hence, \eqref{eq:meromorphic-parametrization} can be extended to an isomorphism ${\mathcal C} \isom \overline{\Sigma}$ of compact Riemann surfaces. {We often use this isomorphism to identify the two in what follows.}

In the case of $\Sigma_{\rm HG}$, the above construction is realized as follows. (See \cite[\S 2.3]{IKoT-II} for the others examples.)

\begin{exa}
A rational parametrization \eqref{eq:meromorphic-parametrization}
of $\Sigma_{\rm HG}$ is given by the pair of explicit rational functions 
\begin{equation}
\label{eq:Gauss_parameterization}
\begin{cases}
\displaystyle
x(z) = x_{\rm HG}(z)
:= \frac{\sqrt{\Delta_{\rm HG}({\bm m})}}{4 {m_{\infty}}^2} 
\left( z + z^{-1} \right)
+ \frac{{m_{\infty}}^2 + {m_0}^2 
- {m_1}^2}{2 {m_{\infty}}^2}, \\[10pt]
\displaystyle
y(z) = y_{\rm HG}(z)
:= \frac{4 {m_{\infty}}^3 z^2 \left(z - z^{-1} \right)}
{\sqrt{\Delta_{\rm HG}({\bm m})} \, (z - p_{0_+})(z - p_{0_-})(z - p_{1_+})(z - p_{1_-})},
\end{cases}
\end{equation}
where $\Delta_{\rm HG}(\bm m)$ is given in \eqref{eq:HG-Delta} (which is non-zero under Assumption \ref{ass:genericity}), 
and the set 
${\mathcal P} = \{p_{0_+}, p_{0_-}, p_{1_+}, p_{1_-}, p_{\infty_+}, p_{\infty_-}  \}$
of poles consists of 
\begin{equation} 
p_{0_\pm} :=
- \frac{(m_{0} \pm m_\infty)^2 
- {m_1}^2}{\sqrt{\Delta_{\rm HG}({\bm m})}},
\quad
p_{1_\pm} := \frac{(m_{1} \pm m_\infty)^2 
- {m_0}^2}{\sqrt{\Delta_{\rm HG}({\bm m})}}, 
\quad
p_{\infty_+}=\infty, ~~ p_{\infty_-}=0.
\end{equation}
The points $p_{s_\pm} \in {\mathcal C}$ are the two preimages of $s \in \{0, 1, \infty \}$ by $x_{\rm HG}(z)$, and the labels are chosen so that the residue formula \eqref{eq:sign-convention-preimages} hold. The point $p_{s_\pm}$ is mapped to $s_\pm$ through the isomorphism ${\mathcal C} \simeq \overline{\Sigma}_{\rm HG}$. The set of ramification points is given by ${\mathcal R} = \{\pm 1 \} \subset {\mathcal C}$, and these two points are mapped to the turning points $b_{1}, b_{2} \in {\mathbb P}^1$ by $x_{\rm HG}(z)$.  The conjugate map is given by $\overline{z} = 1/z$, which corresponds to the covering involution of $\overline{\Sigma}_{\rm HG}$.  
\end{exa}

Together with the canonical choice 
\begin{equation} \label{eq:Bergman-P1}
B(z_1, z_2) = \frac{dz_1 dz_2}{(z_1 - z_2)^2}
\end{equation} 
of the bidifferential $B$ when ${\mathcal C} = {\mathbb P}^1$,
we obtain a spectral curve $({\mathcal C}, x, y, B)$ 
in the sense of Definition \ref{def:spectral-curve-TR}.
We will identify $\Sigma$ defined in \eqref{eq:spectral-curve-x-y}
(rather, its compactification $\overline{\Sigma}$), together with the projection on the first coordinate $\pi:\Sigma\rightarrow\mathbb{P}^1$,
with the spectral curve $({\mathcal C}, x, y, B)$ through the isomorphism and the canonical choice of $B$. 
Under the isomorphism, the set ${\mathcal R} \subset {\mathcal C}$ of ramification points 
is mapped to $\pi^{-1}(T) \subset \overline{\Sigma}$ bijectively. 

\subsection{Free energies of spectral curves of hypergeometric type and BPS structure}
We have now introduced all of the notions needed to recall the result of our previous work \cite{IK20}.
In \cite{IKoT-I, IKoT-II}, a relationship between the free energy of the spectral curves of hypergeometric type and the Voros coefficients of the associated quantum curves was found (see Remark \ref{rem:differenceeq} for the precise relation). As a corollary, there are explicit formulas of $g$th free energy $F_g = F_g^{\bullet}$, which we can write uniformly in terms of BPS structures thanks to the results in \cite{IK20} for all hypergeometric type examples:

\begin{thm}[{\cite{IKoT-II,IK20}}] \label{thm:Borel-sum-free-energy}
The $g$th free energy of the spectral curves of hypergeometric type are given explicitly as follows: 
\begin{align}\label{eq:freenergies1} F_0 \equiv& \sum_{\substack{ \gamma \in \Gamma \\ Z(\gamma)\in  \mathbb{H}_{}}} 
\Omega(\gamma)\, \dfrac{1}{2} 
\left( \frac{Z(\gamma)}{2 \pi i } \right)^2
\log \left( \frac{Z(\gamma)}{2 \pi i} \right) ,
\\ 
\label{eq:freenergies2} F_1 \equiv&  - \frac{1}{12} \sum_{\substack{ \gamma \in \Gamma \\ Z(\gamma)\in  \mathbb{H}_{}}} 
\Omega(\gamma)\, \log \left( \frac{Z(\gamma)}{2 \pi i} \right),
\\ 
\label{eq:freenergies3}  F_{g}=&\dfrac{B_{2g}}{2g(2g-2)}
\sum_{\substack{ \gamma \in \Gamma \\ Z(\gamma)\in \mathbb{H}_{}}}\Omega(\gamma) \left( 
{\dfrac{2 \pi i}{Z(\gamma)}} \right)^{2g-2} \qquad (g \ge 2).
\end{align}
Here the formula for $F_0$ (resp., $F_1$) is valid
up to quadratic polynomials of $m_i's$ 
(resp., additive constants),
$B_{k}$ is the $k$th Bernoulli number given by $B_k=B_k(0)$, and $\mathbb{H}$ is any half-plane whose boundary rays are not BPS.
\end{thm}

We note that the prefactor in \eqref{eq:freenergies3} given by the Bernoulli number has an interpretation as the Euler characteristic of the moduli space of Riemann surfaces (\cite{HZ, Penner}). 
The explicit expression in each example can be recovered from the BPS spectrum computed in \cite{IK20}, summarized in Table \ref{table:BPS-str-HG} (the Airy and degenerate Bessel cases are excluded since the free energy is trivial). These expressions are effectively used to compute the Borel sum of their generating series $F_{\rm TR} = \sum_{g\geq0}\hbar^{2g-2} F_g$.  Below, we will introduce the Voros coefficients of the corresponding quantum curves, and write an analogous explicit formula for them.

\subsection{Quantum curves of hypergeometric type}

It is known since \cite{GS12, Zhou12, MS, DM13, BE-16} etc.\,\,that the Eynard-Orantin correlators are closely related to the expansion coefficients of the WKB solution of a certain Schr\"odinger-type equation, called a {\em quantum curve}, whose classical limit gives the spectral curve.  Here we review the statement for the Gauss hypergeometric spectral curve $\Sigma_{\rm HG}$ (see \cite[\S 2.3]{IKoT-II} for other examples).


First we introduce the following notations for various subsets which exclude turning  and ramification points:
\begin{equation} \label{eq:Sigma-prime}
\Sigma' := \Sigma \setminus \pi^{-1}(T), \quad
\overline{\Sigma}' := \overline{\Sigma} \setminus \pi^{-1}(T), 
\quad 
{\mathcal C}' = {\mathcal C} \setminus {\mathcal R}, 
\quad 
{\mathcal P}' = {\mathcal P} \setminus {\mathcal R}.
\end{equation}
Let us take a divisor
\begin{equation}
D(z; {\bm \nu}) := [z] - \sum_{p \in {\mathcal P}'} \nu_p [p]
\end{equation}
on ${\mathcal C}'$ which depends on $z \in {\mathcal C}'$ and a tuple ${\bm \nu} = (\nu_p)_{p \in P'}$ of complex parameters, which we call {\em quantization parameters}, satisfying 
\begin{equation} \label{eq:relation-nu}
\sum_{p \in {\mathcal P}'} \nu_p = 1. 
\end{equation}
Let $W_{g,n}(z_1, \dots, z_n)$ be the correlator defined by TR from the spectral curve $({\mathcal C}, x, y, B)$ described in \S\ref{section:applying-TR-to-HG} and \cite{IKoT-II}. 
Let us define a ``primitive'' of $W_{g,n}$ by\footnote{
For a second kind meromorphic differential $\omega(z)$ on ${\mathcal C}$ 
which is holomorphic at $p \in {\mathcal P}'$, 
we shall define its integral with the divisor $D(z; {\bm \nu})$ 
for $z \in {\mathcal C} \setminus \{\text{poles of $\omega$} \}$ by
$\int_{\zeta \in D(z; {\bm \nu})} \omega(\zeta) 
= \sum_{p \in {\mathcal P}'} \nu_p \int^{\zeta = z}_{\zeta = p} \omega(\zeta).
$ 
The condition \eqref{eq:relation-nu} implies that 
$d \int_{\zeta \in D(z; {\bm \nu})} \omega(\zeta) = \omega(z)$.
We also generalize the definition for multidifferentials in a straightforward way. 
}
\begin{equation} \label{eq:primitive-F-gn}
F_{g,n}(z_1,\dots,z_n) := 
\int_{\zeta_1 \in D(z_1 ; {\bm \nu})} 
	\cdots \int_{\zeta_n \in D(z_n ; {\bm \nu})} 
	\left( W_{g, n}(\zeta_1, \cdots, \zeta_n) 
	- \delta_{g,0}\delta_{n,2} \frac{dx(\zeta_1) \, dx(\zeta_2)}
   {(x(\zeta_1) - x(\zeta_2))^2} \right)
\end{equation}
Restricting this to a diagonal set $z_1 = \cdots = z_n = z$, we have a function on  $\mathcal{C}'$. 
Taking their generating series, we introduce 
\begin{equation}
S_{\rm TR}({z}, \hbar) =\sum_{k \ge -1} \hbar^{k} S_{{\rm TR},k}({z})
:= 
\sum_{k \ge -1} \hbar^k
	\Biggl\{ \sum_{\substack{2g - 2 + n = k \\ g \geq 0, \, n \geq 1}}
		\frac{1}{n!} F_{g, n}({z}, \cdots, {z})
	\Biggr\} 
	\label{eq:Riccati-solution}
\end{equation}
as a formal (Laurent) series in $\hbar$ whose coefficients are functions\footnote{Although the $S_{\rm TR}$ need not be meromorphic, the derivative $dS_{\rm TR}$ is.} on $\mathcal{C}'$.

Roughly speaking, the framework of quantum curves claims that the exponential of $S_{\rm TR}({z}, \hbar)$ satisfies a Schr\"odinger-type ODE ({after combining this with the inverse map of the isomorphism ${\mathcal C}' \setminus \mathcal{P}' \stackrel{\sim}{\to} \Sigma'$ (which is a restriction of \eqref{eq:meromorphic-parametrization})}), which has the original spectral curve as its classical limit.  We will not give a general statement of quantum curves in \cite{BE-16, IKoT-I} here since it requires considerable notational preparation. Instead, we focus on the main example, the quantization of the Gauss hypergeometric curve $\Sigma_{\rm HG}$.

\begin{thm}[{\cite[\S 2.3.1]{IKoT-II}}] \label{thm:quantum-Gauss-curve}
Let $S_{\rm TR}^{\rm HG}({z},  \hbar)$ be the formal series defined from $\Sigma_{\rm HG}$ by \eqref{eq:Riccati-solution}, with the integration divisor  
\begin{align*}
D(z; {\bm \nu})
&= [z]
 - \nu_{0_+} [p_{0_+}] - \nu_{0_-} [p_{0_-}]
- \nu_{1_+} [p_{1_+}] - \nu_{1_-} [p_{1_-}]
- \nu_{\infty_+} [p_{\infty_+}] - \nu_{\infty_-} [p_{\infty_-}].
\end{align*}
Then, its exponential {(for any choice of local inverse $z(x)$)}
\begin{equation}
\label{eq:wave-function}
\psi^{\rm HG}(x, \hbar) = 
\exp\left( S_{\rm TR}^{\rm HG}({z(x)},  \hbar) \right)
\end{equation}
gives a formal solution of the following differential equation:
\begin{equation}
\label{eq:Gauss_eq(d/dx)}
{{\bm E}_{\rm HG}} 
~:~ \left( \hbar^2 \frac{d^2}{dx^2} + q(x,\hbar) \hbar \frac{d}{dx} + r(x,\hbar) \right) \psi = 0,
\end{equation}
where $q(x,\hbar) = q_0(x) + \hbar q_1(x)$, $r(x,\hbar) = r_0(x) + \hbar r_1(x) + \hbar^2 r_2(x)$ and
\begin{align*}
q_0(x) & = 0, \quad
q_1(x) = \frac{1 - \nu_{0_+} - \nu_{0_-}}{x}
		+ \frac{1 - \nu_{1_+} - \nu_{1_-}}{x - 1}, \\
r_0(x) &= - \frac{{m_{\infty}}^2 x^2 - ({m_{\infty}}^2 +
 {m_0}^2 - {m_1}^2)x + {m_0}^2}{x^2 (1 - x)^2}, \\
r_1(x) &=  - \frac{(\nu_{0_+} - \nu_{0_-}) m_0}{x^2 (x - 1)}
		+ \frac{(\nu_{1_+} - \nu_{1_-}) m_1}{x(x - 1)^2}
		+ \frac{(\nu_{\infty_+} - \nu_{\infty_-}) m_{\infty}}{x(x - 1)}, \\
r_2(x) &=  - \frac{\nu_{0_+} \nu_{0_-}}{x^2 (x - 1)}
		+ \frac{\nu_{1_+} \nu_{1_-}}{x(x - 1)^2}
		+ \frac{\nu_{\infty_+} \nu_{\infty_-}}{x(x - 1)}.
\end{align*}
\end{thm}

The equation ${\bm E}_{\rm HG}$ given in \eqref{eq:Gauss_eq(d/dx)} is a 2nd order linear ODE with three regular singular points at $0$, $1$ and $\infty$, and hence, it is equivalent to the Gauss hypergeometric differential equation (see {\cite[\S 2.3.1]{IKoT-II}}). We note that the Gauss hypergeometric curve \eqref{eq:Sigma-HG} is nothing but the classical limit $\hbar \to 0$ of the equation ${\bm E}_{\rm HG}$; that is, ${\bm E}_{\rm HG}$ is a quantization of the Gauss curve \eqref{eq:Sigma-HG}. We call ${\bm E}_{\rm HG}$ the {\em quantum Gauss curve} (for the divisor $D(z;{\bm \nu})$).

In \cite{IKoT-I, IKoT-II}, it was proved that this quantization procedure works for all spectral curves $\Sigma = \Sigma_{\bullet}$ shown in Table \ref{table:classical}, and the resulting quantum curves ${\bm E}$ are members of the confluent family of hypergeometric type differential equations (there is some overlap with the results of \cite{BE-16}). We refer to \cite[\S 2.3]{IKoT-II} for explicit descriptions of rational parametrizations of hypergeometric type curves, choice of divisors, and the associated quantum curves for the other eight cases. We call these equations ${\bm E}$ the {\em quantum curves of hypergeometric type} (associated to each example).

\begin{rem} \label{rem:WKB-sol}
It is well-known that a Schr\"odinger-type ODE with a small parameter $\hbar$ like \eqref{eq:Gauss_eq(d/dx)} has a formal solution of the form 
\begin{equation} 
\psi_{\rm WKB}(x,\hbar) = \exp(S_{\rm WKB}({z(x)},\hbar)), \quad
S_{\rm WKB}({z(x)},\hbar) = \sum_{k \ge -1} \hbar^{k} S_{{\rm WKB},k}({z(x)}), 
\end{equation}
called the {\em WKB solution}. For a construction of the WKB solution, we refer \cite[\S 2]{KT98} for example. The derivative $dS_{{\rm WKB},k}({z(x)})/dx$ of the coefficients in the WKB solutions are meromorphic functions on the compactified classical limit $\overline{\Sigma}$ which are determined by a certain recursion relation.  Theorem \ref{thm:quantum-Gauss-curve} claims that $S_{\rm WKB}$ coincides with the formal series $S_{\rm TR}$ in \eqref{eq:Riccati-solution} obtained from the topological recursion, up to additive constants at each order of $\hbar$. 
\end{rem}

\subsection{Voros coefficients of quantum curves of hypergeometric type}
\label{subsubsection:Voros-coefficient}

Here we recall the definition of the Voros symbols for the quantum curve $E$ constructed above. See \cite{Voros83, DDP93, IN14, IN15} for the definition for general Schr\"odinger-type ODEs, and their roles in the exact WKB analysis. 

\subsubsection{Definition of Voros coefficients}

We define two kinds of Voros coefficients; the first type are associated with homology classes $\gamma \in H_1(\widetilde{\Sigma}, {\mathbb Z})$, while the second type are associated with relative homology classes $\beta \in H_1(\overline{\Sigma}, D_\infty, {\mathbb Z})$. Following \cite{IN14}, we refer to these as \emph{cycles} and \emph{paths}, respectively. 
We will keep using notations like $\gamma$ for cycles in $H_1(\widetilde{\Sigma}, {\mathbb Z})$ and $\beta$ for paths in $H_1(\overline{\Sigma}, D_\infty, {\mathbb Z})$ throughout.

Consider the formal series $S_{\rm TR}({z},  \hbar)$
defined in \eqref{eq:Riccati-solution}, 
and consider its term-wise derivative 
$dS_{\rm TR}({z},  \hbar) = \sum_{k \ge -1} \hbar^k \,
dS_{{\rm TR},k}({z})$ $(= dS_{\rm WKB}({z},\hbar))$ 
with respect to ${z}$. 
We take its ``odd part'' $dS_{\rm TR}^{\rm odd}({z},\hbar)
= \sum_{k \ge -1} \hbar^{k} S_{{\rm TR}, k}^{\rm odd}({z})$ 
defined by\footnote{We use the term ``odd'' although in general there terms of any order in $\hbar$. $dS^{\rm odd}_{\rm TR}$ is the anti-invariant part of $dS_{\rm TR}$ with respect to the covering involution $\iota_\ast$.} 
\begin{equation}
dS_{\rm TR}^{\rm odd}({z}, \hbar) = 
\frac{dS_{\rm TR}({z},  \hbar) - \iota^\ast dS_{\rm TR}({z},  \hbar)}{2},
\end{equation}
where $\iota^\ast$ means the pullback by the covering involution $\iota$ of $\overline{\Sigma}$.
This is a formal series whose coefficients are meromorphic on $\overline{\Sigma}$ satisfying 
\begin{equation} \label{eq:anti-invariance-dSodd}
\iota^\ast dS_{\rm TR}^{\rm odd}({z}, \hbar) 
= - dS_{\rm TR}^{\rm odd}({z}, \hbar).
\end{equation}

Although the coefficients $dS_{{\rm TR}, k}^{\rm odd}$ have poles on $\pi^{-1}(T)$, the residue of $dS_{{\rm TR}, k}^{\rm odd}$ at those points vanish due to the property \eqref{eq:anti-invariance-dSodd}. Therefore, the integral of $dS_{{\rm TR}, k}^{\rm odd}$ along a cycle $\gamma \in H_1(\widetilde{\Sigma}, {\mathbb Z})$ is well-defined for any $k \ge -1$. On the other hand, among the coefficients,   $dS_{{\rm TR}, -1}^{\rm odd}$ and $dS_{{\rm TR}, 0}^{\rm odd}$ have poles on $\pi^{-1}(T) \cup {  D_\infty}$, while $dS_{{\rm TR}, k}^{\rm odd}$ for $k \ge 1$ have poles only at $\pi^{-1}(T)$. This is because the correlators $W_{g,n}$ with $2g-2+n \ge 1$ have poles only at the ramification points. Moreover, $W_{g,n}$ with $2g-2+n \ge 1$ have no residue at ramification points, and hence, the integral of $dS_{{\rm TR}, k}^{\rm odd}$ along a path $\beta \in H_1(\overline{\Sigma}, D_\infty, {\mathbb Z})$ is well-defined for $k \ge 1$. Thus, we define
\begin{equation}dS_{{\rm TR},\geq 1}^{
\rm odd}({z},\hbar):=
dS_{\rm TR}^{\rm odd}({z}, \hbar) - 
\dfrac{ dS_{{\rm TR}, -1}^{\rm odd}({z})}{\hbar} - dS_{{\rm TR}, 0}^{\rm odd}({z}).
\end{equation}
by subtracting off the $k=0,1$ terms. Then, we define the Voros coefficients for the quantum curves as follows.\bigskip

\begin{dfn}
~
\begin{itemize}
\item 
For any cycle $\gamma \in H_1(\widetilde{\Sigma}, {\mathbb Z})$, the {\em (cycle) Voros coefficient} of the quantum curve ${\bm E}$ for $\gamma$ is a formal series of $\hbar$ defined by the term-wise integral:
\begin{equation}  \label{eq:def-Voros-closed}
V_{\gamma}(\hbar)=V_{\gamma}({\bm m}, {\bm \nu}, \hbar)  
:= \oint_{\gamma} dS_{\rm TR}^{\rm odd}({z}, \hbar).
\end{equation}

\item 
For any $\beta \in H_1(\overline{\Sigma}, D_\infty, {\mathbb Z})$, the {\em (path) Voros coefficient} of the quantum curve ${\bm E}$ for $\beta$ is a formal series of $\hbar$ defined by the term-wise integral\footnote{ 
The Voros coefficients for relative cycles were denoted by $W_\beta$ in \cite{IN14, IN15}.
}:
\begin{equation}\label{eq:def-Voros-relative}
V_\beta(\hbar)=V_\beta({\bm m},{\bm \nu},\hbar):= \int_\beta dS_{{\rm TR},\geq 1}^{\rm odd}({z},\hbar)
\end{equation}
\item 
The exponentiated formal series $e^{V_{\gamma}}$ and $e^{V_{\beta}}$ are called the {\em Voros symbol} of ${\bm E}$ for $\gamma \in \Gamma$ and $\beta \in \Gamma^{*}$, respectively. We refer to these as \emph{cycle Voros symbols} and \emph{path Voros symbols}, respectively.
\end{itemize}
\end{dfn}

For our purpose, we will be mainly interested in the Voros coefficients for cycles and paths in our previously chosen sublattices; that is, for $\gamma \in \Gamma$ and $\beta \in \Gamma^{*}$. We will recall the explicit form of the Voros coefficients in these cases.

\begin{rem}
In \cite{IKoT-I, IKoT-II}, the path Voros coefficients were defined by using $dS_{\rm TR}$ instead of $dS_{\rm TR}^{\rm odd}$. Since the even part $dS_{\rm TR}^{\rm even} := dS_{\rm TR} - dS_{\rm TR}^{\rm odd}$ is invariant under the covering involution $\iota_\ast$, its integral along $\beta \in \Gamma^{*}$ vanishes. Hence, our path Voros coefficients \eqref{eq:def-Voros-relative} for $\beta \in \Gamma^{*}$ agree with those computed in \cite{IKoT-I, IKoT-II}.
\end{rem}

\subsubsection{Cycle Voros coefficients for $\gamma \in \Gamma$}
In what follows, we will use the same notation for cycles and relative cycles on the spectral curve defined in \S \ref{section:cycle-path}. Recall that the homology group $H_1(\widetilde{\Sigma}, {\mathbb Z})$ is generated by the residue cycles around the punctures $D_\infty$. In particular, any element $\gamma \in \Gamma$ can be decomposed as a sum of residue cycles around $s_{+}$ or $s_{-}$ for $s \in P_{\rm ev}$. 
Through the residue computation, we have the following. 

\begin{prop} \label{prop:closed-Voros-coeffs}
For any quantum curve of hypergeometric type 
and any $\gamma \in \Gamma$, we have
\begin{equation} \label{eq:closed-Voros-coeff}
V_{\gamma}( \hbar) = \frac{Z(\gamma)}{\hbar} - \pi i \nu(\gamma)
\end{equation}

\noindent with $\nu : H_1(\widetilde{\Sigma}, {\mathbb Z}) \to {\mathbb C}$ the linear map given on generators by $\nu(\gamma_{s_{\pm}}) = \pm {\nu_{s}}$, 
where we set $\nu_s := \nu_{p_{s_+}} - \nu_{p_{s_-}}$.
\end{prop} 

\begin{proof}
It is enough to show the formula for $\gamma_{s_\pm}$ with any $s \in P_{\rm ev}$.
Since $W_{g,n}$ with $2g-2+n \ge 1$ have poles only at ramification points, we can verify that the residue of $dS_{\rm TR}$ at $s_{\pm}$ only comes from those of $dS_{{\rm TR},-1}$ and $dS_{{\rm TR},0}$. More explicitly, we have 
\begin{equation}
\Res_{{z=p_{s_\pm}}} dS_{\rm TR}({z}, \hbar) = 
\frac{1}{\hbar} \Res_{z=p_{s_\pm}} W_{0,1}(z) 
+ \Res_{z=p_{s_\pm}}  \int_{\zeta \in D(z;{\bm \nu})} 
\left(W_{0,2}(z, \zeta) - \frac{dx(z) dx(\zeta)}{(x(z) - x(\zeta))^2} \right).
\end{equation}
The residue of the first term is given by
\begin{equation}
\Res_{z=p_{s_\pm}} W_{0,1}(z) = {\Res_{x=s_{\pm}} \sqrt{Q(x)}  dx} = \pm m_s
\end{equation}
due to the residue formula \eqref{eq:sign-convention-preimages}. 
The integrand in the second term is expressed as 
\begin{align}
& \int_{\zeta \in D(z;{\bm \nu})} 
\left(W_{0,2}(z, \zeta) - \frac{dx(z) dx(\zeta)}{(x(z) - x(\zeta))^2} \right) \notag \\
& \quad =  
\sum_{p \in {\mathcal P}'} \nu_p \int^{\zeta=z}_{\zeta=p} 
\left(W_{0,2}(z, \zeta) - \frac{dx(z) dx(\zeta)}{(x(z) - x(\zeta))^2} \right) \notag \\ 
& \quad = 
\sum_{p \in {\mathcal P}'} \nu_p 
\left( - \frac{1}{2} \frac{x''(z)}{x'(z)} - \frac{1}{z - p} + \frac{x'(z)}{x(z) - x(p)}  \right).
\label{eq:dS0-summand}
\end{align}
The summands for $p \ne p_{s_{\mp}}$ in \eqref{eq:dS0-summand} are holomorphic at $z = p_{s_\pm}$, 
but the summand for $p = p_{s_\mp}$ has a simple pole at $z = p_{s_\pm}$. 
Consequently, we have 
\begin{equation}
\Res_{z=p_{s_\pm}}  \int_{\zeta \in D(z;{\bm \nu})} 
\left(W_{0,2}(z, \zeta) - \frac{dx(z) dx(\zeta)}{(x(z) - x(\zeta))^2} \right) 
= \nu_{p_{s_\mp}}.
\end{equation}
{Since {$\Res_{z=p_{s_\pm}} dS_{\rm TR}^{\rm odd}(z,  \hbar) = 
\frac{1}{2}
(\Res_{z=p_{s_\pm}} dS_{\rm TR}(z,  \hbar) - 
\Res_{z=p_{s_\mp}} dS_{\rm TR}(z,  \hbar))$
}holds}, 
we have the desired formula for $\gamma_{s_\pm}$:
\begin{equation} \label{eq:closed-voros-at-s}
V_{\gamma_{s_\pm}}( \hbar) = \pm 2 \pi i
\, \left(\frac{m_s}{\hbar} - \frac{\nu_s}{2} \right).
\end{equation}
The general formula follows immediately from \eqref{eq:closed-voros-at-s}.
\end{proof}
\noindent In particular, we note $V_\gamma$ is just a finite series in $\hbar$ with two terms.

\subsubsection{Path Voros coefficients for $\beta \in \Gamma^{*}$}

As we mentioned in \S \ref{sec:doubling}, $\Gamma^{*}$ is generated by $\beta_s$, $ s \in P_{\rm ev}$, which is a class represented by a path from $s_{-}$ to $s_{+}$. The path Voros coefficients of the hypergeometric differential equations for those paths were studied in \cite{SS, Takei08, KoT11, Aoki-Tanda, KKT14, ATT, AIT, IKoT-I, IKoT-II} etc. Let us recall and rewrite the formula for their explicit expression, following \cite{IKoT-II, IK20}. 

\begin{exa}[Weber and Bessel]
Let us consider the two fundamental examples: 
\begin{align}
\Sigma^{\rm Web} &~:~ y^2 - \left( \frac{x^2}{4} - m_\infty \right) = 0, \quad m_\infty \ne 0. \\
\Sigma^{\rm Bes} &~:~ y^2 - \dfrac{x + 4m_0^2}{4x^2} = 0, \quad m_0 \ne 0.  
\end{align}
The corresponding quadratic differentials $Q_\bullet(x)dx^2$ have a unique even order pole $s$; thus 
${\bm m} = (m_s)$, and we have the Voros coefficient 
\begin{equation}
V_{\beta_s}(\hbar) = 
\sum_{k = 1}^{\infty} \hbar^{k} V_{\beta_s, k} 
\end{equation}
for the canonical relative homology class 
$\beta_s \in \Gamma^{\ast}_{\bullet}$ associated with $s$.
According to \cite[Theorem 3.1 (iv)]{IKoT-II}, the coefficients $V_{\beta_s, k} = V_{\beta_s, k}(m_s, \nu_s)$ are explicitly given as follows: 
\begin{equation} \label{eq:Voros-key-examples}
V^{\bullet}_{\beta_s, k} = 
\begin{cases}
\dfrac{B_{k+1}\bigl( \frac{1+\nu_\infty}{2} \bigr)}{k(k+1) }\cdot  \dfrac{1}{m_\infty^k} & 
\text{for $\bullet = {\rm Web}$ and $s = \infty$} \\[+1.em]
\dfrac{B_{k+1}(\nu_0) + B_{k+1}(1+\nu_0)}
{k(k+1)} \cdot \dfrac{-1}{(2m_0)^{k}} & 
\text{for $\bullet = {\rm Bes}$ and $s = 0$}.
\end{cases}
\end{equation}
Here $B_k(t)$ is the $k$-th Bernoulli polynomial defined through its generating series
\begin{equation}
\label{def:BernoulliPoly}
\frac{w e^{t w}}{e^w - 1} = \sum_{k = 0}^{\infty} B_k(t) \frac{w^k}{k!}.
\end{equation}
These explicit formulas are derived by finding and solving the difference equations (with respect to the $\hbar$-shift of parameters ${\bm m}$) satisfied by the Voros coefficients (see \cite{IKoT-II} and the above references for details).
\end{exa}

These two cases are of fundamental importance, since the Voros coefficients for other examples are obtained as superpositions. In fact, following our interpretation of the topological recursion free energies via BPS structures \cite{IK20} we may write the Voros coefficients computed in \cite{IKoT-II} uniformly as follows:

\begin{prop}\label{prop:voroscoeffs} 
For any quantum curve of hypergeometric type, the coefficient of $\hbar^{k}$ in the Voros coefficient for any relative homology class $\beta \in \Gamma^*$ is expressed as:

\begin{align}
\label{eq:voros-coeffs-general}
V_{{\beta},k} & 
=  \sum_{\substack{\gamma \in \Gamma \\ Z(\gamma) \in {\mathbb H}}}
  \Omega(\gamma) \, \langle {\beta}, \gamma \rangle \, \dfrac{{\mathcal B}_{k+1}(\gamma)}{k(k+1)} \,
\left( \frac{2 \pi i}{Z(\gamma)} \right)^{k},
\end{align}
where 
\begin{equation}
      {\mathcal B}_k(\gamma) = \begin{cases}  
    {B_{k}\Bigl( \frac{1+\nu(\gamma)}{2} \Bigr)},\quad& \Omega(\gamma)\neq -1  \\[+.8em]
    \dfrac{1}{2} 
    \left( {B_{k}\Bigl( \frac{\nu(\gamma)}{2} \Bigr) + B_{k}\Bigl( 1+\frac{\nu(\gamma)}{2} \Bigr)} \right)
    ,\quad& \Omega(\gamma)=-1 \end{cases}
\end{equation}
where $\mathbb{H}$ is any half-plane whose boundary rays are not BPS.
\end{prop}

\begin{proof}The symmetry $\Omega(\gamma)=\Omega(-\gamma)$ together with the identity $B_{k+1}(x) = (-1)^{k+1} B_{k+1}(1-x)$ for Bernoulli polynomials shows the expression \eqref{eq:voros-coeffs-general} is independent of the choice of half-plane. Taking the sign convention \eqref{eq:sign-of-pairings} into account, the claim is proved immediately by comparing with the expression of Voros coefficients obtained in \cite[Theorem 3.1 (iv)]{IKoT-II} (written explicitly in Appendix \ref{sec:explicit-formulas}), together with Table \ref{table:BPS-str-HG}.
\end{proof}

The expression \eqref{eq:voros-coeffs-general} is crucially important for our purpose. 
\subsection{Borel sum of Voros coefficients}
\label{section:Borel-sum-Voros}

Since the values of Bernoulli polynomials grow factorially as $k$ tends to infinity, the Voros coefficients for relative homology classes are divergent series of $\hbar$. 
However, we can see that the Voros coefficients in our examples are {\em Borel summable} as a formal series of $\hbar$ along all but finitely many rays.  
Here we briefly recall the notion of the Borel sum for our specific examples --- for a general theory of the Borel summation, we refer \cite{Costin, Sauzin}. 
We note that our computations presented here were essentially done in previous works \cite{Takei08, KoT11, Aoki-Tanda, ATT, AIT} etc.\ which discussed the Voros coefficients for hypergeometric-type differential equations in slightly different presentations (with or without quantization parameters ${\bm \nu}$).

First, for the Voros coefficient $V_{\beta_s}(\hbar) = \sum_{k \ge 1} \hbar^k V_{\beta_s, k}$, let us define its {\em Borel transform} $\widehat{V}_{\beta_s}(\zeta)$ as follows:
\begin{equation}
\widehat{V}_{\beta_s}(  \zeta) := \sum_{k = 1}^{\infty} \frac{V_{\beta_s, k} }{(k-1)!} \zeta^{k-1}
\end{equation}
Namely, $\widehat{V}_{\beta_s}(  \zeta)$ is the term-wise inverse Laplace transform, where $\zeta$ is the variable which is Laplace dual to $1/\hbar$. 
Although the Voros coefficient is divergent series of $\hbar$, $\widehat{V}_{\beta_s}$ has a finite radius of convergence. For a given ray of the form $\ell =  e^{i \vartheta} {\mathbb R}_{> 0}$, the Voros coefficient $V_{\beta_s}$ is said to be {\em Borel summable along $\ell$} if the Laplace integral 
\begin{equation} \label{eq:Borel-sum}
{\mathcal S}_\ell V_{\beta_s}( \hbar) := 
\int^{\infty e^{i \vartheta}}_{0} 
 \widehat{V}_{\beta_s}(  \zeta)e^{- \zeta/ \hbar}  \,d\zeta
\end{equation}
along $\ell$ converges for $\hbar \in {\mathbb H}_{\ell} \cap \{\hbar \in {\mathbb C}^\ast ~|~ |\hbar| \ll 1  \}$\footnote{
The rays for the Borel summation is regarded as a half-line on the $\zeta$-plane (Borel plane), and hence, we should distinguish them from rays in $\hbar$-plane. However, to save notations, we identify them and use the same symbol $\ell$ for both of the rays. 
}. 
This requires that $\widehat{V}_{\beta_s}(\zeta)$ has an analytic continuation to a sectorial neighborhood of the integration contour $e^{i \vartheta}{\mathbb R}_{\ge 0}$, and growth at most exponentially when $|\zeta| \to \infty$. The resulting integral \eqref{eq:Borel-sum} is called the {\em Borel sum} of the Voros coefficient, if it is well-defined. 
Watson's lemma (e.g., \cite[Chapter4]{BH}) implies that the asymptotic expansion of the Borel sum recovers the original divergent series when $\hbar$ approaches $0$:
\begin{equation}
{\mathcal S}_\ell V_{\beta_s}(  \hbar) 
\sim V_{\beta_s}(  \hbar), \qquad \hbar \to 0, \, \hbar \in {\mathbb H}_{\ell}.
\end{equation}

For the Borel sum of the Voros coefficients of our examples, the following is known: 
\begin{prop} \label{prop:Borel-sum-Voros-relative}
For any example of hypergeometric type and any $s \in P_{{\rm ev}}$, the Voros coefficient $V_{\beta_s}( \hbar)$ 
with the expression \eqref{eq:voros-coeffs-general} and the corresponding Voros symbol are Borel summable along any ray which is not BPS in the corresponding BPS structure in Table \ref{table:BPS-str-HG}. 
Moreover, the Borel sum ${\mathcal S}_\ell e^{V_{\beta_s}(\hbar)}$ of the Voros symbol can be extended as a meromorphic function of $\hbar$ on the whole half-plane ${\mathbb H}_{\ell}$.

\end{prop}

\begin{proof}
The statement was essentially shown (without the language of BPS structures) in the previous works \cite{Takei08, KoT11, Aoki-Tanda, KKT14, ATT, AIT}, but we include a proof here for the convenience of reader. 

Note that, since the Voros coefficients in our examples are obtained by superposing $V^{\rm Web}_{\beta_\infty}$ and $V^{\rm Bes}_{\beta_0}$, it suffices to prove the statement for the two examples. 

For the Weber case, using the expression \eqref{eq:Voros-key-examples} of the coefficients, 
we have
\begin{equation} \label{eq:Weber-Voros-Borel}
\widehat{V}^{{\rm Web}_{\beta_\infty}}(\zeta)
= \frac{1}{\zeta} \, 
\left.\left(\dfrac{e^{t w}}{e^w-1}- \frac{1}{w} - t + \frac{1}{2} \right)\right
\vert_{w=\frac{\zeta}{m_\infty},\, t=\frac{1 + \nu_\infty}{2}}.
\end{equation}
Hence, $V^{\rm Web}_{\beta_\infty}$ is not Borel summable along the rays $\pm 2 \pi i m_\infty {\mathbb R}_{\ge 0}$, which are exactly the BPS rays for the BPS structure constructed from  $\Sigma_{\rm Web}$, since the Borel transform $\widehat{V}^{\rm Web}_{\beta_\infty}$ has a sequence of simple poles on the set $2 \pi i m_\infty  \mathbb{Z}_{\neq 0}$.  
The Borel sum is well-defined for any other ray $\ell = e^{i \vartheta} {\mathbb R}_{> 0}$, and explicitly computed as follows:
\begin{equation}
{\mathcal S}_\ell V^{\rm Web}_{\beta_\infty}( \hbar) = 
\begin{cases}
\displaystyle 
\log \Lambda\left(\frac{m_\infty}{\hbar}, \frac{1 - \nu_\infty}{2} \right), \quad & 
\text{$\arg (2 \pi i m_\infty) - \pi < \vartheta < \arg (2 \pi i m_\infty)$}, 
\\[+1.em]
\displaystyle 
- \log \Lambda\left(- \frac{m_\infty}{\hbar}, \frac{1 + \nu_\infty}{2} \right), & 
\text{$\arg (2 \pi i m_\infty) < \vartheta < \arg (2 \pi i m_\infty) + \pi$},
\end{cases}
\label{eq:Borel-sum-Weber-Voros}
\end{equation}
where $\Lambda(z,\eta)$ is the explicit function given in \eqref{eq:Lambda-function-original}. The formula \eqref{eq:Borel-sum-Weber-Voros} is derived with the aid of 
the following Binet's integral formula for the logarithm of the $\Gamma$-function
(see \cite[p. 249]{WW} for example):  
\[
\log \frac{\Gamma(z)}{\sqrt{2 \pi}} = \left(z-\frac{1}{2} \right) \log z - z 
+ \int^{\infty}_{0} e^{- z w} \left(\frac{1}{e^{w} - 1} - \frac{1}{w} + \frac{1}{2} \right) \frac{dw}{w}
\qquad\quad~ ({\rm Re}\, z > 0).
\]
Since the Borel summation commutes with taking exponential (c.f., \cite[\S 13]{Sauzin}), we have 
\begin{equation}
{\mathcal S}_{\ell} e^{V^{\rm Web}_{\beta_\infty}( \hbar)} = 
\begin{cases}
\displaystyle 
\Lambda\left(\frac{m_\infty}{\hbar}, \frac{1 - \nu_\infty}{2} \right), & 
\text{$\arg (2 \pi i m_\infty) - \pi < \vartheta < \arg (2 \pi i m_\infty)$}, 
\\[+1.em]
\displaystyle 
\Lambda\left(- \frac{m_\infty}{\hbar}, \frac{1 + \nu_\infty}{2} \right)^{-1}, & 
\text{$\arg (2 \pi i m_\infty) < \vartheta < \arg (2 \pi i m_\infty) + \pi$}.
\end{cases}
\label{eq:Borel-sum-Weber-Voros-symbol}
\end{equation}

Thus we have proved the claim in the Weber case. 

We can perform a similar computation in the Bessel case. We can verify that $V^{\rm Bes}_{\beta_0}$ is Borel summable except for the BPS rays $\pm 2 \pi i m_0 {\mathbb R}_{>0}$, and obtain:
\begin{align}
 {\mathcal S}_\ell e^{V^{\rm Bes}_{\beta_0}( \hbar)} = 
\begin{cases}
\displaystyle 
\Lambda\left(\frac{2m_0}{\hbar}, 1 - \nu_0 \right)^{-1}
\Lambda\left(\frac{2m_0}{\hbar}, -\nu_0 \right)^{-1}, & 
\text{$\arg (2 \pi i m_0) - \pi < \vartheta < \arg (2 \pi i m_0)$}, 
\\[+1.5em]
\displaystyle 
\Lambda\left(- \frac{2m_0}{\hbar}, \nu_0 \right) \,
\Lambda\left(- \frac{2m_0}{\hbar}, 1+\nu_0 \right),
& 
\text{$\arg (2 \pi i m_0) < \vartheta < \arg (2 \pi i m_0) + \pi$}.
\end{cases}
\label{eq:Borel-sum-Bessel-Voros}
\end{align}

The other cases follow from the formula \eqref{eq:voros-coeffs-general}
\end{proof}

On the other hand, since the Voros coefficient $V_\gamma = V_\gamma^\bullet$ for $\gamma \in \Gamma$ is a finite series of $\hbar$ (c.f., Proposition \ref{prop:closed-Voros-coeffs}), its Borel sum along any ray $\ell$ is identical to itself in all examples:
\begin{equation} \label{eq:Borel-sum-of-closed-Voros-coefficient}
{\mathcal S}_\ell V_{\gamma}( \hbar) 
= V_{\gamma}(  \hbar) 
\qquad  (\gamma \in \Gamma).
\end{equation}

\subsection{Stokes phenomenon for Voros symbols}
\label{section:Stokes-Voros}

The equality \eqref{eq:Borel-sum-of-closed-Voros-coefficient} shows that the cycle Voros coefficients do not exhibit any jumping behaviour when we vary the ray $\ell$ for the Borel summation. On the other hand, as we have seen in the proof of Proposition \ref{prop:Borel-sum-Voros-relative}, the Borel sum of the path Voros coefficients have a discontinuity across the non-Borel summable directions when we vary $\ell$. The jump of the Borel sum can be understood as the {\em Stokes phenomenon} for the divergent series $V_{\beta_s}$ around $\hbar = 0$. The discontinuity follows from the presence of singular points of the Borel transform $\widehat{V}_{\beta_s}( \zeta)$ in the $\zeta$-plane. We have the following behaviour for the two fundamental cases of Weber and Bessel:

\begin{thm}[c.f., {\cite{Takei08, AIT}}] \label{thm:Stokes-jump-Voros}
Let $\bullet \in \{{\rm Web}, {\rm Bes} \}$, and take a BPS ray $\ell_{\rm BPS} = \pm 2 \pi i m_s {\mathbb R}_{>0}$ of the corresponding BPS structure. 
We denote by $\ell_1$ and $\ell_2$ the rays obtained by perturbing $\ell_{\rm BPS}$ in counter-clockwise and clockwise direction, respectively (i.e., $\ell_k = e^{i \vartheta_k} {\mathbb R}_{> 0}$ with $\vartheta_{1} = \arg(\pm 2 \pi i m_s) + \delta$ and $\vartheta_{2} = \arg(\pm 2 \pi i m_s) - \delta$ with a sufficiently small $\delta > 0$). 
Then, we have the following relation (as meromorphic functions) between the Borel-resummed Voros symbols on ${\mathbb H}_{\ell_1} \cap {\mathbb H}_{\ell_2}$:
{
\begin{align}
& {\mathcal S}_{\ell_2}e^{{V}^{\bullet}_{\beta_s}(  \hbar)} 
= 
\begin{cases}
\displaystyle
{\mathcal S}_{\ell_1} e^{V^{\rm Web}_{\beta_\infty}( \hbar)} 
\Bigl(1 + 
e^{-V^{\rm Web}_{\gamma_{\infty_\pm}}( \hbar)}
\Bigr)^{- (\gamma_{\infty_\pm}, \beta_\infty)}
& \text{for $\bullet = {\rm Web}$ and $s = \infty$},
\\[+1.em]
\displaystyle
{\mathcal S}_{\ell_1} e^{V^{\rm Bes}_{\beta_0}( \hbar)} 
\Bigl( 1 - 
e^{- V^{\rm Bes}_{\gamma_{0_\pm}- \gamma_{0_\mp}}(  \hbar)}
\Bigr)^{(\gamma_{0_\pm} - \gamma_{0_\mp}, \beta_0)}
& \text{for $\bullet = {\rm Bes}$ and $s = 0$}.
\end{cases} 
\label{eq:Stokes-jump-Voros}
\end{align}
}
\end{thm}

\begin{proof}
We only give a proof for the Weber case based on the idea of \cite{Takei08} where the case $\nu_\infty = 0$ was discussed. 
It follows from the expression \eqref{eq:Weber-Voros-Borel} that the Borel transform has simple poles only on BPS rays. 
Through the residue calculus, we have 
\begin{align} \label{eq:alien-calculas}
& {\mathcal S}_{\ell_2} V^{\rm Web}_{\beta_\infty}(  \hbar) - 
{\mathcal S}_{\ell_1} V^{\rm Web}_{\beta_\infty}(  \hbar) \notag \\
& \quad = 2 \pi i \sum_{n=1}^{\infty} \Res_{\zeta = \pm 2 \pi i m_\infty n} 
\left( e^{- \frac{\zeta}{\hbar}} \widehat{V}^{\rm Web}_{\beta_\infty}(\zeta)
\right) d\zeta  \notag \\
& \quad = \pm \sum_{n=1}^{\infty} e^{\mp \frac{2 \pi i m_\infty n}{\hbar}} 
\frac{(-1)^n e^{\pm \pi i \nu_\infty n}}{n} 
= \mp \log \Bigl( 1 + e^{- V^{\rm Web}_{\gamma_{\infty_\pm}}} \Bigr)
\end{align}
Since the overall factor $\mp 1$ is identical to minus the intersection pairing \eqref{eq:intersetion-gamma-beta}, we obtain \eqref{eq:Stokes-jump-Voros} for 
the Weber case by taking the exponential. The same method can be applied to Bessel.
\end{proof}


The jump formula \eqref{eq:Stokes-jump-Voros} for the Borel-resummed Voros symbols is closely related to the one imposed in the BPS Riemann-Hilbert problem in \S \ref{section:BPS-RHP}.
Indeed, when $\nu_\infty = 1$, the Borel resumed Voros symbol 
${\mathcal S}_\vartheta e^{V^{\rm Web}}$ for the Weber curve, 
which is expressed by the $\Lambda$ function \eqref{eq:Lambda-function-original},
coincides with the solution of the doubled $A_1$ BPS Riemann-Hilbert problem discussed 
in \cite[\S 5]{Bri19} after identifying the mass parameter with the central charge (recall that $\Lambda(w,0) = \Lambda(w,1) = \Lambda(w)$).
For generic $\nu_\infty$, the Voros symbol also gives the solution to 
the ``meromorphic'' BPS Riemann-Hilbert problem discussed by Barbieri
\cite[\S 4]{Bar}. 
{  
We will make this more precise in \S \ref{subsection:solving-BPS-RHP-by-Voros}. 
}

\begin{rem}
In the language of the {\em resurgent analysis} (c.f., \cite{Sauzin}), 
the Voros coefficients for our examples are {\em simple resurgent series}, 
and the computation \eqref{eq:alien-calculas} in the above proof is closely related 
to the so-called {\em alien calculus}. See \cite{Takei08, KoT11, KKKT11, KKT14, AIT} for more discussion. 
Note that \eqref{eq:Stokes-jump-Voros} is a consequence of the Euler reflection formula of the gamma function, as shown in \cite{Bri19, Bar}. The above computation gives a resurgence-theoretic proof of the formula. 
\end{rem}



\subsection{The Voros potential}
\label{sec:vorospotential} Since Voros coefficients are integrals over homology classes, we may consider $V_{\cdot,k}$ as a map $V_{\cdot,k}: (\Gamma\oplus \Gamma^*)\otimes \mathbb{C} \rightarrow \mathbb{C}$, where we complexify in the obvious way. Recall that we are interested in the variation of BPS structures $(\Gamma_{\text{\rm \DJ}},Z_{\text{\rm \DJ}},\Omega_{\text{\rm \DJ}})$, over $T^*M_\bullet$ which by Corollary \ref{cor:mini} is in fact miniversal. The miniversality of the original $(\Gamma,Z,\Omega)$ provides an isomorphism:
\begin{equation}
    D\pi: T_{\bm m}M_\bullet \rightarrow \Gamma^*\otimes \mathbb{C}\,\,\left(\simeq\Gamma^\vee\otimes \mathbb{C}\right)
\end{equation}
and dually we may identify $T_{\bm m}^*M$ and $\Gamma\otimes\mathbb{C}$. In particular, this means the ``interesting part'' (namely, the path Voros symbols) $V_{\cdot, k}: \Gamma^\vee\otimes\mathbb{C} \rightarrow \mathbb{C}$ (abusing notation) of the maps $V_{\cdot,k}$ may be interpreted as a one-form on $M_\bullet$.

Fixing generators $\{\gamma_s-\iota_*\gamma_s\}$ for $\Gamma\otimes\mathbb{C}$ and taking as local coordinates $m_s=Z(
\gamma_s)/4\pi i$ on $M_\bullet$ we may write these one-forms as
\begin{equation}
    \omega_k := \sum_{s\in P_{\rm ev}}{2 \pi i\, V_{\beta_s,k}dm_s},\qquad (k\geq 1) 
\end{equation}
It turns out these forms are closed. Thus, there is a formal series in $\hbar$ of functions $\phi = \sum_{k\geq1}\phi_k \hbar^k$ such that
\begin{equation}
\label{eq:vorpot}
    \omega =  2\pi i \, d_{M_\bullet}\phi
\end{equation}
where $d_{M_\bullet}$ denotes the exterior derivative on $M_\bullet$ only, and $\omega:=\sum_{k\geq1}\omega_k\hbar^k$ is a formal series of one-forms. Each term in $\phi$ is a function on the space $M_\bullet$ depending on the quantization parameters ${\bm \nu}$. More precisely:

\begin{prop}
\label{prop:vorospot}
The Voros coefficients of quantum curves of hypergeometric type are generated by a potential $2\pi i \phi$, where $\phi = \sum_{k\geq1}\phi_k\hbar^k$ is called the \emph{Voros potential}. That is, the forms $\omega_k$, $k\geq 1$ are closed, and the Voros potential is given by
\begin{align}
\phi_k & 
=\begin{cases} \displaystyle \frac{1}{2} \sum_{\substack{ \gamma \in \Gamma \\ Z(\gamma)\in  \mathbb{H}}}  {{\mathcal B}_{2}(\gamma)}  \cdot \Omega(\gamma) 
\log{\left(\frac{Z(\gamma)}{2\pi i} \right)}, \quad k=1 \\ 
\displaystyle
\dfrac{-1}{(k-1)k(k+1)}\left( \sum_{\substack{\gamma \in \Gamma \\ Z(\gamma) \in {\mathbb H}}}
  {{\mathcal B}_{k+1}(\gamma)}  \cdot \Omega(\gamma)  
\left( \frac{2 \pi i}{Z(\gamma)} \right)^{k-1} \right), \quad k\geq 2 
\end{cases}
\end{align}
where $\mathbb{H}$ is any half-plane whose boundary rays are not BPS. Whenever $\ell$ is not a BPS ray, $\phi$ is Borel summable along $\ell$, and $2\pi i \, \mathcal{S}_\ell\phi$ is a potential for $\mathcal{S}_\ell \omega$.
\end{prop}
\begin{proof}
To determine closedness and compute $\phi_k$, we calculate directly from \eqref{eq:vorpot}, case-by-case using the explicit formula \eqref{eq:voros-coeffs-general} or Appendix \ref{sec:explicit-formulas}. In the Gauss hypergeometric and degenerate hypergeometric cases, we use the identity for Bernoulli polynomials $(-1)^{k+1}B_k(x)+B_k(1-x)=0$, ($k\geq0$).
The Borel summability follows from the known asymptotic expansion (see e.g. \cite{BBS}) of the Barnes double gamma function in a sufficiently large sector.
\end{proof}

We will provide an interpretation of the Voros potential in the next section.

\section{Solution to BPS Riemann-Hilbert problem and the $\tau$-function}
\label{sec:final}
We may now formulate and prove our main results.

\subsection{BPS indices and the Stokes structure of the Voros coefficients}
\label{subsection:solving-BPS-RHP-by-Voros}

In the work of Gaiotto-Moore-Neitzke \cite{GMN09}, the BPS indices $\{ \Omega(\gamma) \}_{\gamma \in \Gamma}$ appeared in a formula which describes the jump of the Fock-Goncharov cluster coordinate on the moduli space of the framed local systems. In fact, it is easy to see this jump formula is in fact nothing but the required property \eqref{eq:RHP-jump-explicit} in the BPS Riemann-Hilbert problem. 
The purpose of this subsection is to solve the BPS Riemann-Hilbert problem, and thus show that our definition \eqref{eq:def-of-BPS-indices} of BPS indices proposed in \cite{IK20} is consistent with this point of view if we identify the Fock-Goncharov coordinates with the Borel-resummed Voros symbols (c.f., \cite{I16, All19, Kuw20}). 

We consider the almost-doubled BPS structure $(\Gamma_{\text{\DJ}}, Z_{\text{\DJ}}, \Omega_{\text{\DJ}})$ with the almost-doubled lattice 
$\Gamma_{\text{\DJ}} = \Gamma \oplus \Gamma^*$. 
For convenience, in what follows we take $Z^\vee = 0$ so when we consider the variation of our almost-doubled BPS structures, they are regarded as a family over $M$ (i.e., the zero section of $T^\ast M$).

\begin{dfn} Let $(\Gamma_{\text{\DJ}},Z_{\text{\DJ}},\Omega_{\text{\DJ}})$ denote any of the almost-doubled BPS structures of hypergeometric type.  We define a meromorphic map, the \emph{Voros solution}, 
\begin{equation}
    X_\ell^{{\rm Vor}}: \mathbb{H}_\ell \rightarrow \mathbb{T}_{-,{\text{\DJ}}}
\end{equation}
as follows. For any non-BPS ray $\ell$, set\footnote{
Indeed, when $\mu=(\gamma,0)$ for a cycle $\gamma \in \Gamma$, the Borel summation in \eqref{eq:X-gamma-as-Voros-symbols} is not necessary because of \eqref{eq:Borel-sum-of-closed-Voros-coefficient}.} 
\begin{equation} \label{eq:X-gamma-as-Voros-symbols}
X^{{\rm Vor}}_{\ell, \gamma_{}}(\hbar) := 
\displaystyle \sigma(\gamma)   \cdot 
{\mathcal S}_{\ell}  e^{- V_{\gamma_{}}(\hbar)}
\end{equation}
for any cycles $\gamma \in \Gamma$, and 
\begin{equation} \label{eq:X-beta-as-Voros-symbol}
X^{{\rm Vor}}_{\ell, \beta}( \hbar) := 
\sigma(\beta)\cdot{\mathcal S}_{\ell} e^{V_{\beta}( \hbar)}
\end{equation}
for any paths $\beta\in \Gamma^{*}$. 
\end{dfn}

Here, $\sigma : \Gamma_{\text{\DJ}} \to \{ \pm 1 \}$ is the quadratic refinement characterized by the following property: 
\begin{equation}
\label{eq:quadref}
\sigma_{}(\gamma_{\rm BPS})=\begin{cases}
-1, \qquad \Omega(\gamma_{\rm BPS}) \neq -1 \\
+1, \qquad \Omega(\gamma_{\rm BPS})=-1
\end{cases}
\end{equation}
and $\sigma(\beta) = + 1$ for $\beta \in \Gamma^\ast$. Observe that some choice of quadratic refinement was necessary in order to convert the $\mathbb{T}_{{\text{\rm \DJ}},+}$-valued $(e^{-V_{\gamma}},e^{V_{\beta}})$ into something $\mathbb{T}_{{\text{\rm \DJ}},-}$-valued. It is shown in \cite[Lemma 3.2]{BS13} that an appropriate subset of BPS cycles generates the lattice $\Gamma$, and it is not difficult to see that such a quadratic refinement exists. The collection $\{ X^{\rm Vor}_{\ell, \mu} \}_{\mu  \in \Gamma_{\text{\rm \DJ}}}$ gives a map $X^{\rm Vor}_{\ell}$ from ${\mathbb H}_{\ell}$  to the almost-doubled twisted torus ${\mathbb T}_{-,{\text{\rm \DJ}}}$.

We justify the terminology ``solution'' in the remainder of the section. First we show the jumping behaviour:

\begin{prop} ~ \label{prop:jump-of-Voros-symbols-in-BPS-form}
Let $({\Gamma}_{\text{\rm \DJ}}, Z_{\text{\rm \DJ}}, \Omega_{\text{\rm \DJ}})$ be an almost-doubled BPS structure of hypergeometric type, and $\Delta \subset {\mathbb C}^\ast$ an acute sector whose boundary rays $\ell_1, \ell_2$, taken in the clockwise order, are non-BPS rays. Then we have the following relation on $\mathbb{H}_{\ell_1} \cap \mathbb{H}_{\ell_2}$:
\begin{equation} \label{eq:RHP-jump-in-HG}
X^{{\rm Vor}}_{\ell_2, \mu}(\hbar) = 
X^{{\rm Vor}}_{\ell_1, \mu}( \hbar) \,
\prod_{\substack{\eta \in \Gamma_{\text{\rm \DJ}} \\ Z_{\text{\rm \DJ}}(\eta) \in \Delta}} 
{(}1 - X^{{\rm Vor}}_{\ell_1, \eta}(\hbar){)}^{\Omega(\eta) \langle \eta, \mu \rangle} 
\quad (\mu \in \Gamma_{\text{\rm \DJ}}).
\end{equation}
\end{prop}

\begin{proof}

The BPS cycles and BPS indices for hypergeometric type spectral curves are summarized in Table \ref{table:BPS-str-HG}. Since our BPS structure is finite, it suffices to prove the claim when $\Delta$ contains only one BPS ray. 
Also, it is enough to verify \eqref{eq:RHP-jump-in-HG} for the generators of the almost doubled lattice; that is, for $\mu = (\gamma,0)$ with $\gamma \in \Gamma$ and for $\mu = (0,\beta)$ with $\beta \in \Gamma^{*}$.

We first note that the cycle Voros symbols, for homology classes $\mu = (\gamma,0)$ with $\gamma \in \Gamma$, do not jump as we have seen in \S \ref{section:Stokes-Voros}. hence, they automatically satisfy the formula \eqref{eq:RHP-jump-in-HG} since $\langle\gamma,\gamma'\rangle=0$ for all $\gamma'\in\Gamma$, but $\Omega(\mu)=0$ for $\mu = (0,\beta)$. 

Since $\Gamma^{*}$ is generated by $\beta_{s}$ for $s \in P_{\rm ev}$, the remaining task is to verify \eqref{eq:RHP-jump-in-HG} for $\mu = (0, \beta_s)$. It is easily see that this follows from the expression \eqref{eq:voros-coeffs-general} of the Voros coefficients for those paths and the jump property \eqref{eq:Stokes-jump-Voros} for the Weber and Bessel case. 
Indeed, taking $\Omega_{\rm Web}(\gamma_{\infty_\pm}) = 1$, $\Omega_{\rm Bes}(\gamma_{0_\pm} - \gamma_{0_\mp}) = -1$ (c.f., Table \ref{table:BPS-str-HG}) and the sign convention \eqref{eq:sign-of-pairings} into account, the formula \eqref{eq:Stokes-jump-Voros} can be written as 
\begin{equation}
{\mathcal S}_{\ell_2}e^{{V}^{\bullet}_{\beta_s}(\hbar)} 
= 
{\mathcal S}_{\ell_1} e^{V^{\bullet}_{\beta_s}( \hbar)} 
\Bigl(1 + e^{-V^{\bullet}_{\gamma_{\rm BPS}}( \hbar)}
\Bigr)^{\Omega(\gamma_{\rm BPS}) \langle \gamma_{\rm BPS}, \beta_s \rangle},
\end{equation}
where $\gamma_{\rm BPS}$ is the unique BPS cycles whose central charge lies on $\ell_{\rm BPS} = \pm 2 \pi i m_s \,  {\mathbb R}_{>0}$ (i.e., $\gamma_{\rm BPS} = \gamma_{\infty_{\pm}}$ for $\bullet = {\rm Web}$, and $\gamma_{\rm BPS} = \gamma_{0_{\pm}} - \gamma_{0_{\mp}}$ for $\bullet = {\rm Bes}$). Thus we have verified \eqref{eq:RHP-jump-in-HG} for $\bullet = {\rm Web}$ and ${\rm Bes}$. Since the path Voros coefficients in the general case \eqref{eq:voros-coeffs-general} is a superposition of those of Weber and Bessel, we can conclude that \eqref{eq:RHP-jump-in-HG} is valid for all other cases as well.
\end{proof}

\begin{rem}From the explicit expressions \eqref{eq:voros-coeffs-general} for the Voros coefficients, we have the explicit formula of their Borel sum for any non-BPS ray $\ell$ and any $\mu \in \Gamma_{\text{\rm \DJ}}$:
\begin{align}
\label{eq:vorosexplicit}
X^{\rm Vor}_{\ell, \mu}(\hbar) & 
= e^{- {Z_{\text{\rm \DJ}}(\mu)}/{\hbar}} \, \xi_{{\text{\DJ}},\bm \nu }(\mu)  
\prod_{\substack{\gamma \in \Gamma \\ Z(\gamma) \in i {\mathbb H}_{\ell} \\ \Omega(\gamma) \ne -1}} 
\Lambda \left( \frac{Z(\gamma)}{2 \pi i \hbar} 
, \frac{1-\nu(\gamma)}{2}\right)^{\Omega(\gamma) \langle \mu, \gamma \rangle} 
\notag \\
& \qquad \times
\prod_{\substack{\gamma \in \Gamma \\ Z(\gamma) \in i {\mathbb H}_{\ell} \\ \Omega(\gamma) = -1}} 
\Lambda \left( \frac{Z(\gamma)}{2 \pi i \hbar}
,1-\frac{\nu(\gamma)}{2} \right)^{\Omega(\gamma) \frac{\langle \mu, \gamma \rangle}{2}} \Lambda \left( \frac{Z(\gamma)}{2 \pi i \hbar}, -\frac{\nu(\gamma)}{2}\right)^{\Omega(\gamma) \frac{\langle \mu, \gamma \rangle}{2}}.
\end{align} 
The factor $\xi_{\text{\rm \DJ}, \bm \nu} \in {\mathbb T}_{\text{\rm \DJ}, -}$ is defined by 
\begin{equation}
\label{eq:nuxi}
\xi_{\text{\DJ},\bm \nu}(\mu) = \sigma(\mu) \,  e^{\pi i \nu(\mu)}
\end{equation}
where $\nu(\gamma)$ was defined in Proposition \ref{prop:closed-Voros-coeffs} for $\gamma \in \Gamma$, 
and we extend it to an element in $\Hom(\Gamma_{\text{\DJ}}, \mathbb C)$ by setting $\nu(\beta)=0$ for $\beta \in \Gamma^{\ast}$. 
\end{rem}

\begin{rem}
When ${\nu_s} = {\nu_{p_{s_+}}} - {\nu_{p_{s_-}}} = 0$ for all $s \in P_{\rm ev}$, $\xi(\gamma)=\sigma(\gamma)$ for all $\gamma\in\Gamma$. This choice $\xi=\sigma\big|_{\Gamma}$ turns out to be exactly the value of the constant term for which Allegretti \cite{All19} solves the (non-doubled, coupled) BPS Riemann-Hilbert problem, also using Voros symbols.
\end{rem}
\begin{rem} \label{rem:use-of-almost-doubling}
Although a fractional power appears in the second line of  \eqref{eq:vorosexplicit}, we can verify that \eqref{eq:vorosexplicit} is indeed a meromorphic function of $\hbar$ since the pairing of $\gamma_{s_{\pm}} - \gamma_{s_{\mp}} \in \Gamma$ (only such cycles can give $-1$ as BPS invariant) with $\mu \in \Gamma_{\text{\rm \DJ}}$ is always even. This shows clearly why we had to consider the almost-doubled lattice rather than the (fully) doubled one. Otherwise, the powers of $\frac{1}{2}$ appearing in the presence of loop-type BPS cycles give rise to non-meromorphic $X^{\rm Vor}$, since the intersection pairing on the doubled lattice is not always even.
\end{rem}



Having shown that $X^{{\rm Vor}}_{\ell}$ satisfies the jump property, we may  conclude
 \begin{thm}
 \label{thm:voros-soln} 
 Let $(\Gamma,Z,\Omega)$ denote a BPS structure obtained from any spectral curve of hypergeometric type. Then $X^{{\rm Vor}}_\ell$ is a meromorphic solution to the BPS Riemann-Hilbert problem associated to the almost-doubled BPS structure $(\Gamma_{\text{\rm \DJ}},Z_{\rm \text{\rm \DJ}},\Omega_{\text{\rm \DJ}})$, with the constant term $\xi_{{\text{\rm \DJ}},\bm \nu}$ given by \eqref{eq:nuxi}.
  \end{thm}
  \begin{proof}
  We have already shown (RH1), and (RH2) follows from the explicit formula \eqref{eq:vorosexplicit}. The asymptotics at infinity follow from the explicit expression and the properties of the $\Lambda$ function (see, e.g. \cite[\S 3]{Bar}.
  \end{proof}
  


\subsection{Relation to the minimal and holomorphic solutions}
\label{sec:relationtominimal}
We note that while we have given a solution to the BPS Riemann-Hilbert problem, it does \emph{not} always agree with the ones given by \cite{Bar,Bri19} which we wrote down in \S\ref{sec:BPS}. However, we may relate them explicitly (for simplicity, we restrict ${\rm Re}\,\nu(\gamma_{\rm BPS})$), noting in each case that the difference is at worst multiplication by a simple meromorphic function:

 \begin{coro} \label{cor:vormin} 

 Fix a non-BPS ray $\ell$. For all BPS cycles $\gamma_{\rm BPS} \in \Gamma$ with $Z(\gamma_{\rm BPS})\in i\mathbb{H}_\ell$,
  suppose ${\rm Re}\,\nu(\gamma_{\rm BPS})\in (-1,1]$ whenever $\Omega(\gamma_{\rm BPS})\neq -1$, and ${\rm Re}\,\nu(\gamma_{\rm BPS})\in { (-2,0]}$ whenever $\Omega(\gamma_{\rm BPS})=-1$.
Then, in the half-plane $\mathbb{H}_\ell$, the Voros solution \eqref{eq:vorosexplicit} and the minimal solution \eqref{eq:Anna-solution} to the BPS Riemann-Hilbert problem associated to $(\Gamma_{\text{\rm\DJ}}, Z_{\text{\rm\DJ}}, \Omega_{\text{\rm\DJ}})$ with constant term $\xi_{{\text{\rm \DJ}}, \bm \nu}$ are related via
 \begin{align}
 \label{eq:vormin1}
  X^{{\rm Vor}}_{\ell,\gamma} &= \;\; \; 1 \, \cdot \, X^{{\rm min}}_{\ell,\gamma},  \qquad \qquad \text{$\gamma \in \Gamma$} \\[+.5em]
 \label{eq:vormin2} 
 X^{{\rm Vor}}_{\ell,\beta} & = \rho_{\ell, \beta} \cdot X^{{\rm min}}_{\ell,\beta}, \qquad \qquad \text{$\beta \in \Gamma^{*}$}
\end{align}
where $\rho_{\ell,\beta} : {\mathbb{H}_\ell} \to {\mathbb C}$ is meromorphic, and given explicitly by 
\begin{align}
\label{eq:vorminfactor}
\rho_{\ell, \beta}(\hbar)= \prod_{\substack{\gamma \in \Gamma \\  Z(\gamma) \in i{\mathbb H}_\ell\\\Omega(\gamma)=-1 }}
\left(1-\frac{\pi i \nu(\gamma)}{Z(\gamma)}\hbar\right)^{\Omega(\gamma) \, \frac{\langle \beta, \gamma \rangle}{2}}.
\end{align}
In particular, when $\nu_s=0$ for all $s \in P_{\rm ev}$, the Voros solution agrees with the minimal solution --- that is, $X^{{\rm Vor}}_\ell =X^{{\rm \min}}_\ell$ holds for any non-BPS $\ell$.
 \end{coro}
 \begin{rem}
 Note, the relation for additional values of ${\bm \nu}$ can be obtained by shifting the second argument in the $\Lambda$ function.
 \end{rem}
 \begin{proof}
 The first line of \eqref{eq:vormin1} is immediate. In order to prove the second line, we may restrict our attention to the Weber and Bessel cases thanks to the superposition structure \eqref{eq:voros-coeffs-general}. In view of \eqref{eq:Borel-sum-Weber-Voros-symbol}, we have the following in the Weber case:
\begin{equation}
X^{{\rm Vor}, {\rm Web}}_{\ell, \beta_{\infty}} = 
\mathcal{S}_\ell e^{V_{\beta_\infty}^{\rm Web}}=\Lambda \left( \frac{Z(\gamma)}{2 \pi i \hbar} 
, \frac{1-\nu(\gamma_{})}{2}\right)^{\Omega(\gamma_{}) \langle \beta_\infty, \gamma \rangle},
\end{equation}
where $\gamma = \gamma_{\infty_\pm}$ is chosen so that $Z(\gamma)\in i \mathbb{H}_\ell$. 
Since $\Omega(\gamma) = 1$, our assumption ensures that  ${\rm Im}\, (\pi i \nu(-\gamma)+\pi i)\in [0,2\pi)$, so that $\frac{\log\xi_{\nu_\infty}(-\gamma)}{2\pi i}= \frac{1-\nu(\gamma)}{2}$. Thus the Voros solution agrees with the minimal solution 
\begin{equation}
    X^{{\rm min},{\rm Web}}_{\ell,\beta_\infty} = \Lambda \left( \frac{Z(\gamma_{})}{2 \pi i \hbar} 
, \frac{\log(\xi_{\nu_\infty}(-\gamma))}{2\pi i }\right)^{\Omega(\gamma) \langle \beta_\infty, \gamma \rangle}.
\end{equation}

{

On the other hand, in the Bessel case, we may rewrite the expression for the Voros coefficient as follows:
\begin{align}
X^{\rm Vor, Bes}_{\ell, \beta_0} 
& = \Lambda \left( \frac{Z(\gamma)}{2 \pi i \hbar}
,1-\frac{\nu(\gamma)}{2} \right)^{\Omega(\gamma) \frac{\langle \beta_0, \gamma \rangle}{2}} \Lambda \left( \frac{Z(\gamma)}{2 \pi i \hbar}, -\frac{\nu(\gamma)}{2}\right)^{\Omega(\gamma) \frac{\langle \beta_0, \gamma \rangle}{2}} 
\notag \\[+.5em]
& = 
\Lambda \left( \frac{Z(\gamma)}{2 \pi i \hbar}, -\frac{\nu(\gamma)}{2} \right)^{\Omega(\gamma) {\langle \beta_0, \gamma \rangle}} \cdot \left( 1 - \frac{\pi i \nu(\gamma)}{Z(\gamma)} \hbar \right)^{\Omega(\gamma) \, \frac{\langle \beta_0, \gamma \rangle}{2}}. 
\label{eq:Bessel-Voros-alternative-expression}
\end{align}
where $\gamma = \gamma_{0_{\pm}} - \gamma_{0_{\mp}}$ is chosen to satisfy $Z(\gamma) \in i \, {\mathbb H}_{\ell}$, and we used the property 
$\Lambda(w,\eta+1) = \Lambda(w,\eta) \, (1 + \frac{\eta}{w})$
to obtain the second line (recall the pairing is even on the almost-doubled lattice, so there is no issue of branching). Under the assumption on $\nu(-\gamma)$, the first factor in \eqref{eq:Bessel-Voros-alternative-expression} agrees with 
the minimal solution 
\begin{equation}
X^{{\rm min}, {\rm Bes}}_{\ell,\beta_0} = \Lambda \left( \frac{Z(\gamma_{})}{2 \pi i \hbar} 
, \frac{\log(\xi_{\nu_0}(-\gamma))}{2\pi i }\right)^{\Omega(\gamma_{}) \langle \beta_0, \gamma \rangle}, 
\end{equation}
and hence, the second factor in \eqref{eq:Bessel-Voros-alternative-expression} is
$\rho_{\ell, \beta_0}$. By the evennness of the pairing on the almost-doubled lattice,
this is a meromorphic function of $\hbar$. 
}

The general expression follows from the superposition structure, formula \eqref{eq:voros-coeffs-general}. 
\end{proof}

{ 

\begin{coro}
\label{coro:vorminhol}

For any spectral curve of hypergeometric type, there exists a choice of parameter ${\bm \nu}_*$ so that 
{
\begin{equation}
\label{eq:xi-specialized}
\xi_{\text{\rm \DJ},{\bm \nu}_{\ast}}(\gamma) = 1
\end{equation} 
holds for all active $\gamma$}, and the Voros and holomorphic solutions for the associated BPS Riemann-Hilbert problem are related, for any non-BPS ray $\ell$, as:
\begin{align}
 X_{\ell, \gamma}^{\rm Vor {}}\,\Big|_{{\bm\nu} = {\bm\nu}_*} & =  \,\,1\,\cdot\, X_{\ell, \gamma}^{\rm  hol}, \qquad \gamma\in 
   \Gamma \\
 X_{\ell, \beta}^{\rm Vor {}}\,\Big|_{{\bm\nu} = {\bm\nu}_*} & = \varrho_\beta \cdot X_{\ell, \beta}^{\rm  hol}, \qquad \beta\in 
   \Gamma^\ast
\end{align}
where $\varrho_\beta=\varrho_\beta(\hbar)$ is either $1$ or a simple meromorphic function independent of $\ell$.
\end{coro}
\begin{proof}
We need only check at the level of asymptotic expansions. The asymptotic expansion of $X^{\rm hol}_{\ell, \beta}$ is given by \cite[\S 5.3]{Bri19}:
\begin{equation} \label{eq:asymptotic-log-holomorphic-solution}
    {\rm log} X_{\ell, \beta}^{\rm hol}  \sim  \sum_{k\geq1} \sum_{\substack{ \gamma\in\Gamma \\ Z(\gamma) \in i \, {\mathbb H}_\ell }} \,   \dfrac{\Omega(\gamma) \, \langle\beta,\gamma\rangle \, B_{k+1}}{k(k+1)} \,  \left(\dfrac{2 \pi i}{Z(\gamma)}\right)^k \hbar^k.
\end{equation}
Comparing this with equation \eqref{eq:voros-coeffs-general}, and recalling that $B_{k+1}(0)=B_{k+1}(1)=B_{k+1}$, we can proceed case-by-case. For the Weber and Whittaker cases there is only one term in the sum, so choosing $\nu_\infty=1$ (or $\nu_\infty=-1$) gives exact agreement, and for Bessel $\nu_0=0$ similarly gives
\begin{equation}
   \mathcal{B}_{k+1}({\gamma}) \, \Bigl|_{\nu_0=0} 
   = \frac{ B_{k+1}(0)+B_{k+1}(1)}{2} = B_{k+1}
\end{equation}
which also agrees after accounting for the intersection number. 
In the case of Kummer we may demand $\mathcal{B}_{k+1}(\gamma_{\rm BPS})= B_{k+1}$ whenever $\gamma_{\rm BPS}$ has $\Omega(\gamma_{\rm BPS})\neq -1$ which forces $\nu_0, \nu_\infty = 0 \text{ or } 1$ (but not both). Setting $\nu_0=0 , \nu_\infty=1$, we obtain the agreement.
{It is not hard to see that this choice satisfies \eqref{eq:xi-specialized}.}

The simplest case in which a discrepancy between $X^{\rm Vor}$ and $X^{\rm hol}$ appears is Legendre. For concreteness, let us write out \eqref{eq:voros-coeffs-general} in this case:
\begin{equation}
V^{\rm Leg}_{{\beta_\infty},k} = 
V^{\rm Leg}_{{\beta_\infty},k}(m_\infty, \nu_\infty) =
\dfrac{1}{k(k+1)}\left( \frac{4 B_{k+1}\left(\frac{\nu_\infty+1}{2}\right)}{m_\infty^k} - \frac{B_{k+1}(\nu_\infty)+B_{k+1}(1 + \nu_\infty)}{(2m_\infty)^k}\right)
\end{equation}
Setting $\nu_\infty =1$ we obtain 
\begin{align}
    V^{\rm Leg}_{{\beta_\infty},k}({m_\infty},1)
    & =\dfrac{1}{k(k+1)} 
    \left( \frac{ 4 B_{k+1} }{ m_\infty^k } - \frac{B_{k+1}+B_{k+1}(2)}{(2m_\infty)^k} \right) \notag \\
    & =\dfrac{1}{k(k+1)}\left( \frac{4B_{k+1}}{m_\infty^k} - \frac{2B_{k+1}}{(2m_\infty)^k} \right) - \frac{1}{k (2m_\infty)^k} ,
\end{align}
where we used the identity $B_{k}(1) = B_k$ and  $B_k(\nu+1)=B_k(\nu)+k\nu^k$. 
Through the comparison to \eqref{eq:asymptotic-log-holomorphic-solution}, 
we can see that the Borel sum of $e^{V^{\rm Leg}_{\beta_\infty}}$ along $\ell$ differs from $X_{\ell, \beta_\infty}^{\rm hol}$ by an additional term $-\log(1-\frac{\hbar}{2m_\infty})$, which is independent of $\ell$ since the formal series itself is convergent. Its exponential $(1-\frac{\hbar}{2m_\infty})^{-1}$ is meromorphic function of $\hbar$ and gives the factor ${\varrho}_{\beta_\infty}={\varrho}^{\rm Leg}_{\beta_\infty}$.

The degenerate Gauss and Gauss hypergeometric cases follow similarly, first choosing $\nu_s=0 \text{ or } 1$ in order to obtain agreement from terms with $\Omega(\gamma)\neq -1$, and computing the resulting discrepancy appearing for terms with $\Omega(\gamma)=-1$. In both cases, we may make an arbitrary choice of one of the second order poles and set the corresponding $\nu_s=1$, with the remaining $\nu_{\tilde{s}}=0$. Thus we obtain a discrepancy of 
\begin{equation}\varrho^{\rm HG,dHG}_{\beta_s}=
    \left(1-\frac{\hbar}{2m_s}\right)^{-1}    
\end{equation}
For the other poles $\tilde{s}$, we have $\varrho_{\beta_{\tilde{s}}} = 1$. 
{Again, we can verify that \eqref{eq:xi-specialized} is satisfied even in the cases where there is a non-trivial discrepancy $\varrho_\beta$. 
}
\end{proof}

}

\subsection{BPS $\tau$-function and the Voros potential}
Having established the Voros solution to the almost-doubled BPS Riemann-Hilbert problem, we can ask for a corresponding BPS $\tau$-function. Such a function exists, and it turns out there is a very natural interpretation on the TR side in terms of the Voros potential:


\begin{thm}\label{thm:taupotential}
For any non-BPS ray $\ell$, the Borel-resummed Voros potential provides a BPS $\tau$-function for the Voros solution as:
\begin{equation} \label{eq:def-tau-Voros}
\tau_{{\rm BPS}, \ell}^{\rm Vor} = 
c_\ell \, 
\mathcal{S}_{\ell} \, e^{ - \partial_\hbar \phi}.
\end{equation}
Here, the prefactor $c_\ell =  c_\ell(\hbar)$ is independent of ${\bm m}$ and chosen explicitly as follows: 
\begin{equation} \label{eq:c-ell}
    c_\ell(\hbar) = \prod_{\substack{ \gamma \in \Gamma \\ 
    Z(\gamma) \in i \, {\mathbb H}_\ell }} \hbar^{\frac{1}{2} {\mathcal B}_2(\gamma) \Omega(\gamma)}. 
\end{equation}
\end{thm}
\begin{proof}
Using the definition of the $\tau$-function for the almost-doubled problem, it is not hard to see the condition for being a $\tau$-function is exactly the same as the condition (rightmost in \eqref{eq:tau-concrete}) for being a potential for  ${-} \partial_\hbar\log Y_{\cdot,\ell}$, where $Y_{\cdot,\ell}$ is the untwisted part of the solution and $\cdot$ is valued in $\Gamma^*$. For the Voros solution, this is the equation \eqref{eq:vorpot} for the Voros potential {(up to the sign)}. Then from the explicit formulas, it is easy to check the equation holds at the level of asymptotic expansions, and since both sides are Borel summable, the relation holds analytically. 
The constant prefactor $c_\ell$ is chosen so that $\tau_{\rm BPS}^{\rm Vor}$ is invariant under the rescaling $(m_s, \hbar) \mapsto (\lambda m_s, \lambda \hbar)$ for any $\lambda \in {\mathbb C}^\ast$.
\end{proof}


As in the previous section, we can also obtain the relation between $\tau^{\rm Vor}_{\rm BPS}$ and $\tau^{\rm min}_{\rm BPS}$ (which together with the explicit formula for $\tau^{\rm min}_{\rm BPS}$ in Proposition \ref{prop:annatau} may be used to compute the Voros potential explicitly). We have:


{
\begin{coro}
\label{thm:maintheorem-1}

Fix a non-BPS ray $\ell$.
Under the same assumptions on ${\bm \nu}$ as Corollary \ref{cor:vormin}, the BPS $\tau$ function $\tau^{\rm Vor}_{{\rm BPS}}$ is related to the $\tau$-function $\tau_{\rm BPS}^{\rm min}$ \eqref{eq:mintau} for the minimal solution by
\begin{equation}
\tau_{{\rm BPS}, \ell}^{{\rm Vor}} = \kappa_\ell\cdot \tau_{{\rm BPS}, \ell}^{{\rm min}}
\end{equation}
where $\kappa_\ell$ is a simple explicit function given by 
\begin{equation} \label{eq:expression-of-kappa-ell}
\kappa_\ell = \kappa_\ell(\hbar) 
= \prod_{\substack{ \gamma \in \Gamma \\ Z(\gamma) \in i {\mathbb H}_\ell \\ \Omega(\gamma) = -1 }}
\left( \frac{Z(\gamma)}{2 \pi i \hbar} -  \frac{\nu(\gamma)}{2}  \right)^{ \frac{\nu(\gamma) \, \Omega(\gamma) }{4}}.
\end{equation}

\end{coro}
\begin{proof}
From equation \eqref{eq:tau-concrete} and \eqref{eq:vormin2}, $\kappa_\ell$ must satisfy 
\begin{equation}
\frac{\partial}{\partial{m_s}} {\rm log} \kappa_\ell =
    \frac{\partial}{\partial{m_s}}\left({\rm log}\tau^{\rm Vor}_{\rm BPS, \ell}-{\rm log}\tau^{\rm min}_{\rm BPS, \ell}\right) = - \frac{\partial}{\partial \hbar} \log \rho_{\ell, \beta_s}. 
\end{equation}
Using \eqref{eq:vorminfactor}, 
we can solve the relation in an elementary computation and obtain $\kappa_\ell$ up to multiplicative constant which is independent of $m_s$. 
To determine the constant, let us look at the leading terms of asymptotic expansion of the BPS $\tau$-functions when $\hbar \to 0$ in ${\mathbb H}_\ell$: 
\begin{align}
\label{eq:asymptotic-tau-min}
\log \tau_{\rm BPS, \ell}^{\rm min} ~ \sim & ~ 
- \frac{1}{2} \sum_{\substack{ \gamma \in \Gamma \\ Z(\gamma) \in i {\mathbb H}_\ell }} \Omega(\gamma) \, B_2\Bigl( \frac{\log \xi_{\bm \nu}(- \gamma)}{2 \pi i} \Bigr) \, \log \left( \frac{Z(\gamma)}{2 \pi i \hbar} \right)
\\
\label{eq:asymptotic-tau-Voros}
\log \tau_{\rm BPS, \ell}^{\rm Vor} ~ \sim & ~ 
- \frac{1}{2} \sum_{\substack{ \gamma \in \Gamma \\ Z(\gamma) \in i {\mathbb H}_\ell }} \Omega(\gamma) \, {\mathcal B}_2(\gamma) \, \log \left( \frac{Z(\gamma)}{2 \pi i \hbar} \right)
\end{align}
Here \eqref{eq:asymptotic-tau-min} follows from the asymptotic property \eqref{eq:asymptotic-upsilon} of $\Upsilon$-function, and \eqref{eq:asymptotic-tau-Voros} is given by the leading term of $\hbar$-derivative of the Voros potential (together with $\log \, {c}_\ell$; see Theorem \ref{thm:taupotential}).
As we have seen in the proof of Corollary \ref{cor:vormin}, 
\begin{equation}
\frac{\log \xi_{\bm \nu}(- \gamma)}{2 \pi i} = 
\begin{cases}
\displaystyle  \frac{1 - \nu(\gamma)}{2} \quad & \Omega(\gamma) \ne -1,  \\[+1.em]
\displaystyle - \frac{\nu(\gamma)}{2} \quad & \Omega(\gamma) = -1
\end{cases}
\end{equation}
holds under the same assumptions on ${\bm \nu}$. 
Therefore, using the properties 
$B_2(x) = B_2(1 - x)$ and $B_2(1+x) - B_2(x) = 2 x$
of the Bernoulli polynomial, we have
\begin{equation}
{\mathcal B}_2(\gamma) - B_2\left( \frac{\log \xi_{\bm \nu}(- \gamma)}{2 \pi i} \right) = 
\begin{cases}
\displaystyle 0 \quad & \Omega(\gamma) \ne -1, \\[+.5em] 
\displaystyle -\frac{\nu(\gamma)}{2} \quad & \Omega(\gamma) = -1.
\end{cases}
\end{equation}
This fixes the overall constant factor in $\kappa_\ell$, and 
thus we have \eqref{eq:expression-of-kappa-ell}.
\end{proof}

\noindent In particular, we see that the $\tau$-functions do not agree for the Voros and minimal solutions, but differ only by a very simple factor coming from BPS cycles with $\Omega(\gamma_{\rm BPS})=-1$.

}

Finally, thanks to Corollary \ref{coro:vorminhol} we may relate the original BPS $\tau$-function of Bridgeland $\tau^{\rm hol}_{\rm BPS}$ to $\tau^{\rm Vor}_{\rm BPS}$ at a special parameter value:

\begin{coro}
\label{thm:maintheorem-2}
For any of the families of spectral curves of hypergeometric type, the BPS $\tau$-functions $\tau^{{\rm Vor}}_{{\rm BPS},{\ell}}$ for the Voros solution and $\tau^{\rm hol}_{\rm BPS}$ for the holomorphic solution are related, as functions of $\hbar$ defined on $\mathbb{H}_\ell$ for any non-BPS ray $\ell \in \mathbb{C}^\ast$, by
\begin{equation} \label{eq:tau-Voros-and-tau-holomorphic}
\tau_{{\rm BPS}, \ell}^{{\rm Vor}} \,\Big|_{{\bm \nu} = {\bm \nu}_*} 
=\varkappa \cdot 
{ 
\tau_{{\rm BPS}, \ell}^{{\rm hol}},
}
\end{equation}
where $\varkappa$ is either 1 or a simple explicit function independent of $\ell$.
\end{coro}
\begin{proof}
Replacing $\rho$ with $\varrho$ in the previous proof, we obtain $\varkappa = 1$ for the Weber, Whittaker, Bessel, and Kummer cases.
For Legendre, degenerate Gauss, and Gauss, we obtain
\begin{equation}
{ 
\varkappa(\hbar)=\Big(\frac{2 m_s}{\hbar} - 1 \Big)^{- \frac{1}{2}}
}
\end{equation}
where $s\in P_{\rm ev}$ is a pole chosen as in the proof of Corollary \ref{coro:vorminhol}.


\end{proof}

\subsection{Borel-resummed TR partition function as a BPS $\tau$-function}  
\label{section:TR-tau-vs-BPS-tau-hol}

Let us finally show that the Borel-resummed topological recursion partition function essentially agrees with the BPS $\tau$-function which we have constructed so far --- this provides the ``Borel-resummed version'' of our previous result in \cite{IK20}.

\begin{thm}
\label{thm:tauhol}
For any spectral curve of hypergeometric type, 
let $F_{\rm TR}$ and $Z_{\rm TR}$ denote the topological recursion free energy and partition function. Then, 
$F_{\rm TR}$ and $Z_{\rm TR}$ are Borel summable along any non-BPS ray $\ell$ for the corresponding BPS structure. 
In particular, if we take the following normalization
\begin{equation} 
\label{eq:F1-normalization-hbar}
F_{1} = 
- \frac{1}{12} \sum_{\substack{ \gamma \in \Gamma \\ Z(\gamma)\in  i \mathbb{H}_{\ell}}} 
\Omega(\gamma)\, \log \left( \frac{Z(\gamma)}{2 \pi i \hbar} 
\right)
\end{equation}
of $F_1$, then the Borel sum 
\begin{equation}
{\tau}^{\mathsmaller{>0}}_{{\rm TR},\ell} = 
{\mathcal S}_{\ell}  \exp \left( {\sum_{g>0} \hbar^{2g-2} F_g} \right)
\end{equation}
agrees with the BPS $\tau$-function $\tau_{\rm BPS, \ell}^{\rm hol}$:
\begin{equation} \label{eq:main-equality-TR-vs-BPS}
  \tau_{\rm BPS, \ell}^{\rm hol} = {\tau}^{\mathsmaller{>0}}_{{\rm TR},\ell}. 
\end{equation}
\end{thm}

\begin{proof}
Using the asymptotic property \eqref{eq:asymptotic-upsilon} of $\Upsilon$-function, we can verify that  $\tau_{\rm BPS, \ell}^{\rm hol}$ has the following asymptotic expansion when $\hbar \to 0$ in ${\mathbb H}_\ell$:
\begin{align}
 \log \tau_{\rm BPS, \ell}^{\rm hol} 
& = \sum_{\substack{\gamma \in \Gamma \\ Z(\gamma) \in i {\mathbb H}_\ell}} \Omega(\gamma) \log \Upsilon\left( \frac{Z(\gamma)}{2 \pi i \hbar} \right)  \notag \\
& \sim  - \frac{B_2}{2}  \sum_{\substack{\gamma \in \Gamma \\ Z(\gamma) \in i {\mathbb H}_\ell}} \Omega(\gamma) \log \left( \frac{Z(\gamma)}{2 \pi i \hbar} \right) 
+  \sum_{g \ge 2} \hbar^{2g-2} \frac{B_{2g}}{2g(2g-2)} \sum_{\substack{\gamma \in \Gamma \\ Z(\gamma) \in i {\mathbb H}_\ell}} \Omega(\gamma) \left( \frac{2\pi i}{Z(\gamma)} \right)^{2g-2}. 
\end{align}
We can verify that the right hand side is nothing but 
$\sum_{g \ge 1} \hbar^{2g-2} F_g$ by using Theorem \ref{thm:Borel-sum-free-energy}. Since the asymptotic expansion of the $\Upsilon$-function is valid in a sufficiently large sector, this asymptotic property ensures that $F_{\rm TR}$ and its exponential $Z_{\rm TR}$ are Borel summable along $\ell$, and the relation \eqref{eq:main-equality-TR-vs-BPS} holds tautologically. 
\end{proof}

Together with Corollary \ref{thm:maintheorem-1}, the above formula tells us that the TR partition function is realized as the specialization of the BPS $\tau$-function associated with the Voros solution.

\begin{rem}
\label{rem:differenceeq} We note that the specialization to the special value of ${\bm\nu}$ is necessary, since the BPS $\tau$-function in general clearly depends on ${\bm \nu}$, but $\tau_{\rm TR}$ depends only on free energies (and thus not on the quantization parameter ${\bm \nu}$ ). We may offer a middle ground as follows. In \cite{IKoT-I,IKoT-II}, it was shown that the Voros coefficient can be recovered as the so-called ``difference value'' of the free energies. To simplify the notation, let us focus only on the Weber case, and refer to \cite{IKoT-II} for the others. It was found that
\begin{align}
\label{eq:diffeq}
& V^{\rm Web}(m_\infty, \nu_\infty,\hbar) \notag \\ 
& \quad = F^{\rm Web}\left( m_\infty - \frac{\nu_\infty \hbar}{2} + \frac{\hbar}{2} , \hbar \right) - F^{\rm Web}\left( m_\infty - \frac{\nu_\infty \hbar}{2} -  \frac{\hbar}{2} , \hbar \right) - \frac{1}{\hbar} \dfrac{\partial F_0}{\partial m_\infty} + \frac{1}{2} \dfrac{\partial^2F_0}{\partial m_\infty^2}.
\end{align}
 Essentially, this relation is due to the variational formula (\cite{EO})
 \begin{equation}
     \frac{\partial^n F^{\rm Web}_g}{\partial m_\infty^n}= \int_{0}^{\infty} \cdots \int_{0}^{\infty} W^{\rm Web}_{g,n}(z_1, \dots, z_n)
 \end{equation} 
 (see \cite{IKoT-II}). The last two terms in \eqref{eq:diffeq} may be removed by a suitable redefinition of the Voros coefficient. Thus, at the level of asymptotic expansions, the difference value of $\tau_{\rm BPS}^{\rm hol}$ gives the Voros coefficient, i.e. the solutions to the Riemann-Hilbert problem for arbitrary choices of ${\bm \nu}$ are encoded in the topological recursion $\tau$-function $\tau_{\rm TR}$ (the BPS $\tau$-function at ${\bm \nu}={\bm \nu}_*$). This should be related to the observation of \cite{BBS}, who obtained a similar difference equation for the minimal solution (see also \cite{Alim1, Alim2, AS21} where closely related difference equations are discussed). Since it nonetheless encodes the solution for the full family of ${\bm \nu}$, one may consider this as an alternative way to define $\tau_{\rm BPS}$, depending on which properties (geometric characterization, independence of ${\bm \nu}$) one desires the most.
\end{rem}
\begin{rem}{   
By looking at the asymptotic expansion of the both sides of \eqref{eq:def-tau-Voros} evaluated at ${\bm \nu} = {\bm \nu}_{\ast}$, we have
\begin{equation} \label{eq:VP-vs-TRFE}
\bigl( - \partial_\hbar \phi + \log c_{\ell} \bigr)\Big|_{{\bm \nu} = {\bm \nu}_\ast} = 
\sum_{g \ge 1} \hbar^{2g-2} F_g + \log {\varkappa}, 
\end{equation}
where we use the same normalization \eqref{eq:F1-normalization-hbar} of $F_1$ on the right-hand side.  
Thus, we see there are two different relations between the solution to the BPS Riemann-Hilbert problem (or the Voros coefficients) and the BPS $\tau$-function (or the TR partition function); that is, the difference equation discussed \eqref{eq:diffeq} and the relation \eqref{eq:VP-vs-TRFE} with the Voros potential. 
From the view point of special functions, this reflects the fact that the $\Gamma$ function and Barnes $G$-functions are related in two different ways (i.e., difference relation and integro-differential relation) as:
\begin{align*}
\log G(w+1) 
& = \log G(w) + \log \Gamma(w)  \\
& = \frac{w(1-w)}{2} + \frac{w}{2} \log 2 \pi + w \log \Gamma(w) - \int^{w}_{0} \log \Gamma(z) dz.
\end{align*}

}
\end{rem}

\subsection{Higher degree spectral curves}
\label{sec:conjectures}

We have so far only discussed variations of BPS structure arising from quadratic differentials. However, due to the work \cite{GMN12} by Gaiotto-Moore-Neitzke, it is natural to expect that a tuple of higher differentials (or a higher degree spectral curve) also defines a BPS structure. Although many examples have been studied in e.g. \cite{GMN12}, there is still no mathematically rigorous description of the corresponding BPS structures in these cases, so such a picture remains conjectural.

On the other hand, TR is applicable to higher degree spectral curves \cite{BHLMR-12, BE-12}, and its relationship to WKB analysis is also discussed in \cite{BE-16} (under a certain admissibility assumption on spectral curves). 
The Borel summability and resurgence structure for the WKB solutions / Voros symbols of the associated quantum curve is not fully established, but it is natural to expect that they are controlled by an (appropriately generalized) spectral network (Stokes graph) as in the rank 2 cases (c.f., \cite{BNR, AKT94, AKT98}). 

Thus it is reasonable to ask if the results in this paper generalize to higher degree curves. Rather than formulate a precise conjecture, we will simply recall two higher degree examples considered numerically in Part I, and observe that (thanks to results of Y.M. Takei \cite{YM-Takei20}) they behave as expected based on the results we have obtained so far. 

{\subsubsection*{\bf Example: $(1,4)$ curve}
Let us consider the spectral curve arising as the classical limit of the 3rd order hypergeometric differential equation of type $(1,4)$, studied in \cite{Okamoto-Kimura, Hirose}, which is explicitly given as follows:
\begin{equation}
    \Sigma_{(1,4)} ~:~ 3y^3 + 2 t y^2 + x y - m_\infty = 0.
    \label{eq:14curve}
\end{equation}
Here $m_\infty$ is a parameter assumed to be non-zero, and we regard $t$ as a constant. 
The curve $\Sigma_{(1,4)}$ is a genus $0$ curve with two punctures at $\infty_\pm$, and we can verify that $\Res_{x=\infty_\pm} y dx = \pm m_\infty$.

A crude numerical experiment performed in our previous work \cite{IK20} suggested a single degeneration occurs, the so-called ``three-string web'' also observed in \cite{GMN12}. According to their result, the associated BPS cycle is simply the residue cycles $\gamma_{\infty_\pm}$, and the BPS index is assigned so that  $\Omega(\gamma_{\infty_\pm}) = 1$.

According to \cite{IKo} and \cite[Theorem 4.6]{YM-Takei20}, the explicit expression for the Voros coefficient of the quantum curve for the path from $\infty_{-}$ to $\infty_+$ is given as follows: 
\begin{equation}
    V_{\beta_\infty} = \sum_{k=1}^{\infty}{\dfrac{B_{k+1}(\nu_\infty)}{k(k+1)} \left(\dfrac{\hbar}{m_\infty}\right)^k}.
\end{equation}

Since the formula is identical to that of the Weber curve, we have an explicit description via \eqref{eq:Borel-sum-Weber-Voros-symbol} of the jump property after taking the Borel sum. Taking an appropriate lattice, one may check that this Voros symbol solves the corresponding almost-doubled BPS Riemann-Hilbert problem.

\subsubsection*{\bf Example: $(2,3)$ curve}

We consider the spectral curve arising as the classical limit of the 3rd order hypergeometric differential equation of type $(2,3)$, which was also studied in \cite{Okamoto-Kimura}. The curve is explicitly given as follows:
\begin{equation}
    \Sigma_{(2,3)} ~:~ 4y^3 -2 x y^2 + 2 m_\infty y -t = 0.
    \label{eq:23curve}
\end{equation}
Again, it is a genus $0$ curve with two punctures at $\infty_\pm$, and $\Res_{x=\infty_{\pm}} y \, dx = \pm m_\infty$.

Numerical experiment \cite{IK20} together with the general rules for spectral networks and BPS states predicted from physics (e.g. \cite{GMN12}) suggests that $\Omega(\gamma)=0$ for all classes  $\gamma$. 
\cite[Theorem 4.6]{YM-Takei20} computed the Voros coefficient explicitly, which is simply
\begin{equation}
    V_{\beta_\infty}=0,
\end{equation}
After taking the Borel sum, the Voros symbol does not jump, and provides a solution in accordance with the vanishing of the BPS indices $\Omega(\gamma)$.

}



\appendix
\section{{Explicit expression for Voros coefficients}}
\label{sec:explicit-formulas}
To make formula \eqref{eq:voros-coeffs-general} for the Voros coefficient more concrete, and to make the BPS structure in each example more evident, we give their explicit expressions that were obtained in \cite[Theorem 3.1]{IKoT-II}. The Voros coefficients  arising from our examples are explicitly given by 
\begin{equation}
V_{\beta_s}(\hbar)
= \sum_{k = 1}^{\infty} \hbar^{k} V_{\beta_s, k}
\quad (s \in P_{\rm ev}),
\end{equation}
where the coefficients $V_{\beta_s, k}$ for $k \ge 1$ are given as follows (several of these expressions or special cases thereof were obtained in \cite{SS, Takei08, KoT11,  ATT, ATT2, AIT}).\medskip

\begin{itemize}
\item[$\bullet$]  
For the quantum Gauss curve:
\begin{align*}
& V^{\rm HG}_{\beta_0, k} 
 = 
\frac{1}{k (k + 1)}
\left\{
\frac{B_{k + 1} \bigl( \frac{1 + \nu_0 + \nu_1 + \nu_{\infty}}{2} \bigr)}
{(m_0 + m_1 + m_{\infty})^k}
+ \frac{B_{k + 1} \bigl( \frac{1 + \nu_0 - \nu_1 + \nu_{\infty}}{2} \bigr)}
{(m_0 - m_1 + m_{\infty})^k}
\right. \nonumber \\ 
& \hspace{+1.em} \left.
+ \frac{B_{k + 1} \bigl( \frac{1 + \nu_0 + \nu_1 - \nu_{\infty}}{2} \bigr) }
{(m_0 + m_1 - m_{\infty})^k}
+ \frac{B_{k + 1} \bigl( \frac{1 + \nu_0 - \nu_1 - \nu_{\infty}}{2} \bigr)}
{(m_0 - m_1 - m_{\infty})^k} 
- \frac{B_{k + 1}(\nu_0) + B_{k + 1}(1+\nu_0)}{(2m_0)^k} 
\right\}, 
\\[+.5em]
& V^{\rm HG}_{\beta_1, k} 
 = 
\frac{1}{k (k + 1)}
\left\{
\frac{B_{k + 1} \bigl( \frac{1 + \nu_0 + \nu_1 + \nu_{\infty}}{2} \bigr)}
{(m_0 + m_1 + m_{\infty})^k}
- \frac{B_{k + 1} \bigl( \frac{1 + \nu_0 - \nu_1 + \nu_{\infty}}{2} \bigr)}
{(m_0 - m_1 + m_{\infty})^k}
\right. \nonumber \\ 
& \hspace{+1.em} \left.
+ \frac{B_{k + 1} \bigl( \frac{1 + \nu_0 + \nu_1 - \nu_{\infty}}{2} \bigr) }
{(m_0 + m_1 - m_{\infty})^k}
- \frac{B_{k + 1} \bigl( \frac{1 + \nu_0 - \nu_1 - \nu_{\infty}}{2} \bigr)}
{(m_0 - m_1 - m_{\infty})^k} 
- \frac{B_{k + 1}(\nu_1) + B_{k + 1}(1+\nu_1)}{(2m_1)^k} 
\right\},
\\[+.5em]
& V^{\rm HG}_{\beta_\infty, k} 
 = 
\frac{1}{k (k + 1)}
\left\{
\frac{B_{k + 1} \bigl( \frac{1 + \nu_0 + \nu_1 + \nu_{\infty}}{2} \bigr)}
{(m_0 + m_1 + m_{\infty})^k}
+ \frac{B_{k + 1} \bigl( \frac{1 + \nu_0 - \nu_1 + \nu_{\infty}}{2} \bigr)}
{(m_0 - m_1 + m_{\infty})^k}
\right. \nonumber \\ 
& \hspace{+1.em} \left.
- \frac{B_{k + 1} \bigl( \frac{1 + \nu_0 + \nu_1 - \nu_{\infty}}{2} \bigr) }
{(m_0 + m_1 - m_{\infty})^k}
- \frac{B_{k + 1} \bigl( \frac{1 + \nu_0 - \nu_1 - \nu_{\infty}}{2} \bigr)}
{(m_0 - m_1 - m_{\infty})^k} 
- \frac{B_{k + 1}(\nu_\infty) + B_{k + 1}(1+\nu_\infty)}{(2m_\infty)^k} 
\right\}.
\end{align*}
We note that these Voros coefficients are related by permutions of the parameters: 
\begin{align*}
V^{\rm HG}_{\beta_1, k} 
& =
V^{\rm HG}_{\beta_0, k} \bigl|_{
(m_0, m_1, m_\infty, \nu_0, \nu_1, \nu_\infty) = 
(m_1, m_0, m_\infty, \nu_1, \nu_0, \nu_\infty)}, 
\\
V^{\rm HG}_{\beta_\infty, k}
& =
V^{\rm HG}_{\beta_0, k} \bigl|_{
(m_0, m_1, m_\infty, \nu_0, \nu_1, \nu_\infty) = 
(m_\infty, m_1, m_0, \nu_\infty, \nu_1, \nu_0)}. 
\end{align*}

\bigskip
\item[$\bullet$]  
For the quantum degenerate Gauss curve:
\begin{align*}
V^{\rm dHG}_{\beta_1,k} 
& = \frac{1}{k(k+1)}\left\{
\frac{2 B_{k + 1} \bigl( \frac{1 +\nu_1 + \nu_{\infty}}{2} \bigr)}
{(m_1 + m_{\infty})^k}
+ \frac{2 B_{k + 1} \bigl( \frac{1 +\nu_1 - \nu_{\infty}}{2} \bigr)}
{(m_1 - m_\infty)^k}
- \frac{B_{k + 1}(\nu_1) + B_{k + 1}(1+\nu_1)}{(2m_1)^k} 
\right\}, \\[+.5em]
V^{\rm dHG}_{\beta_\infty,k} 
& = \frac{1}{k(k+1)}\left\{
\frac{2 B_{k + 1} \bigl( \frac{1 +\nu_1 + \nu_{\infty}}{2} \bigr)}
{(m_1 + m_{\infty})^k}
- \frac{2 B_{k + 1} \bigl( \frac{1 +\nu_1 - \nu_{\infty}}{2} \bigr)}
{(m_1 - m_\infty)^k}
- \frac{B_{k + 1}(\nu_\infty) + B_{k + 1}(1+\nu_\infty)}{(2m_\infty)^k} 
\right\}.
\end{align*}
We note that these Voros coefficients are related by permutions of the parameters:
\begin{align*}
V^{\rm dHG}_{\beta_\infty, k} =
V^{\rm dHG}_{\beta_1, k} \bigl|_{
(m_1, m_\infty,  \nu_1, \nu_\infty) = 
(m_\infty, m_1, \nu_\infty, \nu_1)}. 
\end{align*}

\bigskip
\item[$\bullet$]  
For the quantum Kummer curve: 
\begin{align*}
		V^{\rm Kum}_{\beta_0, k}
		&= \frac{1}{k(k+1)} 
			\bigg\{ \frac{B_{k+1}\bigl( \frac{1 + \nu_0 + \nu_{\infty}}{2} \bigr)}
			{(m_0 + m_{\infty})^{k}} 
		+ \frac{B_{k+1}\bigl( \frac{1 +\nu_0 - \nu_{\infty}}{2} \bigr)}
		{(m_0 - m_{\infty})^{k}} 
		- \frac{B_{k+1}(\nu_{0}) + B_{k+1}(1 + \nu_{0})}
		{{(2 m_0)}^{k}} \bigg\}, \notag \\[+.5em]
		V^{\rm Kum}_{\beta_\infty, k} 
		&=  \frac{1}{k(k+1)} 
			\bigg\{ \frac{B_{k+1}\bigl( \frac{1 +\nu_0 + \nu_{\infty}}{2} \bigr)}
			{(m_0 + m_{\infty})^{k}} 
			- \frac{B_{k+1}\bigl( \frac{1 +\nu_0 - \nu_{\infty}}{2} \bigr)}
			{(m_0 - m_{\infty})^{k}} \bigg\}. 
\end{align*}

\bigskip
\item[$\bullet$]  
For the quantum Legendre curve: 
\[
V^{\rm Leg}_{\beta_\infty, k} =  
\frac{1}{k(k+1)} \left\{  
\frac{4B_{k+1}\bigl( \frac{1 + \nu_\infty}{2} \bigr)}{m_\infty^k} 
- \frac{B_{k+1}(\nu_\infty) + B_{k+1}(1 + \nu_\infty)}{(2m_\infty)^{k}}
\right\}.
\]

\bigskip
\item[$\bullet$]  
For the quantum Bessel curve:
\[
V^{\rm Bes}_{\beta_0, k} =  
- \frac{B_{k+1}(\nu_0) + B_{k+1}(1 + \nu_0)}{k(k+1) (2m_0)^{k}}. 
\]


\bigskip
\item[$\bullet$]  
For quantum Whittaker curve: 
\[
V^{\rm Whi}_{\beta_\infty, k} =  
\frac{2 B_{k+1} \bigl( \frac{1 + \nu_\infty}{2})}{k(k+1) \, m_{\infty}^{k}}.
\]

\bigskip
\item[$\bullet$]  
For quantum Weber curve:
\[
V^{\rm Web}_{\beta_\infty, k} = 
\frac{B_{k + 1}\big( \frac{1 + \nu_\infty}{2} \big)}{k (k + 1) \, m_\infty^{k}}. 
\]

\bigskip
\item[$\bullet$]  
The Voros coefficients are not defined for the quantum Airy curve and quantum degenerate Bessel curve, since they have trivial BPS structure (i.e., $\Gamma = 0$).
\end{itemize}

\addtocontents{toc}{\protect\setcounter{tocdepth}{0}}
\section*{Declarations}
\subsection*{Data availability}
Data sharing is not applicable to this article as no datasets were generated or analyzed during the current study.
\subsection*{Competing interests}
The present work was funded by various organizations outlined in the Acknowledgements section above.

\addtocontents{toc}{\protect\setcounter{tocdepth}{1}}

\end{document}